\documentclass[prodmode,acmtocl]{acmsmall}
\let\orgsetcounter\setcounter 

\usepackage{etex}
\usepackage{makeidx}  
\usepackage{amsmath}
\usepackage{amsfonts}
\usepackage{amssymb}
\usepackage{stmaryrd}
\usepackage{graphics}
\usepackage{graphicx}
\usepackage{palatino}
\usepackage{lmodern}
\usepackage[latin1]{inputenc}
\usepackage[T1]{fontenc}
\usepackage{listings, framed}
\usepackage{subfigure}
\usepackage{color}
\usepackage{galois}
\usepackage{amsbsy}
\usepackage{latexsym}
\usepackage{url}
\usepackage{ifthen}
\usepackage{epsfig}
\usepackage{pstricks}
\usepackage{xspace}
\usepackage{listings}
\usepackage[all]{xy}
\usepackage{tikz}
\usepackage{cancel}
\usepackage{fixmath}

\newcommand{\MMM}{\mathbb{M}}
\newcommand{\LLL}{\mathbb{L}}
\newcommand{\VV}{\mathsf{V}}
\newcommand{\PP}{\mathsf{P}}

\newcommand{\AAA}{\mathsf{A}}
\newcommand{\SSS}{\mathsf{S}}
\newcommand{\XX}{\mathsf{L}}

\newcommand{\ccompa}[1]{\stackrel{\mbox{\scriptsize $\alpha_{#1}$}}{\rightarrow}}
\newcommand{\wkth}[1]{\mbox{\scriptsize $\stackrel{\mbox{\scriptsize $#1$}}{\subseteq}$}}

\newcommand{\absstcslcg}{\emph{abstract static slicing}}

\newcommand{\absdynslcg}{\emph{abstract dynamic slicing}}

\newcommand{\absconslcg}{\emph{abstract conditioned slicing}}

\newcommand{\stm}{\mbox{\sl Stm}}
\newcommand{\caX}{\cX}
\newenvironment{nomedef}
{\rm\bfseries}
{}
\newcommand{\defthename}[1]{\begin{nomedef}(#1)\quad\!\!\!\end{nomedef}}

\lstdefinelanguage{pseudocode}{
  morekeywords={data,do,else,for,to,downto,let,call,foreach,function,procedure,if,return,returns,takes,then,until,while,repeat,read,skip},
  sensitive=false,
  morecomment=[l]{//},
  morestring=[b]'',}
\lstnewenvironment{pseudocode}{
  \lstset{language=pseudocode}
  \lstset{commentstyle=\textit}
  \lstset{frame=tb}
  \lstset{linewidth=120mm}
  \lstset{xleftmargin=3mm}
  \lstset{frameround=ffff}
  \lstset{framextopmargin=1mm}
  \lstset{aboveskip=3mm}
  \lstset{belowskip=4mm}
  \lstset{framexbottommargin=1mm}
  \lstset{framexleftmargin=-5mm}
  \lstset{stringstyle=\underbar}
  \lstset{mathescape=true}
}{}

\def\prog{\ensuremath{P}\xspace}
\def\progq{\ensuremath{Q}\xspace}
\def\progr{\ensuremath{R}\xspace}
\def\progss{\ensuremath{S}\xspace}
\def\progs{\ensuremath{\mathbb{P}}\xspace}

\def\state{\ensuremath{\sigma}\xspace}
\def\states{\ensuremath{\Sigma}\xspace}
\def\istates{\ensuremath{\Sigma_\iota}\xspace}

\def\trace{\ensuremath{\tau}\xspace}
\def\traces{\ensuremath{\mathbb{T}}\xspace}
\def\atraces{\ensuremath{\mathbb{T^\rho}}\xspace}

\def\variables{\ensuremath{\mathbb{X}}\xspace}
\def\values{\ensuremath{\mathbb{V}}\xspace}
\def\zvalues{\ensuremath{\mathbb{Z}}\xspace}
\def\rvalues{\ensuremath{\mathbb{R}}\xspace}

\def\memory{\ensuremath{\mu}\xspace}
\def\abmemory{\ensuremath{\mu^\rho}\xspace}
\def\memories{\ensuremath{\mathbb{M}}\xspace}
\def\store{\ensuremath{\varepsilon}\xspace}
\def\astore{\ensuremath{\varepsilon^\rho}\xspace}

\def\heap{\ensuremath{h}\xspace}

\newcommand{\BIND}[2]{#1\!\leftarrow\!#2}

\def\lnums{\ensuremath{\mathbb{L}}\xspace}

\def\crit{\ensuremath{C}\xspace}

\def\uco{\ensuremath{\rho}\xspace}
\def\ucos{\ensuremath{\mbox{\sl uco}}\xspace}
\def\avalue{\ensuremath{\mbox{{\sc{v}}}}\xspace}
\def\avaluee{\ensuremath{\mbox{{\sc{u}}}}\xspace}
\def\astate{{\sigma^\rho}\xspace}
\newcommand{\astatei}[1]{{\sigma^\rho_{#1}}\xspace}
\def\astates{{\Sigma^\rho}\xspace}
\def\iastates{{\Sigma_\iota^\rho}\xspace}

\def\KL{\ensuremath{\mathtt{KL}}\xspace}
\def\KLi{\ensuremath{\mathtt{KL}i}\xspace}
\def\IC{\ensuremath{\mathtt{IC}}\xspace}
\def\SIM{\ensuremath{\mathtt{SIM}}\xspace}

\def\ABSMEMop{\upharpoonright^\alpha}
\newcommand{\ABSMEM}[3]{#1 \upharpoonright^\alpha_{#3} #2}
\newcommand{\MEM}[2]{#1 \upharpoonright #2}
\def\Proj{\ensuremath{\mathit{Proj}}\xspace}
\def\UEA{\ensuremath{\mathcal{U}^{\mathcal{A}}}}

\newcommand{\ABSEQUALSPY}[4]{#3\left(#1\right) =_{y} #3\left(#2\right)}

\def\ASG{\ensuremath{\mathit{ASG}}\xspace}
\def\SA{\ensuremath{\mathit{SA}}\xspace}
\def\EA{\ensuremath{\mathit{EA}}\xspace}

\newcommand{\li}{\ar@{-}} 
\newcommand{\lp}{\ar@{.}} 
\newcommand{\fp}{\ar@{.>}}

\newcounter{todos}
\newcommand{\irule}[2]{\frac{\textstyle\rule[-
1.3ex]{0cm}{3ex}#1}{\textstyle\rule[-.5ex]{0cm}{3ex}#2}}

\newcommand{\val}{\mathbb{V}}

\newcommand*{\disj}[1]   {{\bigcurlyvee_{}^{}\left(#1\right)}}

\newcommand{\iid}{\mbox{\sl\tiny id}}

\newcommand{\Var}{\mathbb{X}}
\newcommand{\ra}{\rightarrow}
\newcommand{\la}{\leftarrow}
\newcommand{\Ra}{\Rightarrow}

\newcommand{\Lra}{\Leftrightarrow}
\newcommand{\lra}{\longrightarrow}

\newcommand{\ov}{\overline}

\def\defi{\mbox{\raisebox{0ex}[1ex][1ex]{$\stackrel{\mbox{\tiny
def}}{\; =\;}$}}}

\def\ok#1{\mbox{\raisebox{0ex}[1ex][1ex]{$#1$}}}

\newcommand{\COMMENT} [1]{}

\def\etal{{\it et al.\ }}
\def\comp{\mathrel{\hbox{\footnotesize${}\!{\circ}\!{}$\normalsize}}}

\def\rarr#1{\mbox{\raisebox{0ex}[1ex][1ex]{$
  \mathrel{\mathop{
\hspace*{1pt}\longrightarrow\hspace*{1pt}}\limits^{\,_{\mbox{\tiny 
\hspace*{-2.2pt}#1}}}}$}}}

\def\defemb#1#2{\expandafter\def\csname #1\endcsname
                              {\relax\ifmmode #2\else\hbox{$#2$}\fi}}
 
\def\2c-math#1#2{{\par\medskip\noindent ${#1}$
                      \par\smallskip
                        \noindent\hspace*{\fill} ${#2}$}
                           \\[10pt]}
			   
\def\cal{\mathcal}			   
\defemb{aF}{{\cal F}^\cA}
\defemb{sF}{{\cal F}^{\star}}
\defemb{cA}{{\cal A}}
\defemb{cB}{{\cal B}}
\defemb{cC}{{\cal C}}
\defemb{cD}{{\cal D}}
\defemb{cE}{{\cal E}}
\defemb{cF}{{\cal F}}
\defemb{cG}{{\cal G}}
\defemb{cH}{{\cal H}}
\defemb{cI}{{\cal I}}
\defemb{cJ}{{\cal J}}
\defemb{cK}{{\cal K}}
\defemb{cL}{{\cal L}}
\defemb{cM}{{\cal M}}
\defemb{cN}{{\cal N}}
\defemb{cO}{{\cal O}}
\defemb{cP}{{\cal P}}
\defemb{cQ}{{\cal Q}}
\defemb{cR}{{\cal R}}
\defemb{cS}{{\cal S}}
\defemb{cT}{{\cal T}}
\defemb{cU}{{\cal U}}
\defemb{cV}{{\cal V}}
\defemb{cW}{{\cal W}}
\defemb{cY}{{\cal Y}}
\defemb{cX}{{\cal X}}
\defemb{cZ}{{\cal Z}}
\defemb{cHys}{{\cH_{\cY}^{\cS}}}
\defemb{cHy}{{\cH_{\cY}}}
\defemb{ld}{\dashv}
\defemb{rd}{\vdash}
\defemb{cGH}{{\cal GH}}
\defemb{cGWP}{{\cal GWP}}
\defemb{tL}{{\tt L}}
\defemb{tH}{{\tt H}}
\defemb{tT}{{\tt T}}
\defemb{tS}{{\tt S}}
\defemb{tP}{{\tt P}}
\defemb{tQ}{{\tt Q}}
\defemb{tR}{{\tt R}}
\defemb{tB}{{\tt B\:}}
\defemb{tA}{{\tt A\:}}
\defemb{tC}{{\tt C\:}}

\newcommand{\sset}[2]{\left\{~#1 ~ \left |
                               \begin{array}{l}#2\end{array}
                          \right.     \right\}}

\newcommand{\ifc}{\mbox{\bf if}}
\newcommand{\thenc}{\mbox{\bf then}}
\newcommand{\elsec}{\mbox{\bf else}}
\newcommand{\nil}{\mbox{\bf null}}
\newcommand{\false}{\mbox{\sl false}}
\newcommand{\true}{\mbox{\sl true}}
\newcommand{\while}{\mbox{\bf while}}

\newcommand{\dow}{\mbox{\bf do}}

\def\tuple#1{\langle #1 \rangle}

\newcommand{\lgfp}[3]{\mbox{\sl #1}^{\mbox{\tiny {$#2$}}}_{\mbox{\tiny
{$#3$}}}}

\newcommand{\NANI}[3]{\left[#1\right]#2\left(#3\right)}

\defemb{tD}{{\tt D\:}}

\defemb{tX}{{\tt X}}

\newcommand{\f}{\ar@{->}}

\newcommand{\relc}[1]{\mbox{\sl\small Rel}^{#1}}
\newcommand{\clor}[1]{\mbox{\sl\small Clo}^{#1}}

\DeclareSymbolFont{italics}{OT1}{cmr}{m}{it}

\DeclareMathSymbol{a}{\mathalpha}{italics}{"61}
\DeclareMathSymbol{b}{\mathalpha}{italics}{"62}
\DeclareMathSymbol{c}{\mathalpha}{italics}{"63}
\DeclareMathSymbol{d}{\mathalpha}{italics}{"64}
\DeclareMathSymbol{e}{\mathalpha}{italics}{"65}
\DeclareMathSymbol{f}{\mathalpha}{italics}{"66}
\DeclareMathSymbol{g}{\mathalpha}{italics}{"67}
\DeclareMathSymbol{h}{\mathalpha}{italics}{"68}
\DeclareMathSymbol{i}{\mathalpha}{italics}{"69}
\DeclareMathSymbol{j}{\mathalpha}{italics}{"6A}
\DeclareMathSymbol{k}{\mathalpha}{italics}{"6B}
\DeclareMathSymbol{l}{\mathalpha}{italics}{"6C}
\DeclareMathSymbol{m}{\mathalpha}{italics}{"6D}
\DeclareMathSymbol{n}{\mathalpha}{italics}{"6E}
\DeclareMathSymbol{o}{\mathalpha}{italics}{"6F}
\DeclareMathSymbol{p}{\mathalpha}{italics}{"70}
\DeclareMathSymbol{q}{\mathalpha}{italics}{"71}
\DeclareMathSymbol{r}{\mathalpha}{italics}{"72}
\DeclareMathSymbol{s}{\mathalpha}{italics}{"73}
\DeclareMathSymbol{t}{\mathalpha}{italics}{"74}
\DeclareMathSymbol{u}{\mathalpha}{italics}{"75}
\DeclareMathSymbol{v}{\mathalpha}{italics}{"76}
\DeclareMathSymbol{w}{\mathalpha}{italics}{"77}
\DeclareMathSymbol{x}{\mathalpha}{italics}{"78}
\DeclareMathSymbol{y}{\mathalpha}{italics}{"79}
\DeclareMathSymbol{z}{\mathalpha}{italics}{"7A}

\DeclareMathSymbol{A}{\mathalpha}{italics}{"41}
\DeclareMathSymbol{B}{\mathalpha}{italics}{"42}
\DeclareMathSymbol{C}{\mathalpha}{italics}{"43}
\DeclareMathSymbol{D}{\mathalpha}{italics}{"44}
\DeclareMathSymbol{E}{\mathalpha}{italics}{"45}
\DeclareMathSymbol{F}{\mathalpha}{italics}{"46}
\DeclareMathSymbol{G}{\mathalpha}{italics}{"47}
\DeclareMathSymbol{H}{\mathalpha}{italics}{"48}
\DeclareMathSymbol{I}{\mathalpha}{italics}{"49}
\DeclareMathSymbol{J}{\mathalpha}{italics}{"4A}
\DeclareMathSymbol{K}{\mathalpha}{italics}{"4B}
\DeclareMathSymbol{L}{\mathalpha}{italics}{"4C}
\DeclareMathSymbol{M}{\mathalpha}{italics}{"4D}
\DeclareMathSymbol{N}{\mathalpha}{italics}{"4E}
\DeclareMathSymbol{O}{\mathalpha}{italics}{"4F}
\DeclareMathSymbol{P}{\mathalpha}{italics}{"50}
\DeclareMathSymbol{Q}{\mathalpha}{italics}{"51}
\DeclareMathSymbol{R}{\mathalpha}{italics}{"52}
\DeclareMathSymbol{S}{\mathalpha}{italics}{"53}
\DeclareMathSymbol{T}{\mathalpha}{italics}{"54}
\DeclareMathSymbol{U}{\mathalpha}{italics}{"55}
\DeclareMathSymbol{V}{\mathalpha}{italics}{"56}
\DeclareMathSymbol{W}{\mathalpha}{italics}{"57}
\DeclareMathSymbol{X}{\mathalpha}{italics}{"58}
\DeclareMathSymbol{Y}{\mathalpha}{italics}{"59}
\DeclareMathSymbol{Z}{\mathalpha}{italics}{"5A}

\newtheorem{mydefinition}[theorem]{Definition}

\newcommand{\tM}{\mathbb{M}}

\defemb{tO}{{\tt O}}
\defemb{tI}{{\tt I}}

\def\defemb#1#2{\expandafter\def\csname #1\endcsname{\relax\ifmmode #2 
\else\hbox{$#2$}\fi}}
\def\ok#1{\mbox{\raisebox{0ex}[1ex][1ex]{$#1$}}}

\newcommand{\UNARYFUNCTION}[2]{#1\ifthenelse{\equal{#2}{}}{}{\left(#2\right)}}
\newcommand{\BINARYFUNCTION}[3]{#1\ifthenelse{\not\equal{#2}{}}{\left(#2\ifthenelse{\equal{#3}{}}{}{,#3}\right)}{}}
\newcommand{\TERNARYFUNCTION}[4]{#1\ifthenelse{\not\equal{#2}{}}{\left(#2\ifthenelse{\equal{#3}{}}{}{,#3}\ifthenelse{\equal{#4}{}}{}{,#4}\right)}{}}
\newcommand{\FOURARYFUNCTION}[5]{#1
\ifthenelse{\equal{#2}{}}{}{\left(#2\ifthenelse{\equal{#3}{}}{}{,#3}\ifthenelse{\equal{#4}{}}{}{,#4}\ifthenelse{\equal{#5}{}}{}{,#5}\right)}}
\newcommand{\UNARYFUNCTIONWITHSUBSCRIPT}[3]{#1\ifthenelse{\equal{#2}{}}{}{_{#2}}\ifthenelse{\equal{#3}{}}{}{\left(#3\right)}}
\newcommand{\BINARYFUNCTIONWITHSUBSCRIPT}[4]{#1\ifthenelse{\equal{#2}{}}{}{_{#2}}\ifthenelse{\equal{#3}{}}{}{\left(#3,#4\right)}}
\newcommand{\BINARYINFIXFUNCTION}[3]{\ifthenelse{\equal{#2}{}}{}{#2} #1 \ifthenelse{\equal{#3}{}}{}{#3}}

\newcommand{\CODE}[1]{\ensuremath{\mbox{\lstinline!#1!}\xspace}\xspace}
\def\xx{\CODE{x}}
\def\yy{\CODE{y}}
\def\zz{\CODE{z}}
\def\ii{\CODE{i}}
\def\jj{\CODE{j}}

\def\nn{\CODE{n}}

\def\ss{\CODE{s}}

\def\ww{\CODE{w}}
\newcommand{\UPDATE}[3]{#1\left[#2 \leftarrow #3\right]}

\newcommand{\VARS}[1]{\UNARYFUNCTION{\textsc{vars}}{#1}}

\newcommand{\SEMANTICS}[1]{\left\llbracket #1 \right\rrbracket}

\def\EQDEF{\mbox{\raisebox{0ex}[1ex][1ex]{$\stackrel{\mbox{\tiny def}}{\; =\;}$}}}

\newcommand{\PROOFSTEP}[1]{\ifthenelse{\equal{#1}{}}{}{[\mbox{#1}]}}

\def\INTEGERS{\mathbb{Z}}
\def\NATURALS{\mathbb{N}}

\def\IMPEXP{{\mbox{\sc Exp}}}
\def\exp{e}

\newcommand{\IMPASSIGN}[2]{#1\ \CODE{:=}\ #2}

\newcommand{\IMPIF}[3]{\CODE{if (}#1\CODE{)}~#2~\CODE{else}~#3}
\newcommand{\IMPWHILE}[2]{\CODE{while (}#1\CODE{)}~#2}

\def\IDDOM{\ensuremath{\uco_{\textsc{id}}}\xspace}
\def\PARDOM{\ensuremath{\uco_{\textsc{par}}}\xspace}
\def\INTDOM{\ensuremath{\uco_{\textsc{int}}}\xspace}
\def\EVEN{\ABSVAL{even}}
\def\ODD{\ABSVAL{odd}}
\def\TOP{\ABSVAL{top}}
\def\BOT{\ABSVAL{bot}}
\def\SIGNDOM{\ensuremath{\uco_{\textsc{sign}}}\xspace}
\def\POS{\ABSVAL{pos}}
\def\NEG{\ABSVAL{neg}}
\def\ZERO{\ABSVAL{zero}}
\def\PARSIGNDOM{\ensuremath{\uco_{\textsc{parSign}}}\xspace}
\def\POSEVEN{\ABSVAL{poseven}}
\def\POSODD{\ABSVAL{posodd}}
\def\NEGEVEN{\ABSVAL{negeven}}
\def\NEGODD{\ABSVAL{negodd}}
\def\ZERODOM{\ensuremath{\uco_{\textsc{zero}}}\xspace}
\def\NULLDOM{\ensuremath{\uco_{\textsc{null}}}\xspace}
\def\NULL{\ABSVAL{null}}
\def\NONNULL{\ABSVAL{non{-}null}}
\def\CYCLEDOM{\ensuremath{\uco_{\textsc{cyc}}}\xspace}
\def\CYCLIC{\ABSVAL{cyc}}
\def\ACYCLIC{\ABSVAL{acyc}}
\def\TOPDOM{\ensuremath{\uco_{\top}}\xspace}

\newcommand{\ABSVAL}[1]{\left[\mathbf{#1}\right]}

\newcommand{\REFSET}[1]{{\bf ref}\left(#1\right)}

\newcommand{\DEFSET}[1]{{\bf def}\left(#1\right)}

\newcommand{\CDEPENDS}[3]{\BINARYINFIXFUNCTION{{\rightsquigarrow}_{\mbox{\tiny\sc $#1$}}}{#2}{#3}}
\newcommand{\NCDEPENDS}[3]{\BINARYINFIXFUNCTION{{\not\rightsquigarrow}_{\mbox{\tiny\sc $#1$}}}{#2}{#3}}

\newcommand{\CLOSETOARROW}[2]{%
  \mathrel{\vbox{\offinterlineskip\tabskip=0pt\halign{%
    \hfil##\hfil\cr
    {\tiny $#1,#2$}\cr
    \noalign{\kern-0.1ex}
    $\rightsquigarrow\!\rightsquigarrow$\cr
}}}}
\newcommand{\ANARROWDEPENDS}[5]{\ok{\BINARYINFIXFUNCTION{\CLOSETOARROWA{#5}{#2}{#3}}{#1}{#4}}}

\newcommand{\CLOSETOARROWA}[3]{%
  \mathrel{\vbox{\offinterlineskip\tabskip=0pt\halign{%
    \hfil##\hfil\cr
    {\tiny $#1#2,#3$}\cr
    \noalign{\kern-0.1ex}
    $\rightsquigarrow\!\rightsquigarrow\!\rightsquigarrow$\cr
}}}}

\newcommand{\ATOMDEPENDS}[5]{\ok{\BINARYINFIXFUNCTION{\CLOSETOARROWA{#5}{#2}{#3}_{\mbox{\tiny\sc at}}}{#1}{#4}}}

\newcommand{\RELEVANT}[1]{\UNARYFUNCTION{\textsc{rel}}{#1}}

\newcommand{\EVAL}[2]{\ok{\UNARYFUNCTION{\SEMANTICS{#1}}{#2}}}
\newcommand{\ABSEVAL}[2]{\ok{\UNARYFUNCTION{\SEMANTICS{#1}^\rho}{#2}}}

\newcommand{\AC}[1]{\UNARYFUNCTION{\mathbb{A}_e}{#1}}
\newcommand{\ACC}[2]{\BINARYFUNCTION{\mathbb{A}'_e}{#1}{#2}}

\newcommand{\ACu}[2]{\UNARYFUNCTION{\mathbb{A}^{#1}_e}{#2}}

\def\FINDNDEPS{\textsc{findNdeps}}
\def\EDEPALG{\textsc{Edep}}

\def\PROVE{\textsc{prove}}

\newcommand{\SUBS}[2]{\left[#1|#2\right]}

\newcommand{\ISATOM}[2]{\textsc{Atom}_{#2}\left(#1\right)}
\newcommand{\ATOMS}[1]{\UNARYFUNCTION{\textsc{Atoms}}{#1}}

\newcommand{\EDEP}[3]{\TERNARYFUNCTION{\textsc{Edep}}{#1}{#2}{#3}}
\newcommand{\ATOMIZE}[2]{\BINARYFUNCTION{\textsc{atomize}}{#1}{#2}}

\defemb{tL}{{\tt L\:}}
\defemb{tH}{{\tt H\:}}

\def\PRED{\beta}
\def\AGREEM{\mathcal{G}}
\newcommand{\AGRS}[2]{[#2{::}#1]}

\newcommand{\RULENAME}[1]{\textsc{#1}}

\newcommand{\GRULENAME}[1]{\textsc{g-#1}}
\newcommand{\GRULE}[3]{\frac{#1}{#2}\ \GRULENAME{#3}}
\newcommand{\PRULENAME}[1]{\textsc{pp-#1}}
\newcommand{\PRULE}[3]{\frac{#1}{#2}\ \PRULENAME{#3}}
\newcommand{\PSMTH}[1]{$\RULENAME{pp}$-#1}
\def\PSYSTEM{\PSMTH{system}\xspace}
\newcommand{\GSMTH}[1]{$\RULENAME{g}$-#1}
\def\GSYSTEM{\GSMTH{system}\xspace}
\def\GSOUNDNESS{\GSMTH{soundness}}
\newcommand{\FSEQ}[1]{\left\langle #1 \right\rangle}
\newcommand{\FSET}[1]{\mathit{fields}(#1)}
\newcommand{\SEQSHARE}[4]{\textsc{sh}_{seq}^{#3,#4}\left(#1,#2\right)}

\newcommand{\PRESERVES}[2]{\textsc{pp}\left(#2,#1\right)}
\newcommand{\PRESERVESB}[3]{\textsc{pp}^{#1}\left(#3,#2\right)}
\newcommand{\PRES}[2]{\vdash #1 : #2}

\newcommand{\TRIPLE}[3]{\left\{#1\right\}\ #2\ \left\{#3\right\}}
\newcommand{\TRIPLEB}[4]{\left\{#1\right\}^{#2}\ #3\ \left\{#4\right\}}

\newcommand{\SHARE}[1]{\textsc{sh}\left(#1\right)}

\newcommand{\DALIAS}[1]{\textsc{dal}\left(#1\right)}

\usetikzlibrary{calc}
\newcommand{\lcancel}[1]{%
    \tikz[baseline=(tocancel.base)]{
        \node[inner sep=0pt,outer sep=0pt] (tocancel) {#1};
        \draw[->] ($(tocancel)+(-5pt,5pt)$) -- ($(tocancel)+(5pt,-5pt)$);
    }%
}%
\newcommand{\rcancel}[1]{%
    \tikz[baseline=(tocancel.base)]{
        \node[inner sep=0pt,outer sep=0pt] (tocancel) {#1};
        \draw[->] ($(tocancel)+(5pt,5pt)$) -- ($(tocancel)+(-5pt,-5pt)$);
    }%
}%
\newcommand{\fpropositionl}[1]{\mbox{\lcancel{\ensuremath{#1}}}}
\newcommand{\fpropositionr}[1]{\mbox{\rcancel{\ensuremath{#1}}}}
\newcommand{\FSHARE}[2]{\textsc{sh}_{#1}\left(#2\right)}

\begin{document}

\pagestyle{headings}  

\markboth{I.~Mastroeni and D.~Zanardini}{Abstract Program Slicing: an Abstract Interpretation-based approach to Program Slicing}

\title{Abstract Program Slicing:\\ an Abstract Interpretation-based approach to\\ Program Slicing}

\author{
  ISABELLA MASTROENI \affil{Universit\`a di Verona, Italy}
  and DAMIANO ZANARDINI \affil{Technical University of Madrid (UPM),
    Spain}
  }

\begin{abstract}
  In the present paper we formally define the notion of \emph{abstract
    program slicing}, a general form of program slicing where
  properties of data are considered instead of their exact value.
  This approach is applied to a language with numeric and reference
  values, and relies on the notion of \emph{abstract dependencies}
  between program components (statements).
  
  The different forms of (backward) abstract slicing are added to an
  existing formal framework where traditional, non-abstract forms of
  slicing could be compared.  The extended framework allows us to
  appreciate that abstract slicing is a generalization of traditional
  slicing, since traditional slicing (dealing with syntactic
  dependencies) is generalized by (semantic) non-abstract forms of
  slicing, which are actually equivalent to an abstract form where the
  \emph{identity} abstraction is performed on data.
  
  Sound algorithms for computing abstract dependencies and a
  systematic characterization of program slices are provided, which
  rely on the notion of \emph{agreement} between program states.
\end{abstract}


{\let\setcounter\orgsetcounter
\begin{bottomstuff}
  Authors' addresses: Isabella Mastroeni, Dipartimento di Informatica,
  Facolt\`a di Scienze, Universit\`a di Verona, Strada Le Grazie 15,
  37134 Verona, Italy; Damiano Zanardini, Departamento de Inteligencia
  Artificial, Escuela T\'ecnica Superior de Ingenieros Inform\'aticos,
  Campus de Montegancedo, Boadilla del Monte, 28660 Madrid, Spain.
\end{bottomstuff}
} 

\maketitle

\section{Introduction}
 % introduction
It is well-known that, as the size of programs increases, it becomes
impractical to maintain them as monolithic structures.  Indeed,
splitting programs into smaller pieces allows to construct, understand
and maintain large programs much more easily.  \emph{Program slicing}
\cite{BinGalla96,DeLucia,Tip95,weiser} is a program-manipulation
technique that extracts, from programs, those statements which are
\emph{relevant} to a particular computation.  In the most traditional
definition, a \emph{program slice} is an executable program whose
behavior must be identical to a specific subset of the original
program's behavior.  The specification of this subset is called the
\emph{slicing criterion}, and can be expressed as the value of some
set of variables at some set of statements and/or program points
\cite{weiser}.  Slicing\footnote{We use \emph{slicing} (\emph{slice})
  and \emph{program slicing} (\emph{program slice}) as interchangeable
  terms.} can be and is used in several areas like debugging
\cite{weiser}, software maintenance \cite{GL91ieee}, comprehension
\cite{Conditioned,FRT96}, or re-engineering \cite{CDM96}.

Since the seminal paper introducing slicing \cite{weiser}, there have
been many works proposing several notions of slicing, and different
algorithms to compute slices (see \cite{DeLucia,Tip95} for good
surveys about existing slicing techniques).  Program slicing is a
transformation technique that reduces the size of programs to analyze.
Nevertheless, the reduction obtained by means of standard slicing
techniques may be not sufficient for simplifying program analyses
since it may keep more statements than those strictly necessary for
the desired analysis.  Suppose we are analyzing a program, and suppose
we want a variable $x$ to have a particular property $\uco$ at a given
program point $n$.  If we realize that $x$ does not have that property
at $n$, then we may want to understand which statements affect that
property of $x$, in order to find out more easily where the
computation went wrong.  In this case, we are not interested in the
exact value of $x$, so that we may not need \emph{all} the statements
that a standard slicing algorithm would return.  Instead, we would
need a technique that returns the minimal amount of statements that
actually affect that specific property of $x$.

\paragraph{Abstract program slicing}

This paper introduces and discusses a "semantic" general notion of
slicing, called \emph{abstract program slicing}, looking for those
statements affecting a property (modeled in the context of abstract
interpretation \cite{CC77}) of a set of variables of interest, the so
called \emph{abstract criterion}.  The idea behind this new notion of
slicing is investigating more \emph{semantically} precise notions of
dependency between variables.  In other words, when a syntactic
dependency is detected, such as the dependency, in an assignment, of
\emph{defined} variables from \emph{used} variables, we look further
for semantic dependencies, i.e., dependencies between \emph{values} of
variables.

Consider the program $\prog$ in Fig.~\ref{fig:ProgP33}, and suppose
that we are interested in the variable \CODE{d} at the end of the
execution.  Standard slicing algorithms extract slices by computing
syntactic dependencies; in this sense, \CODE{d} depends on all
\CODE{c}, \CODE{b} and \CODE{a}, so that a sound slice would have to
take all the statements involving all these variables.  In the figure,
$\progq$ is a \emph{slice} of $\prog$ with respect to that criterion.
However, if we are interested in a more precise, semantic notion of
slicing, then we could observe that the value of \CODE{d} only depends
on the values of variables \CODE{c} and \CODE{b}, so that a more
precise slice would be represented by $\progr$.  Finally, if we are
interested in the parity of \CODE{d} at that point, then we observe
that parity of \CODE{d} does not depend on the value of \CODE{c}, and
$\progss$ is an \emph{abstract slice} of $\prog$ with respect to the
specified criterion.  Even in this simple case, the \emph{abstract
  slice} gives more precise information about the statements affecting
the property of interest.

\begin{figure}
  \begin{center}
    \begin{tabular}{cccc}
      \begin{lstlisting}
   a:=1;
   b:=b+1;
   c:=c+2;
   e:=e+1;
   d:=2*c+b+a-a;
      \end{lstlisting}
      &
      \begin{lstlisting}
   a:=1;
   b:=b+1;
   c:=c+2;
   
   d:=2*c+b+a-a;
      \end{lstlisting}
      &
      \begin{lstlisting}
  
   b:=b+1;
   c:=c+2;
   
   d:=2*c+b+a-a;
      \end{lstlisting}
      &
      \begin{lstlisting}
   
   b:=b+1;
   
   d:=2*c+b+a-a;
      \end{lstlisting} \\[8mm]
      Program \prog & Program \progq & Program \progr &Program \progss
    \end{tabular}
  \end{center}
  \caption{$\progq$, $\progr$ and $\progss$ are, respectively, a
    \emph{slice}, a \emph{semantic slice} and an \emph{abstract
      slice} of $\prog$.\label{fig:ProgP33}}
\end{figure}

\paragraph{Contributions}

In this paper, we aim at introducing a generalized notion of slicing,
allowing us to weaken the notion of "dependency" (from syntax, to
semantics, to abstract semantics) with respect to what is considered
\emph{relevant} for computing the slice.  Since our generalization is
a semantic one, we start from the unifying framework proposed in
\cite{AForm,TheoFoun}, where different forms of slicing are defined
and compared w.r.t.~their characteristics (static/dynamic,
iteration-count/non-iteration-count, etc.), into a comprehensive
formal framework.  The structure of this framework is based on the
formal definition of the criterion, inducing a semantic equivalence
relation $\cE$ which uniquely characterizes the set of possible slices
of a program $\prog$ as the set of all the sub-programs\footnote{The
  framework proposed in \cite{AForm,TheoFoun} is parametric on the
  syntactic relation, but here we only consider the relation of being
  a subprogram.} equivalent to $\prog$ w.r.t. $\cE$.  This structure
makes the framework suitable for the introduction and the formal
definition of an abstract form of slicing, since abstraction
corresponds simply to consider a weaker criterion, which implies
weakening the equivalence relation $\cE$ defining slicing.

Once we have the equivalence relation defining a desired notion of
slicing w.r.t. a given criterion, we show how this corresponds to
fixing the notion of dependency we are interested in (namely, the
notion of dependency determining what has to be considered relevant in
the construction of slicing), and we show how the extension to
semantic dependencies may be used to extend the program dependency
graph-based approach to computing slices \cite{horPR89}.  Finally, we
define a notion of abstract dependencies implying abstract criteria.

We show that this new notion of dependency is not suitable for
computing slices by using Program Dependency Graphs, and propose
algorithms for computing (abstract) dependencies and a systematic
approach to compute backward slices.  Such an approach relies on two
systems of logical rules in order to prove (1) Hoare-style tuples
capturing the effect of executing a statement $s$ on a pair of states
for which some similarity (\emph{agreement}) is required by the
slicing criterion (indeed, this similarity corresponds to the semantic
equivalent relation); and (2) when some properties of the state do not
change (are \emph{preserved}) after $s$ is executed.  The combination
of the results provided by these rule systems allows to decide whether
it is safe to remove a statement from a program without changing the
observation corresponding to the criterion.

Importantly, the rule systems and algorithms provided in Section
\ref{sec:theQuestForAbstractSlices} rely on the knowledge and
manipulation of a ``library'' of abstract properties.  For example, in
order to infer that \CODE{2*x} is always even, the abstract domain
representing the parity of number must be known.  If no abstract
property is known except the identity (which is the most precise
property, and is not really abstract), then the approach boils down to
standard slicing.  Importantly, it becomes clear in this case that
slices on the same variables (properties of them in the abstract case;
exact values in the concrete case) are generally bigger in the
concrete setting (when identity is the only available property) with
respect to the corresponding abstract slicing.  Needless to say, this
does not mean that every algorithm for abstract slicing will perform
better than any algorithm for non-abstract slicing; rather, it
provides a practical insight of how optimal (purely semantic-based)
abstract slices may not include statements which are included in
concrete slices.

Part of this work has been previously published in conference
proceedings \cite{MastroeniZanardini,Zanardini,MastroeniNicolic}.  The
present paper joins these works into a coherent framework, and
contains a number of novel contributions
\begin{itemize}
\item[$\bullet$] We formally prove that abstract slicing in the formal
  framework of \cite{AForm} generalizes concrete forms of slicing.
\item[$\bullet$] We formally define the notion of dependency induced
  by a particular criterion, i.e., by the equivalence relation among
  programs induced, in the formal framework, by the chosen criterion.
\item[$\bullet$] We define and prove how we can approximate this
  (concrete semantic) dependency in order to use it for pruning PDGs
  and computing slicing with the well known PDG-based algorithm for
  slicing \cite{RY88}.
\item[$\bullet$] We discuss why the idea of pruning PDGs is not
  applicable to the abstract notion of dependency motivating the need
  of providing different approaches for computing abstract slices.
\item[$\bullet$] The treatment of non-numerical values when computing
  slices was already considered in \cite{Zanardini}.  However, the
  language under study in the present paper is different in that it is
  closer to standard object-oriented languages.  More concretely, that
  work used complex identifiers $x.f.g$ as if they were normal
  variables, thus obtaining that sharing between variables was easier
  to deal with.  However, this came at the cost of increasing the
  number of ``variables'' to be tracked by the analysis.  Moreover,
  examples have been provided to illustrate how properties of the heap
  can be taken into account.
\item[$\bullet$] The \GSYSTEM introduced here is a quite refined
  version of the $\RULENAME{a}$-system \cite{Zanardini}; rules for
  variable assignment and field update have been changed according to
  the new language (which implies a number of technical issues); there
  is a new rule $\GRULENAME{id}$; the overall discussion has been
  improved.
\item[$\bullet$] The rule system for proving the preservation of
  properties (the \PSYSTEM) is explicitly introduced here.
\item[$\bullet$] The description of how statements can be erased has
  been improved; an algorithm has been explicitly introduced, which
  labels each program point with agreements according to the \GSYSTEM.
  A thorough discussion and proofs are provided, so that it is
  guaranteed that the conditions for erasing a statement (relying on
  the \GSYSTEM, the \PSYSTEM, and the \CODE{labelSequence} procedure
  for labeling program points with agreements) are sound.
\item[$\bullet$] Recent work on field-sensitive sharing analysis
  \cite{ZanardiniG15sh} is included in the computation of abstract
  slices, which results in improving the precision when data structure
  in the heap overlap.
\end{itemize}

\section{Preliminaries}
\label{sec:preliminaries}
 % preliminaries
\subsection{The programming language}
\label{sec:theProgrammingLanguage}

The language is a simple imperative language with basic
object-oriented features, whose syntax will be easy to understand for
anyone who is familiar with imperative programming and object
orientation.  The language syntax includes the usual arithmetic
expressions $\IMPEXP$ and access to object fields via ``dot''
selectors.  A statement can be \CODE{skip}, a variable assignment
\CODE{x:=e}, a field update \CODE{x.f:=e}, a conditional or a
\CODE{while} loop.  In addition, there exist special statements (1)
\CODE{read} which reads the value of some variable from the input,
simulating the use of parameters; this kind of statement can only
appear at the beginning of the program; and (2) \CODE{write}, which
can only appear at the end of the program and outputs the current
value of some variables\footnote{As a matter of fact, this kind of
  statement is only included in the language for back-compatibility
  and readability.}.  For simplicity, guards in conditionals and loops
are supposed not to have \emph{side effects}.  We denote by $\progs$
the set of all programs.

$\variables$ is the set of program variables and $\values$ denotes the
set of values, which can be either integer or reference values, or the
$\nil$ constant ($\values = \zvalues \cup \rvalues \cup \{\nil\}$);
every variable is supposed to be well-typed (as integer or reference)
at every program point.  $\lnums$ denotes the set of \emph{line
  numbers} (program points).  Let $l\in\lnums$, and $\stm(l)$ be the
statement at program line $l$.  For a given program $\prog$, we denote
by $\lnums_\prog\subseteq\lnums$ the set of all and only the line
numbers corresponding to statements of the program $\prog$, i.e.,
$\lnums_\prog=\sset{l\in\lnums}{\stm(l)\in\prog}$.  This definition is
necessary since when we look for slicing we erase statements without
changing the numeration of line numbers; for instance, in
Figure~\ref{fig:ProgP33}, we have that $\stm(4)\notin\progq$, so that
$\lnums_\progq=\{1,2,3,5\}$.

A \emph{program state} $\state\in\states$ is a pair
$\tuple{n^k,\memory}$ where $n$ is the executed program point, $k$ is
the number of times the statement at $n$ has been reached so far,
$\memory$ is the memory.
A \emph{memory} is a pair $(\store,\heap)$ where the \emph{store}
$\store:\Var\ra\values$ maps variables to values, and the \emph{heap}
$\heap$ is a sequence of locations where objects can be stored; a
reference value corresponds to one of such locations.  An
\emph{object} $o$ maps field identifiers to values, in the usual way;
$o.f$ is the value corresponding to the field $f$ of the object $o$,
and can be either a number, the location in which another object is
stored, or $\nil$.  For the sake of simplicity, \emph{classes} are
supposed to be declared somewhere, and field accesses are supposed to
be consistent with class declarations.

Unless ambiguity may arise, a memory (or even an entire program state)
can be represented directly as a store, so that $\memory(x)$ (resp.,
$\state(x)$) will be the value of $x$ in the store contained in
$\memory$ (resp., in $\state$).  Moreover, a store $\store$ can be
represented as $\{ \BIND{x_1}{v_1}, .., \BIND{x_m}{v_m} \}$, meaning
that $\store(x_i) = v_i$ for every $i$, and, again, $\memory = \{
\BIND{x_1}{v_1}, .., \BIND{x_m}{v_m} \}$ (resp., $\state = \{
\BIND{x_1}{v_1}, .., \BIND{x_m}{v_m} \}$) can be used instead of
$\store = \{ \BIND{x_1}{v_1}, .., \BIND{x_m}{v_m} \}$ whenever the
store is the only relevant part of the memory (resp., the state).

A \emph{state trajectory} $\trace\in\traces=\states^*$ is a sequence
of program states through which a program goes during the execution.
State trajectories are actually traces equipped with the
$\overline{k}$ component.
The state trajectory obtained by executing program $\prog$ from the
input memory $\memory$ is denoted $\trace_\prog^\memory$.  Moreover,
$\trace[n]$ will be the set of states in $\trace$ where the program
point is $n$. Any initial state has $n=1$, i.e., the set of initial
states is
$\istates=\sset{\tuple{1^{1},\memory}}{\memory\in\memories}$.

In the following,
$\SEMANTICS{\cdot}:\progs\times\wp(\istates)\ra\wp(\traces)$ denotes
the program semantics where $\SEMANTICS{\prog}(S)$ returns the set of
state trajectories obtained by executing the program $\prog$ starting
from any initial state in $S\subseteq\istates$, i.e.,
$\SEMANTICS{\prog}(S)=\sset{\trace_\prog^\memory}{\tuple{1^{1},\memory}\in
  S}$.  We abuse notation by denoting in the same way also the
semantics of expressions, namely,
$\SEMANTICS{\cdot}:\IMPEXP\times\states\ra\values$, which is such that
$\SEMANTICS{\exp}(\state)$ ($\exp\in\IMPEXP$) returns the evaluation
of $\exp$ in $\state$.  Finally, if $S\subseteq\states$, in sake of
simplicity, we still abuse notation by denoting in the same way also
the additive lift of semantics, i.e.,
$\SEMANTICS{\exp}(S)=\sset{\SEMANTICS{\exp}(\state)}{\state\in S}$.

\subsection{Basic Abstract Interpretation}
\label{sec:basicAbstractInterpretation}

This section introduces the lattice of \emph{abstract
  interpretations}~\cite{CC77}.  Let $\langle C, \leq, \vee, \wedge,
\top, \bot \rangle$ denote a complete lattice $C$, with ordering
$\leq$, lub $\vee$, glb $\wedge$, top and bottom element $\top$ and
$\bot$, respectively.  A \emph{Galois connection} (G.c.) is a pair of
monotone functions $\alpha:C \rightarrow A$ and $\gamma:A \rightarrow
C$ such that $\alpha(x)\leq_{A} y \Leftrightarrow x \leq_{C}
\gamma(y)$.  In standard terminology, $C$ and $A$ are, respectively,
the concrete and the abstract domain.  Abstract domains can be
formulated as upper closure operators ($\uco$) \cite{CC77}.  Given an
ordered set $C$ with ordering $\leq_C$, a uco on $C$, $\uco: C
\rightarrow C$, is a monotone, idempotent ($\uco(\uco(x))=\uco(x)$)
and extensive ($\forall x \in C.~x \leq_C \uco(x)$) map.  Each uco
$\uco$ is uniquely determined by the set of its fixpoints, which is
its image; i.e., $\uco(C) = \{x \in C \mid \uco(x) = x\}$.  When
$C=\wp(D)$ for some set $D$, and $v\in D$ then we usually write
$\uco(v)$ instead of $\uco(\{v\})$ (and in general for any function,
$f(v)$ instead of $f(\{v\})$).  If $C$ is a complete lattice, then
$\langle \ucos(C), \sqsubseteq, \sqcup, \sqcap, \lambda x . \top,
\lambda x . x\rangle$ is a complete lattice, where $\ucos(C)$ is the
domain of all the upper closure operators on the lattice $C$; for
every two ucos $\uco_1, \uco_2 \in \ucos(C)$, $\uco_1 \sqsubseteq
\uco_2$ if and only if $\forall y \in C.~\uco_1(y) \leq \uco_2(y)$ iff
$\uco_2(C) \subseteq \uco_1(C)$; and, for every $\{\uco_i\}_{i\in I}
\subseteq \ucos(C)$, $(\sqcap_{i \in I}\uco_i)(x) = \wedge_{i \in
  I}\uco_i$ and $(\sqcup_{i \in I}\uco_i)(x) = x \Leftrightarrow
\forall i \in I.~\uco_i(x) = x$.  In the following we will denote by
$\IDDOM$ the most concrete uco on a domain, i.e., $\lambda x.x$, and
by $\TOPDOM$ the most abstract one $\lambda x.\top$. $A_1$ is more
precise than $A_2$ (i.e., $A_2$ is an abstraction of $A_1$) iff $A_1
\sqsubseteq A_2$ in $\ucos(C)$.  The \emph{reduced product} of a
family $\{\uco_i\}_{i\in I}$ is $\sqcap_{i\in I}\uco_i$ and is one of
the best-known operations for composing abstract domains.

\begin{example}[Numerical abstract domains]
  Let the concrete domain $C$ be $\wp(\INTEGERS)$: the \emph{parity}
  abstract domain $\PARDOM$ in Figure~\ref{dom1} (on the left)
  represents the parity of numbers, and is determined by fix-points $\{
  \BOT, \EVEN, \ODD, \TOP \}$ where $\EVEN$ and $\ODD$ denote even and
  odd numbers, respectively; $\BOT$ is the empty set, and $\TOP =
  \INTEGERS$.  For example, $\PARDOM(\{2,4,10\}) = \EVEN$ (all numbers
  are even), $\PARDOM(\{3,7\}) = \ODD$ (both numbers are odd), and
  $\PARDOM(\{4,5\}) = \TOP$ (there are both even and odd numbers).
  The \emph{sign} abstract domain $\SIGNDOM$ in Figure~\ref{dom1} (on
  the right) is characterized by fix-points $\{ \BOT, \ZERO, \POS,
  \NEG, \TOP \}$ and tracks the sign of integers (zero, positive,
  negative, etc.).  For example, $\SIGNDOM(\{0\}) = \ZERO$,
  $\SIGNDOM(\{-3,-4,-5\}) = \NEG$, $\SIGNDOM(\{1,2,4\}) = \POS$,
  $\SIGNDOM(\{1,-1\}) = \TOP$.  Finally, the \emph{parity-sign} domain
  $\PARSIGNDOM$, which is the reduced product $\sqcap$ of $\PARDOM$
  and $\SIGNDOM$, captures both properties (the parity and the sign),
  and has fix-points $\BOT$, $\ZERO$, $\POSEVEN$, $\POSODD$,
  $\NEGEVEN$, $\NEGODD$, $\EVEN$, $\ODD$, $\POS$, $\NEG$, and $\TOP$.
  \begin{figure}[h]
    \begin{center}
      \begin{tikzpicture}
        \node (partop) at (0.2,0) {$\TOP$};
        \node (pareven) at (-0.6,-1.5) {$\EVEN$};
        \node (parodd) at (1,-1.5) {$\ODD$};
        \node (parbot) at (0.2,-3) {$\BOT$};
        \draw (partop) -- (pareven);
        \draw (partop) -- (parodd);
        \draw (pareven) -- (parbot);
        \draw (parodd) -- (parbot);
        
        \node (signtop) at (3,0) {$\TOP$};
        \node (signpos) at (2,-1.5) {$\POS$};
        \node (signneg) at (4,-1.5) {$\NEG$};
        \node (signzero) at (3,-1.5) {$\ZERO$};
        \node (signbot) at (3,-3) {$\BOT$};
        \draw (signtop) -- (signpos);
        \draw (signtop) -- (signneg);
        \draw (signtop) -- (signzero);
        \draw (signpos) -- (signbot);
        \draw (signneg) -- (signbot);
        \draw (signzero) -- (signbot);

        \node (parsigntop) at (7.7,0) {$\TOP$};
        \node (parsignpos) at (6.2,-1) {$\POS$};
        \node (parsigneven) at (7.2,-1) {$\EVEN$};
        \node (parsignodd) at (8.2,-1) {$\ODD$};
        \node (parsignneg) at (9.2,-1) {$\NEG$};
        \node (parsignposeven) at (5.4,-2) {$\POSEVEN$};
        \node (parsignposodd) at (6.7,-2.5) {$\POSODD$};
        \node (parsignnegeven) at (8.7,-2.5) {$\NEGEVEN$};
        \node (parsignnegodd) at (10,-2) {$\NEGODD$};
        \node (parsignzero) at (7.7,-3) {$\ZERO$};
        \node (parsignbot) at (7.7,-4) {$\BOT$};
        \draw (parsigntop) -- (parsignpos);
        \draw (parsigntop) -- (parsignneg);
        \draw (parsigntop) -- (parsigneven);
        \draw (parsigntop) -- (parsignodd);
        \draw (parsignpos) -- (parsignposeven);
        \draw (parsignpos) -- (parsignposodd);
        \draw (parsignneg) -- (parsignnegeven);
        \draw (parsignneg) -- (parsignnegodd);
        \draw (parsigneven) -- (parsignposeven);
        \draw (parsigneven) -- (parsignnegeven);
        \draw (parsigneven) -- (parsignzero);
        \draw (parsignodd) -- (parsignposodd);
        \draw (parsignodd) -- (parsignnegodd);
        \draw (parsignposeven) -- (6,-3.2) -- (parsignbot);
        \draw (parsignnegeven) -- (parsignbot);
        \draw (parsignposodd) -- (parsignbot);
        \draw (parsignnegodd) -- (9.4,-3.2) -- (parsignbot);
        \draw (parsignzero) -- (parsignbot);
      \end{tikzpicture}
    \end{center}
    \caption{The $\PARDOM$, $\SIGNDOM$ and $\PARSIGNDOM$
      domains.} \label{dom1}
  \end{figure} 
\end{example}

Formally speaking, the value of a reference variable is either a
location $\ell$ or $\nil$.  However, the domains introduced in the
next example classify variables not only with respect to $\ell$
itself, but also on the data structure in the heap which is reachable
from $\ell$.  This point of view is similar to previous work on static
analysis of properties of the heap like \emph{sharing}
\cite{DBLP:conf/sas/SecciS05} or \emph{cyclicity}
\cite{DBLP:conf/vmcai/RossignoliS06,tcs13}.

\begin{example}[Reference abstract domains]
  \label{ex:referenceAbstractDomains}
  Let $C$ be $\wp(\rvalues{\cup}\{\nil\})$, i.e., the possible values
  of reference variables.  The \emph{nullity} domain $\NULLDOM$
  classifies values on nullity, and has fix-points $\{ \BOT, \NULL,
  \NONNULL, \TOP \}$ where the concretizations of $\NULL$ and
  $\NONNULL$ are, respectively, $\{ \nil \}$ and $\rvalues$.
  
  On the other hand, it is possible to define a \emph{cyclicity}
  domain $\CYCLEDOM$ which classifies variables on whether they point
  to \emph{cyclic} or \emph{acyclic} data structures \cite{tcs13}.  A
  \emph{cycle} in the heap is a path in which the same location is
  reached more than once; a double-linked list (one which can be
  traversed in both directions) is a good example of a cyclic data
  structure.  The fix-points of this domain are $\{ \BOT, \CYCLIC,
  \ACYCLIC, \TOP \}$, where all acyclic values (including $\nil$) are
  abstracted to $\ACYCLIC$, and all cyclic values (i.e., locations
  from which a cycle is reachable) are abstracted to $\CYCLIC$.  Both
  domains and their reduced product are depicted in Figure \ref{dom2};
  note that there are values which are both null and cyclic, so that
  their intersection collapses to $\BOT$.
  
  Finally, the identity domain $\IDDOM$, abstracts two concrete values
  to the same abstract value only if they are equal.  Two references
  are \emph{equal} if (1) their are both null; or (2) they are both
  non-null and the objects stored in the corresponding locations are
  equal, where equality on objects means that all their numeric fields
  must be the same number and all reference fields must be equal
  (w.r.t.~this same notion of equality on references).

  \begin{figure}[h]
    \begin{center}
      \begin{tikzpicture}
        \node (nulltop) at (0.2,0) {$\TOP$};
        \node (nullnull) at (-0.6,-1.5) {$\NULL$};
        \node (nullnonnull) at (1,-1.5) {$\NONNULL$};
        \node (nullbot) at (0.2,-3) {$\BOT$};
        \draw (nulltop) -- (nullnull);
        \draw (nulltop) -- (nullnonnull);
        \draw (nullnull) -- (nullbot);
        \draw (nullnonnull) -- (nullbot);
        
        \node (cycletop) at (4,0) {$\TOP$};
        \node (cyclecyclic) at (3,-1.5) {$\CYCLIC$};
        \node (cycleacyclic) at (5,-1.5) {$\ACYCLIC$};
        \node (cyclebot) at (4,-3) {$\BOT$};
        \draw (cycletop) -- (cyclecyclic);
        \draw (cycletop) -- (cycleacyclic);
        \draw (cyclecyclic) -- (cyclebot);
        \draw (cycleacyclic) -- (cyclebot);

        \node (top) at (9.7,0) {$\TOP$};
        \node (cyclic) at (7.7,-1.2)
              {$\left(\!\!\begin{array}{c}\NONNULL\\\CYCLIC\end{array}\!\!\right)$};
        \node (null) at (9.7,-1.7)
              {$\left(\!\!\begin{array}{c}\NULL\\\ACYCLIC\end{array}\!\!\right)$};
        \node (acyclic) at (11.7,-1.2)
              {$\left(\!\!\begin{array}{c}\NONNULL\\\ACYCLIC\end{array}\!\!\right)$};
        \node (bot) at (9.7,-3) {$\BOT$};
        \draw (top) -- (null);
        \draw (top) -- (cyclic);
        \draw (top) -- (acyclic);
        \draw (null) -- (bot);
        \draw (cyclic) -- (bot);
        \draw (acyclic) -- (bot);
      \end{tikzpicture}
    \end{center}
    \caption{The $\NULLDOM$ and the $\CYCLEDOM$ domains, and their
      reduced product.} \label{dom2}
  \end{figure} 
\end{example}

Let us consider now $D=C^n$ ($C$ lattice and $n\in\NATURALS$), namely
$x\in D$ is a $n$-tuple of elements of $C$, and consider
$\uco\in\ucos(D)$.  In this case, we can distinguish between two kinds
of abstractions:\label{rel-nonrel} \emph{non-relational} and
\emph{relational} abstractions \cite{C01-Dag,CC79}.  The
non-relational or \emph{attribute-independent} one \cite[Example
  6.2.0.2]{CC79} consists in ignoring the possible relationships
between the values of the abstracted inputs.  For instance, if $\uco$
is applied to the values of variables $x$ and $y$, then $\uco$ can be
approximated through projection by a pair of abstractions on the
single variables, analyzing the single variables in isolation.  In
sake of simplicity, without losing generality, consider $n=2$, i.e.,
$D=C^2=C\times C$.  Formally, $\uco\in\ucos(C\times C)$ is
non-relational if there exist $\delta_1,\delta_2\in\ucos(C)$ such that
$\uco(x,y)=\tuple{\delta_1(x),\delta_2(y)}$, i.e,
$\uco\in\ucos(C)\times\ucos(C)\subset\ucos(C\times C)$. For instance,
let $\PARDOM$ be the abstract domain depicted in Figure~\ref{dom1}
expressing the parity of integer values; the $\PARDOM$ non-relational
property of $\tuple{x,y}$ provides the parity of $x$ and $y$
independently one from each other, meaning that all the possible
combinations of parity of $x$ and $y$ are possible as results
($\tuple{\EVEN,\EVEN},\tuple{\EVEN,\ODD},\tuple{\ODD,\EVEN},\tuple{\ODD,\ODD}$
and all combinations where at least one variable is $\TOP$ or $\BOT$).
Relational abstractions may preserve some of the relationship between
the analyzed values \cite{C01-Dag}. For instance, we could define an
abstraction preserving the fact that $x$ is even ($\EVEN$) if and only
if $y$ is odd ($\ODD$). It is clear that, in this case, we are more
precise since the only possible analysis results are
$\tuple{\EVEN,\ODD}$, $\tuple{\ODD,\EVEN}$, $\tuple{\TOP,\TOP}$ and
$\tuple{\BOT,\BOT}$.

If $\uco\in\ucos(C)$, $\ok{f\in C\rarr{} C}$, and $\ok{f^\sharp\in
  \uco(C)\rarr{} \uco(C)}$, then $\ok{f^\sharp}$ is a {\em sound\/}
approximation of $f$ if $\ok{\uco\comp f \sqsubseteq
  f^\sharp\comp\uco}$.  $\ok{f^{\alpha}\defi\uco\comp f\comp\uco}$ is
known as the {\em best correct approximation\/} (bca) of $f$ in
$\uco$, which is always sound by construction.  Soundness naturally
implies fix-point soundness, that is,
$\ok{\uco(\lgfp{lfp}{\leq_{\mbox{\tiny $C$}}}{\bot_{\mbox{\tiny $C$}}}
  f_C) \leq_{\uco} \lgfp{lfp}{\leq_{\mbox{\tiny
        $C$}}}{\bot_{\mbox{\tiny $C$}}} f^\uco}$. If $\uco\comp f =
\uco\comp f\comp\uco$ then we say that $f^\uco$ is a {\em complete\/}
approximation of $f$ \cite{CC79,GRSjacm}.  In this case,
$\ok{\uco(\lgfp{lfp}{\leq_{\mbox{\tiny $C$}}}{\bot_{\mbox{\tiny $C$}}}
  f) = \lgfp{lfp}{\leq_{\mbox{\tiny $C$}}}{\bot_{\mbox{\tiny $C$}}}
  f^\uco}$.

\subsection{Equivalence relations, abstractions and partitions}
\label{Sect:partit}

Closure operators and equivalence relations are related concepts
\cite{CC79}.  Recently, this connection has been further studied in
the field of abstract model checking and language based-security
\cite{RT02,HM05}.  In particular, there exists an isomorphism between
equivalence relations and a subclass of upper closure
operators. Consider a set $S$: for each equivalence relation
${\tR}\subseteq S\times S$ we can define an upper closure operator,
$\clor{\tR}\in\ucos(\wp(S))$ such that $\forall x\in
S\:.\:\clor{\tR}(\{x\})=[x]_\tR$ and $\forall X\subseteq
S\:.\:\clor{\tR}(X)=\bigcup_{x\in X}[x]_\tR$.  Conversely, for each
upper closure operator $\eta\in\ucos(\wp(S))$, we are able to define
an equivalence relation $\relc{\eta}\subseteq S\times S$ such that
$\forall x,y\in S\:.\: x \mathrel{\relc{\eta}}
y\ \Lra\ \eta(\{x\})=\eta(\{y\})$.  It is immediate to prove that
$\relc{\eta}$ is an equivalence relation, and this comes from $\eta$
being merely a function, not necessarily a closure operator.
$\clor{\tR}$ is identified as the most concrete closure $\eta$ such
that ${\tR} = \relc{\eta}$ \cite{HM05}.
It is possible to associate with each upper closure operator the most
concrete closure inducing the same partition on the concrete domain
$S$:
\begin{equation}
  \label{partit}
  \Pi(\eta) \defi \clor{\relc{\eta}}
\end{equation}

\noindent
Note that, for all $\eta \in \ucos(\wp(S))$, $\Pi(\eta)$ is the
(unique) most concrete closure that induces the same equivalence
relation as $\eta$ ($\relc{\eta} = \relc{\Pi(\eta)}$).  The fix-points
of $\Pi$ are called the \emph{partitioning} closures.  Being $\wp(S)$
a complete Boolean lattice, an upper closure operator
$\eta\in\ucos(\wp(S))$ is partitioning, i.e., $\eta=\Pi(\eta)$, iff it
is complemented, namely if $\forall X\in\eta.\:\ov{X}\defi
S\smallsetminus X\in\eta$ \cite{HM05}.
  
\begin{figure}[t]  
  \begin{center}  
    \includegraphics[scale=.25]{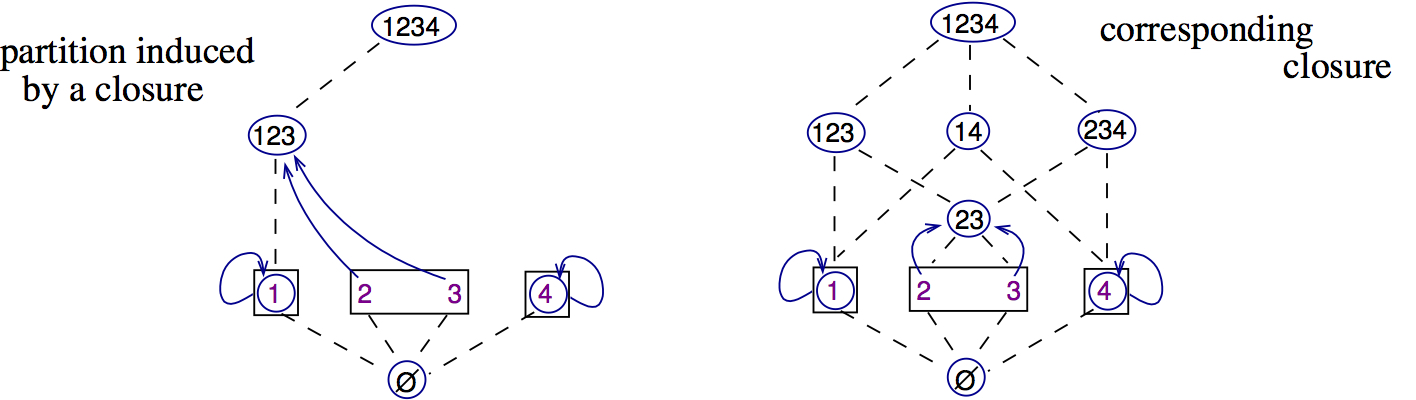}  
    \caption{A partitioning closure.}\label{partizioni}  
  \end{center}  
\end{figure}

\begin{example}  
  Consider the set $S=\{1,2,3,4\}$ and one of its possible partitions
  $\pi=\{\{1\},\{2,3\},\{4\}\}$.  The closure $\eta$ with fix-points
  $\{\emptyset,\{1\},\{4\},\{1,2,3\},S\}$ induces exactly $\pi$ as a
  state partition, but the most \emph{concrete} closure that induces
  $\pi$ is
  $\clor{\pi}=\Pi(\eta)=\disj{\{\emptyset,\{1\},\{2,3\},\{4\}\},S}$,
  which is the closure on the right of Figure \ref{partizioni}.
\end{example}

Given a partitioning upper closure operator $\uco$, an \emph{atom} is
an element $a$ of $\uco$ such that there does not exists another
element $b$ with $\BOT \sqsubset b \sqsubset a$.  For example, the
atoms of $\PARSIGNDOM$ are $\POSEVEN$, $\POSODD$, $\ZERO$, $\NEGEVEN$,
and $\NEGODD$.  In partitioning closures, atoms are all the possible
abstractions of singletons: in fact, $\PARSIGNDOM(\{n\})$ will never
give $\POS$ or $\ODD$ since there is always a more precise abstract
value describing $n$.  In the following, $\ISATOM{a}{\uco}$ holds iff
$a$ is an atom of $\uco$.

\subsection{Abstract semantics}
\label{sec:abstractSemantics}

An abstract program semantics is the abstract counterpart of the
concrete semantics w.r.t.\ an abstract program observation: it is
meant to compute, for each program point, an abstract state which
soundly represents \emph{invariant properties} of variables at that
point.  In general, it is computed by an abstract interpreter
\cite{CC79} collecting the set of all the possible values that each
variable may have in each program point and abstracting this set in
the chosen abstract domain.

Given a concrete program state $\state$ and an abstract domain
$\uco\in\ucos(\wp(\values))$, an \emph{abstract state}
$\astate\in\astates$ is obtained by applying the abstraction $\uco$ to
the values of variables stored in it.
Namely, $\astate=\tuple{n^{k_n},\abmemory}$, where
$\abmemory=\tuple{\astore}$ and $\astore$ is such that
$\astore(x)=\uco(\store(x))$.  For simplicity, we can write
$\astate(x)=\uco(\state(x))$, treating the whole state as a store when
applied to variables.  In the case of a reference variable $x$, the
abstraction $\astate(x)$ gives information about the data structure
pointed to by $x$ (e.g., if $\uco = \CYCLEDOM$, the cyclicity of the
data structure can be represented).  This explains why the heap is not
represented explicitly in the abstract state: instead, relevant
information about the heap is contained in the abstraction of
variables (see the previous discussion before Example
\ref{ex:referenceAbstractDomains}).

In the following, ordering $\leq$ on abstract states is variable-wise
comparison between abstract values:
\[ \state^\uco_1 \leq \state^\uco_2 \quad \Leftrightarrow \quad \forall
x. \state^\uco_1(x) \subseteq \state^\uco_2(x) \] The greater an
abstract state is, the wider is the set of concrete states it
represents.  Moreover, a \emph{covering} of $\astate$ is a set of
abstract states $\{ \state^\uco_1..\state^\uco_n \}$ such that $\vee_i
\state^\uco_i = \astate$.  The set of abstract state trajectories is
$\atraces=\astates^*$, namely an abstract trajectory is the
computation of a program on the set of abstract states. The trace in
$\atraces$ of a program $\prog$, starting from the abstract memory
$\abmemory$ is denoted by $\trace^{\abmemory}_\prog$.

Formally, the \emph{abstract program semantics}
$\ABSEVAL{\cdot}{}:\progs\times\iastates\ra\atraces$ is such that
$\ok{\ABSEVAL{\prog}{}(S)=\{\trace^{\abmemory}_\prog~|~\tuple{1^1,\abmemory}\in
  S\}}$ is the set of the sequences of abstract states computed
starting from the abstract initial states in $S\in\wp(\iastates)$.  We
also abuse notation by denoting $\ABSEVAL{\cdot}{}$ also the abstract
evaluation of expressions. Namely,
$\ok{\ABSEVAL{\cdot}{}:\IMPEXP\times\astates\ra\uco(\values)}$ is such
that $\forall
x.\:\ABSEVAL{\exp}{}(\astate(x))=\uco(\SEMANTICS{\exp}(\astate(x)))=\uco(\SEMANTICS{\exp}(\uco(\state(x))))$. This
definition is correct, since by construction, we have that any
abstract state $\astate$ corresponds to a set of concrete states,
i.e.,
$\ok{\astate=\sset{\ov{\state}\in\states}{\ov{\state}^\uco=\astate}}=\ok{\sset{\ov{\state}\in\states}{\forall
    x\in\variables.\:\ov{\state}^\uco(x)=\astate(x)}}$,
namely, it is the set of all the concrete states having as abstraction
in $\uco$ precisely $\astate$, and we abuse notation by denoting with
$\ABSEVAL{\cdot}{}$ also its additive lift.  In other words,
$\ABSEVAL{\exp}{}$ is the best correct approximation of
$\SEMANTICS{\exp}$ by means of an abstract value in $\uco$.
In general, in order to compute the abstract semantics of a program on
an abstract domain $\uco$, we have to equip the domain $\uco$ with the
abstract versions of all the operators used for defining expressions.
In our language, we should define, for example, the meaning of $+$,
$-$, $*$ and $/$ on abstract values, i.e., on sets of concrete values.
This is standard in abstract interpretation, and these operations are
defined for all the known numerical abstract domains.  For instance,
the sound approximation of the sum operation on $\PARDOM$ is the
following:
\[ \begin{array}{ccccl@{\qquad\qquad}ccccl}
  \EVEN & + & \EVEN & = & \EVEN & \TOP & + & \_ & = & \TOP \\
  \EVEN & + & \ODD & = & \ODD & \_ & + & \TOP & = & \TOP \\
  \ODD & + & \EVEN & = & \ODD & \BOT & + & \_ & = & \BOT \\
  \ODD & + & \ODD & = & \EVEN & \_ & + & \BOT & = & \BOT \\
\end{array} \]
We can reason similarly for all the other operators.  The use of
$\ABSEVAL{\cdot}{}$ in Section
\ref{section:algorithmicIdeasForCheckingNdep} and later in the paper
is twofold: (1) to infer invariant properties, as in Example
\ref{ex:invariantProperties}; and (2) to evaluate expressions at the
abstract level.

\begin{example}
  \label{ex:invariantProperties}
  Consider the following code fragment:
  \begin{lstlisting}
   i := 10;
   j := 0;
   while (i$\geq$0) {
     i := i-1;
     j := j+1;
   }
  \end{lstlisting}
    and an abstraction $\uco=\SIGNDOM$,
  i.e., the property of interest is the sign of both \ii and \jj.  By
  computing the abstract semantics of this simple program, we can
  observe that inside the loop we lose the sign of \ii since \ii
  starts being positive, but then the $\ii-1$ operation makes
  impossible to know statically the sign of \ii (the result may
  be positive, zero or negative starting from \ii positive or zero),
  while we have that \jj always remains positive. Moreover, if the
  loop terminates we can surely say that, at the end, $\ii<0$, namely it is negative
  (due to the negation of the while guard). Hence, we are able to
  infer that \ii is negative and \jj is positive after line 6.  This
  means that the final abstract state $\astate$ is such that
  $\astate(\mbox{\ii}) = \NEG$ and $\astate(\mbox{\jj}) = \POS$ (in
  the following, the extensional notation for $\astate$ will be $\{
  \BIND{\mbox{\ii}}{\NEG},\BIND{\mbox{\jj}}{\POS} \}$, similar to the
  notation for concrete states).
\end{example}

\section{Program Slicing}
 % programSlicing
Program slicing \cite{weiser} is a program-manipulation technique
which extracts from programs those statements which are relevant to a
particular portion of a computation.  In order to answer the question
about which are the relevant statements, an observer needs
a \emph{window} through which only a part of the computation can be
seen \cite{BinGalla96}.  Usually, what identifies the portion of
interest in the computation is the value of some set of variables at a
certain program point, so that a \emph{program slice} comes to be the
subset (syntactically, in terms of statements) of the original program
which contributes directly or indirectly to the values assumed by some
set of variables at the program point of interest.  The \emph{slicing
criterion} is what specifies the part of the computation which is
relevant to the analysis; in this case, a criterion is a pair
consisting of a set $\cX$ of variables and a program point (or line
number) $n$.  The following definition \cite{BinGalla96} is a possible
formalization the original idea of program slicing \cite{weiser}, in
the case of a single variable:
\begin{mydefinition}\cite{BinGalla96}
  \label{defSlice}
  For a statement $s$ (at program point $n$) and a variable $x$, the
  slice $P'$ of the program $P$ with respect to the slicing criterion
  $\tuple{s,\{x\}}$ is any executable program with the following
  properties:
  \begin{enumerate}
  \item $P'$ can be obtained by deleting zero or more statements from
    $P$;
  \item If $P$ halts on the input $I$, then, each time $s$ is reached
    in $P$, it is also reached in $P'$, and the value of $x$ at $s$ is
    the same in $P$ and in $P'$.  If $P$ fails to terminate, then $s$
    may be reached more times in $P'$ than in $P$, but $P$ and $P'$
    have the same value for $x$ each time $s$ is executed by $P$.
  \end{enumerate}
\end{mydefinition}

It is worth noting that Reps and Yang \cite{RY88}, in
their \emph{slicing theorem}, provide implicitly a similar definition
of program slicing, but it only considers terminating
computations. The following example provides the intuition of how
slicing works.

\begin{figure}
  \begin{lstlisting}
   int c, nl := 0, nw := 0, nc := 0;
   int in := false;
   while ((c=getchar())!=EOF) {
     nc := nc+1;
     if (c=' ' || c='\n' || c='\t') {
       in := false; }
     elseif (in = false) { 
       in := true;
       nw := nw+1; }
     if (c = '\n') {
       nl := nl+1; }
   }
  \end{lstlisting}
  \caption{Word-count program.}\label{programs}
\end{figure}
  
\begin{figure}
  \hspace{-3mm}  \begin{tabular}{c|c}
    \begin{lstlisting}
   int c, nl := 0;
   
   while ((c=getchar())!=EOF) {
     
     if (c = '\n') {
       nl := nl+1; }
   }
      \end{lstlisting}
      &
      \begin{lstlisting}
   int c, nw := 0;
   int in := false;
   while ((c=getchar())!=EOF) {
     
     if (c=' ' || c='\n' || c='\t') {
       in := false; }
     elseif (in = false) {
       in := true;
       nw := nw+1; }
     
   }
      \end{lstlisting}
  \end{tabular}
  \caption{Slices of the word-count program.}\label{example}
\end{figure}

\begin{example}
  Consider the word-count program \cite{MDT07} given in
  Figure \ref{programs}.  It takes in a block of text and outputs the
  number of lines (\CODE{nl}), words (\CODE{nw}) and characters
  (\CODE{nc}).  Suppose the slicing criterion only cares for the value
  of \CODE{nl} at the end of the program; then a possible slice is on
  the left in Figure~\ref{example}.  On the other hand, if the
  criterion is only interested in \CODE{nw}, then a correct slice is
  on the right.
\end{example}

Starting from the original definition \cite{weiser}, the notion of
slicing has gone through several generalizations and versions, but one
feature is constantly present: the fact that slicing is based on a
notion of \emph{semantic equivalence} that has to hold between a
program and its slices or on a corresponding notion of
\emph{dependency}, determining what we keep in the slice while
preserving the equivalence relation.  What we can observe about
definitions of slicing such as the one given in
Definition~\ref{defSlice} is that they are enough precise for finding
algorithms for soundly computing slicing, such as \cite{RY88}, but not
enough formal to become suitable to generalizations allowing us to
compare different forms of slicing and/or to define new weaker forms
of slicing.
  
In the following, we use the formal framework proposed in \cite{AForm}
where several notions and forms of slicing are modeled and compared.
This is not the only attempt to provide a formal framework for slicing
(see Section~\ref{sec:relatedWork}), but we believe that, due to its
semantic-based approach, it is suitable to include an abstraction
level to slicing, which can be easily compared with all the other
forms of slicing included in the original framework. Hence, in the
following section we don't rewrite a formal framework, but we
re-formalize the slicing criterion in order to allow us to easily
include abstraction simply as a new parameter. A brief introduction of
the formal framework together with some examples showing the
differences between the different forms of slicing introduced in the
following is given in the Appendix.

\subsection{Defining Program Slicing: the formal framework}
\label{section:Background}
 % programSlicing-frameworkN1
In this section, our aim is to define the form of slicing that we can
lift to an abstract level. Namely, we consider the framework in
\cite{AForm,TheoFoun}, which allows us to define abstract slicing
simply by defining an abstract criterion which, independently from the
kind of slicing (static, dynamic, conditional, standard, etc.) allows
us to observe properties instead of concrete values. Since our aim is
to define abstract program slicing as a form of slicing, perfectly
integrated in the proposed hierarchy and where the criterion simply
has one more parameter describing the abstraction, we need to slightly
revise the construction in order to provide a completely unified
notation for the slicing criterion.
Note that, the present paper will only deal with \emph{backward}
slicing, where the interest is on the part of the program which
\emph{affects} the observation associated with the slicing criterion
and not on the part of the program which \emph{is affected} by such an
observation (called instead \emph{forward} slicing \cite{Tip95}).

\subsubsection*{Defining slicing criteria}

The slicing criterion characterizes what we have to observe of the
program in order to decide whether a program is a slice or not of
another program.  In particular, we have to fix which computations
have to be compared, i.e., the inputs and the observations on which
the slice and the program have to agree.

In the seminal Weiser approach, given a set of variables of interest
$\cX$ and program statement $s$, here referred by the program point
$n$ where $s$ is placed, a slicing criterion was modeled as $\crit =
(\cX,n)$.  In the following, we will gradually enrich and generalize
this model in order to include several different notions and forms of
slicing.  Weiser's approach is known as \emph{static slicing} since
the equivalence between the original program and the slice has,
implicitly, to hold for every possible input.  On the other hand,
Korel and Laski proposed a new technique called \emph{dynamic slicing}
\cite{KorelLaski} which only considers one particular computation, and
therefore one particular input, so that the dynamic slice only
preserves the (subset of the) meaning of the original program for that
input.
Hence, in order to characterize a slicing criterion including also
dynamic slicing we have to add a parameter describing the set of
initial memories $\cI \subseteq \memories$: The criterion is now
$\crit = (\cI,\cX,n)$, where $\cI = \memories$ for static slicing,
while $\cI = \{\memory\}$, with $\memory \in \memories$, for dynamic
slicing.  Finally, Canfora \etal proposed \emph{conditioned slicing}
\cite{Conditioned}, which requires that a conditioned slice preserves
the meaning of the original program for a set of inputs satisfying one
particular condition $\varphi$.  Let
$\cI=\sset{\memory\in\memories}{\memory\ \mbox{satisfies}\ \varphi}$
be the set of input memories satisfying $\varphi$ \cite{AForm}.
Hence, the slicing criterion still can be modeled as
$\crit=(\cI,\cX,n)$.

Each type of slicing comes in four forms which differ on what the
program and the slices must agree on, namely on the observable
semantics that has to agree.  In the following, we provide an informal
definition of these forms in order to provide the intuition of what
will be formally defined afterwards:
\begin{description}
\item[Standard] It considers one point in a program with respect to a
  set of variables. In other words, the standard form of slicing only
  tracks one program point.  Semantically, this form of slicing
  consists in comparing the program and the slices in terms of the
  ({\em denotational}) I/O semantics from the program inputs selected
  by the criterion. Namely, for each selected input, the results of
  the criterion variables in the point of observation must be the
  same, independently from the executed statements.
\item[Korel and Laski ($\KL$)] It is a stronger form where the program
  and the slice must follow identical paths
  \cite{KorelLaski}. Semantically, we could say that the program and
  the slice must have the same ({\em operational}) trace semantics
  w.r.t.\ the statements kept in the slice, starting from the program
  inputs selected by the criterion.  In other words, as before, the
  final value must be the same, but in this case these values must be
  obtained by executing precisely the same statements, i.e., following
  the same execution path.
\item[Iteration count ($\IC$)] When considering the trace semantics,
  the same program point inside a loop may be visited more than once,
  in the following we call {\em $k$-th iteration} of a program point
  $n$ the $k$-th time the program point $n$ is visited. The iteration
  count form of slicing requires that a program and its slice agree
  only at a particular $k$-th iteration of a program point of
  interest.  In this way, when a point of interest is inside a loop,
  we have the possibility to require that the variables must agree
  only at some iterations of the loop and not always.
\item[Korel and Laski iteration count ($\KLi$)] It is the combination
  of the last two forms.
\end{description}
In order to deal with these different forms of slicing, the slicing
criterion must be enriched with additional information.  In
particular, the $\KL$ form of slicing does not change where to observe
variables, but it does change the observed semantics up to that point.
Hence, we simply have to add a boolean parameter $\psi$: \true~means
that we are considering a $\KL$ form and we require that the slice
must agree with the program on the execution of statements that are in
the slice (and obviously also in the original program); on the other
hand, \false~indicates a standard, non-$\KL$ form of slicing.  Hence,
a criterion $\crit$ comes to be $(\cI,\cX,n,\psi)$.

The $\IC$ form, instead, affects the observation: in order to embed
this features in the criterion, the third parameter has to be changed.
Let $\{ k_1,..,k_j \} \subseteq \NATURALS$ be the iterations of the
program point $n \in \lnums$ we are interested in; then, instead of
$n$, in the third parameter of the criterion we should have
$\tuple{n,\{ k_1,..,k_j \}}$.  Therefore, $\crit$ takes the form
$(\cI,\cX,\cO,\psi)$, where $\cO \in \lnums \times \wp(\NATURALS)$.
Note that $\tuple{n,\NATURALS}$ represents the fact that we are
interested in all occurrences of $n$, as it happens in the standard
form.

There are also some \emph{simultaneous} ($\SIM$) forms of slicing that
consider more than one program point of interest.  In order to deal
with $\SIM$ forms of slicing, we simply extend the definition of a
slicing criterion by considering $\cO$ as a set instead of a
singleton, namely, $\cO \in \wp(\lnums \times \wp(\NATURALS))$.

In the Appendix there are some simple examples showing the main
differences between the several forms of slicing introduced so far.

\section{Abstract Program Slicing}
\label{section:AbstractProgramSlicing}
 % abstractProgramSlicing
In this section we define a weaker notion of slicing based on Abstract
Interpretation.  In particular, we generalize the formal framework in
\cite{AForm} in order to include also the abstract versions of
slicing.

Program slicing is used for reducing the size of programs to analyze.
Nevertheless, sometimes this reduction is not sufficient for really
improving an analysis.  Suppose that some variables at some point of
execution do not have a desired property (for example, that they are
different from $0$, or from $\nil$); in order to understand where the
error occurred, it would be useful to find those statements which
affect such a property of these variables.  Standard slicing may
return too many statements, making it hard for the programmer to
realize which one caused the error.

\begin{example}
  \label{ex:abstractSlicing}
  Consider the following program $\prog$, that inserts a new element
  \CODE{elem} at position \CODE{pos} in a single-linked list.  For
  simplicity, let \CODE{pos} never exceed the length of \CODE{list}.
  {\em
    \begin{lstlisting}[firstnumber=34]
   y := null;
   x := list;
   while (pos>0) {
     y := x;
     x := x.next; // by hypothesis, this always succeeds
     pos := pos-1;
   }
   z := new Node(elem);
   z.next := x;
   if (y = null) {
     list := z;
   } else {
     y.next = z;
   }
    \end{lstlisting}
  }
  \noindent Suppose that \CODE{list} is cyclic after line 47, i.e., a
  traversal of the list visits the same node twice.  A close
  inspection of the code reveals that no cycle is created between
  lines 34 and 47: \CODE{list} is cyclic after line 47 if and only if
  it was cyclic before line 34.
  
  In the standard approach, it is possible to set the value of
  \CODE{list} after line 47 as the slicing criterion.  In this case,
  since \CODE{list} can be modified at lines 41--47, at least this
  piece of code must be included in the slice.

  On the other hand, let the cyclicity of \CODE{list} after line 47 be
  the property of interest, represented by $\CYCLEDOM$ (Example
  \ref{ex:referenceAbstractDomains}).  Since this property of
  \CODE{list} does not change, the entire code can be removed from the
  slice.
\end{example}

\subsection{Defining Abstract Program Slicing}
\label{sec:definingAbstractProgramSlicing}

We introduce \emph{abstract program slicing}, which compares a program
and its abstract slices by considering \emph{properties} instead of
exact values of program variables.  Such properties are represented as
abstract domains, based on the theory of \emph{Abstract
  Interpretation} (Section \ref{sec:basicAbstractInterpretation}).

We first introduce the notion of \emph{abstract slicing criterion},
where the property of interest is also specified.  For the sake of
simplicity, the definition only refers to non-$\SIM$ forms (i.e.,
$\cO$ is a singleton instead of a set of occurrences: $\cO \in \lnums
\times \wp(\NATURALS)$). In order to make abstract the criterion we
have to formalize in it the {\em properties} that we aim at observing
on program variables. In particular we could think of observing
different properties for different variables.  Hence, we define a
criterion abstraction $\cA$ defined as a tuple of abstract domains,
each one relative to a specific subset of program variables: Let $\cX$
be a set of variables of interest in $\prog$ and
$\{X_i\}_{i\in[1,k]}\subseteq\wp(\cX)$
a partition of $\cX$, the notation $\cA = \tuple{X_1: \uco_1,..,X_k:
  \uco_k}$ means that each uco $\uco_i$ is applied to the set of
variable $X_i$ (left implicit when it is clear from the context),
meaning that $\uco_i$ is precisely the property to observe on
$X_i$. In the following, we denote by $\cA_{|X_i}$ the property
observed on $X_i$, formally $\cA_{|X_i}=\uco_i$.
This is the most general representation, accounting also for
relational domains. When ucos will be applied to singletons, the
notation will be simplified ($x:\uco$ instead of $\{x\}:\uco$).

\begin{example}
  Let $x$, $y$, $z$ and $w$ be the variables in $\cX$.  Let $\cA$ be
  $\tuple{x:\PARDOM,\{y,z\}:\INTDOM^+,w:\SIGNDOM}$, meaning that the
  interest is on the parity of $x$, the sign of $w$, and the
  (relational) property of \emph{intervals} \cite{CC79} of the value
  $x+y$.  When abstracting a criterion w.r.t.~$\cA$, the required
  observation at a program state $\state$ is
  \[\PARDOM(\state(x)) \qquad 
  \INTDOM^+(\state(y)+\state(z)) \qquad \SIGNDOM(\state(w)) \]
\end{example}

In order to be as general as possible, we consider \emph{relational}
properties of variables (see Section~\ref{sec:preliminaries}), so that
properties are associated with tuples instead of single variables.  In
this case, a property is said to \emph{involve} some set (tuple) of
variables.  Given a memory $\memory$, $\uco(\memory)$ is the result of
applying $\uco$ to the values in $\memory$ of the variables involved
by the abstract domain, and $\cA(\memory)$ is the corresponding notion
for tuples of ucos.

\begin{mydefinition}[Abstract criterion]
  \label{def:AbstractCriteria} Let $\cI \subseteq \memories$ be a set
  of input memories, $\cX \subseteq \variables$ be a set of variables
  of interest; $\cO \in \lnums \times \wp(\NATURALS)$ be a set of
  occurrences of interest; $\psi$ be a truth value indicating if the
  slicing is in $\KL$ form.  Moreover, let $\cX$ be the set of
  variables of interest and $\cA \defi \tuple{X_1:\uco_1, \ldots,
    X_k:\uco_k}$, with $\{X_i\}_{i\in[1,k]}$ a partition of $\cX$.
  Then, the \emph{abstract slicing criterion} is $\crit_\cA =
  (\cI, \caX, \cO, \psi, \cA)$, 
\end{mydefinition}

Note that, when dealing with non-abstract notions of slicing, we have
that each domains is the identity on each single variable, namely $\cA
= \tuple{x_1:\IDDOM, \ldots,x_k:\IDDOM}$, where $\IDDOM \defi \lambda
x. x$.
It is also worth pointing out that, exactly as it happens for
non-abstract forms, $\cI = \memories$ corresponds to static slicing,
and $|\cI| = 1$ corresponds to dynamic slicing; in the intermediate
cases, we have conditioned slicing.

\COMMENT{
\begin{mydefinition}[Abstract slicing]
  \label{def:AbstractSlicing} Let $\prog$ and $\progq$ be executable
  programs such that $\progq$ is obtained from $\prog$ by removing
  zero or more statements, and let $\crit_\cA =
  (\cI,\caX,\cO,\psi,\cA)$, with $\cA \defi \tuple{X_1:\uco_1, \ldots, X_k:\uco_k}$.  
  $\progq$ is an \emph{abstract slice} of
  $\prog$ with respect to $\crit_\cA$ if (1) for each
  $\memory \in \cI$, when the execution of $\prog$ from input
  $\memory$ reaches a point (or occurrence) in $\cO$, the execution of
  $\progq$ from $\memory$ reaches $\cO$ as well; and (2) for each $i \in [1,k]$, each $X_i$
  has the same property $\uco_i$ both in $\prog$ and in $\progq$,
  i.e., $\uco_i(\memory'_i) = \uco_i(\memory''_i)$ where $\memory'_i$
  and $\memory''_i$ are the current memories of $\prog$ and $\progq$,
  respectively.  Moreover, if $\psi = \true$, then the executions of
  $\prog$ and $\progq$ have to follow identical paths as explained in
  Definition \ref{def:UnifiedEquivalence}.
\end{mydefinition}}

\subsection{The extended formal framework}
\label{subsection:theFormalFramework}

In this section, we extend a formal framework in which all forms of abstract
slicing can be formally represented.  It is an extension of the
mathematical structure introduced by Binkley.  Following their framework, we represent a
form of abstract slicing by a pair $(\sqsubseteq, \cE_{\cA})$, where
$\sqsubseteq$ is the traditional syntactic ordering, and $\cE_{\cA}$
is a function mapping abstract slicing criteria to semantic
equivalence relations on programs.  Given two programs $\prog$ and
$\progq$, and an abstract slicing criterion $\crit_\cA$, we say that
$\progq$ is a $(\sqsubseteq, \cE_{\cA})$-\emph{(abstract)-slice} of
$\prog$ with respect to $\crit_\cA$ iff $\progq \sqsubseteq \prog$ and
$\tuple{\prog,\progq} \in \cE_{\cA}(\crit_\cA)$ (i.e., $\prog$ and $\progq$
are equivalent w.r.t.~$\cE_{\cA}$).  Some preliminary notions are
needed to define $\cE_{\cA}$ in the context of abstract slicing.

An \emph{abstract memory} w.r.t.\ a set of variables of interests
$\cX$ (partitioned in $\{X_i\}_{i\in[1,k]}$) is obtained from a memory
by restricting its domain to the variables of interest, and assigning
to each set $X_i$ of variables an abstract value determined by the
corresponding abstract property of interest $\uco_i$.

\begin{mydefinition}
  \label{def:AbstractStateRestriction} Let $\memory\!\in\!\memories$
  be a memory, $\cX$ be the set of 
  a tuple of sets of variables of interest, and $\cA =
  \tuple{X_1:\uco_{1},\ldots,X_k:\uco_{k}}$ be the corresponding tuple
  of properties of interest such that $\{X_i\}_{i\in[1,k]}$ is a
  partition of $\cX$.  The \emph{abstract restriction} of a memory
  $\memory$ w.r.t.~the state abstraction $\cA$ is defined as
  $\ABSMEM{\memory}{\caX}{\cA} \defi \cA \circ \memory(\caX) \defi
  \tuple{\uco_{1}(\memory(X_{1})),\ldots, \uco_{k}(\memory(X_{k}))}$.
\end{mydefinition}

\begin{example}
  \label{ex:Var}
  Let $\variables = \{x_1, x_2, x_3, x_4\}$ be a set of variables, and
  suppose that the properties of interest are the (relational) sign of the product
  $x_1 x_2$ and the parity of $x_3$ (both defined in
  Section~\ref{sec:preliminaries}).
  We slightly abuse notation by denoting as $\SIGNDOM$
  also its extension to pairs $(v,t)$ where the sign of their product
  matters: e.g., $\SIGNDOM(3,-5)=(\NEG)$.  In our formal
  framework, $\cA$ is defined as
  $\tuple{\{x_1,x_2\}:\SIGNDOM, x_3:\PARDOM}$.  Let
  $\memory(x_1) = 1$, $\memory(x_2) = 2$, $\memory(x_3) = 3$, and
  $\memory(x_4) = 4$; then, $\ABSMEM{\memory}{\caX}{\cA}$ comes to be
  $\cA \circ \memory(\caX) = \langle \POS, \ODD \rangle$.
\end{example}

The \emph{abstract projection} operator modifies a state trajectory by
removing all those states which do not contain occurrences or points
of interest.  If there is a state that contains an occurrence of
interest, then its memory state is restricted via $\ABSMEMop$ to the
variables of interest, and only a property is considered for each
tuple.  In the following, the \emph{abstract projection}
$\Proj^\alpha$ is formally defined.

\begin{mydefinition}[Abstract Projection]
  \label{def:ProjA}
  Let $\crit_\cA=(\cI,\caX,\cO,\psi,\cA)$, and $\cL\subseteq\lnums$
  such that $\cL\neq\emptyset$ if $\psi=\true$, $\cL=\emptyset$
  otherwise. For any $n \in \lnums$, $k \in \NATURALS$, $\memory \in
  \memories$, we define a function $\Proj^{0\alpha}$ as:
  \[ \Proj^{0\alpha}_{(\caX, \cO, \cL,
    \cA)}(n^k, \memory) \defi \left\{
  \begin{array}{l l}
    \tuple{n^k, \ABSMEM{\memory}{\caX}{\cA}} & \mbox{if }
      \exists\tuple{n,K}\in \cO.\:k \in K \\ \tuple{n^k,\bot}
      & \mbox{if }
      \nexists(n,K) \in \cO.\:k \in K\mbox{ and }n\in \cL \\
    \varepsilon &
    \mbox{otherwise}
  \end{array}
  \right.
  \]
  The abstract projection $\Proj^{\alpha}$ is the extension of
  $\Proj^{0\alpha}$ to sequences:
  \[ \begin{array}{l}
    \Proj_{(\caX, \cO, \cL, \cA)}^{\alpha} (\tuple{(n_1^{k_1},\memory_1) \ldots 
    (n_l^{k_l}, \memory_l)})  =\\
    \qquad\qquad\qquad\Proj^{0\alpha}_{(\caX, \cO, \cL, \cA)}(n_1^{k_1},
    \memory_1) \circ \ldots \circ {\Proj^{0\alpha}}_{(\caX, \cO, \cL,
        \cA)}(n_l^{k_l}, \memory_l)
  \end{array}
  \]
\end{mydefinition}

\noindent
$\Proj^{0\alpha}$ takes a state from a state trajectory, and returns
either one pair or an empty sequence $\varepsilon$.  Abstract
projection allows us to define all the semantic equivalence relations
we need for representing the abstract forms of slicing.

\begin{example}
  \label{absSli} Consider the program $\prog$ in
  Figure~\ref{fig:exsl}.
\begin{figure}
    \begin{center}
      {\small
        \begin{tabular}{c|c}
          \begin{lstlisting}
  read(n);
  i := 1; 
  s := 0; 
  p := 1;
  while (i <= n) {
    s := s+i;
    p := p*i;
    i := i+1; }
  write(i,n,s,p);
          \end{lstlisting} 
          &
          \begin{lstlisting}
  read(n);
  i := 1; 
  s := 0; 
  
  while (i <= n) {
    s := s+i;
    
    i := i+1; }
  write(i,n,s);
          \end{lstlisting} 
          \\
          &\\%&&\\
          Program $\prog$ & Program $\progq$ 
      \end{tabular}}
    \end{center}
    
    \caption{Program $\prog$ and its slice.}\label{fig:exsl}
  \end{figure}  
  Consider $\cI=\tM$ (meaning that we are considering static slicing),
  $\caX=\{\ii,\ss\}$, $\cO=\tuple{8,\NATURALS}$ (meaning that we check
  the value of variables of interest at each iteration of program
  point $8$).  Moreover, we consider $\cA=\tuple{\mbox{\ii}:\SIGNDOM,
    \mbox{\ss}:\PARDOM}$.  Then in Figure~\ref{absFig} we have the
  corresponding abstract projection (the concrete trace is given in
  the Appendix, Example~\ref{ExTrace}).  In this figure, we depict
  states as set of boxes, the first one contains the number of the
  executed program point (with the iteration counter as apex), while
  the other boxes are the different variables associations. The cross
  on a box means that the projection does not consider that variable
  or state. So for instance, in this example we care only of states
  $6^i$ and $8^i$, and in particular, in states $6^i$ we are not
  interested in the values of variables, while in states $8^i$ we are
  interested in the sign of $i$ and in the parity of $s$ (if we would
  be interested in the value of these variables we would have the
  value instead of their property, as it happens in the examples in
  the Appendix).
  
  \begin{figure}
    \begin{center}
      \includegraphics[scale=.4]{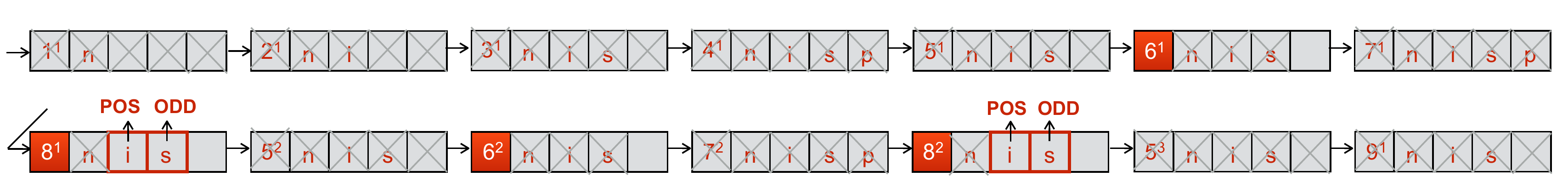}
    \end{center}
    \caption{Abstract trajectory projection for program $\prog$ in
      Example~\ref{absSli}}
    \label{absFig}
  \end{figure}
\end{example}

\subsection{Abstract Unified Equivalence}
\label{subsection:AbstractUnifiedEquivalence}

The only missing step for completing the formal definition of abstract
slicing in the formal framework is to characterize the functions
mapping abstract slicing criteria to abstract semantic equivalence
relations.

\begin{mydefinition}[Abstract Unifying Equivalence]
  \label{def:AbstractUnifiedEquivalence} Let $\prog$ and $\progq$ be
  executable programs, and $\crit_\cA=(\cI,\caX,\cO,\psi,\cA)$ be an
  abstract criterion.  Then $\prog$ is \emph{abstract-equivalent} to
  $\progq$ if and only if, for every $\memory\in \cI$, it holds that
  $\Proj^{\alpha}_{(\caX, \cO, \cL,\cA)}(\trace_\prog^\memory)
  = \Proj^{\alpha}_{(\caX,\cO,\cL,\cA)}(\trace_\progq^\memory)$, where
  $\cL=\lnums_\prog\cap \lnums_\progq$ if $\psi=\true$.  The
  function $\cE_\cA$ maps each criterion $\crit_{\cA}$ to a
  corresponding abstract semantic equivalence relation.
\end{mydefinition}

Therefore, a generic form of slicing can be represented as
$(\sqsubseteq,\cE_{\cA})$.  This can be used to formally define both
traditional and abstract forms of slicing in the presented abstract
formal framework, so that the latter comes to be a generalization of
the original formal framework.  The following examples show how it is
possible to use these definitions in order to check whether a program
is an abstract slice of another one.

\begin{figure}
  \begin{center}
    \begin{tabular}{c@{\quad}|@{\quad}c}
      \begin{lstlisting}
   read(n);
   read(s);
   i := 1;
   while (i<=n) {
     s := s+2*i;
     i := i+1;  }
   write(i,n,s);
      \end{lstlisting}
      &
      \begin{lstlisting}
   read(n);
   read(s);
   
   write(n,s);
      \end{lstlisting}
    \end{tabular}
  \end{center}
  \caption{Programs \prog and \progq}\label{fig:PandQ}
\end{figure}

\begin{example}
  Consider the programs $\prog$ and $\progq$ in
  Figure~\ref{fig:PandQ}. Let $\crit_{\cA} = (\memories, \{\mbox{\ss}
  \},\{\tuple{7,\NATURALS}\}, \false,
  \tuple{\mbox{\ss}:\uco_{\textsc{par}}})$, meaning that we are
  interested in the parity of \ss
  ($\cA=\tuple{\mbox{\ss}:\uco_{\textsc{par}}}$) at the end of
  execution ($\cO=\{\tuple{7,\NATURALS}\}$) for all possible inputs
  ($\cI = \memories$), in non-$\KL$ form.  Since $\progq \sqsubseteq
  \prog$, in order to show that $\progq$ is an abstract static slice
  of $\prog$ with respect to $\crit_{\cA}$, we have to show that
  $\tuple{\prog,\progq} \in \cE_{\cA}(\crit_{\cA})$ holds.  Let
  $\memory = \{\BIND{n}{a}, \BIND{s}{b}\}$ for some $a,b\in\NATURALS$
  be an initial memory.  The trajectory of $\prog$ from $\memory$
  contains the following steps of computation:
\begin{center}
\includegraphics[scale=.4]{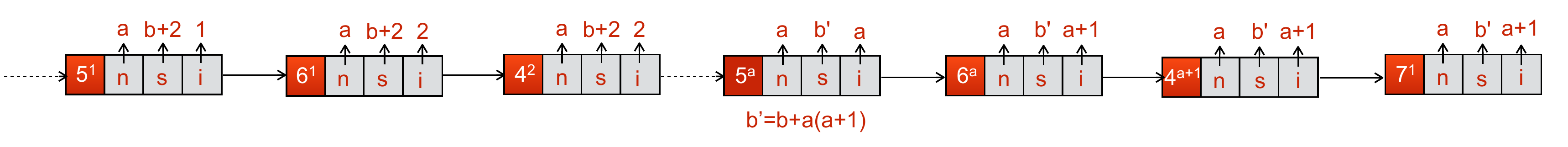}
\end{center}
  Applying $\Proj^\alpha$ (with $\cL=\emptyset$ since $\psi=\false$)
  to $\trace_\prog^\memory$ returns only the abstract value of the
  variable \ss at point $7$ (due to $\crit_\cA$):
  \[
  \begin{array}{rcl}
    \Proj^\alpha_{(\caX,\cO,\emptyset,\cA)}(\trace_\prog^\memory) & =
    &
    \tuple{7^1, \ABSMEM{\{ \BIND{n}{a}, \BIND{s}{b{+}a(a{+}1)},
      \BIND{i}{a{+}1} \}}{\{\langle s \rangle\}}{\cA})} \\ & = & \tuple{7^1,
    \PARDOM(b{+}a(a{+}1))}  =  \tuple{7^1, \PARDOM(b)}
  \end{array}
  \]
  \noindent Since we have $7^1\!=\!\tuple{7, 1}\!\in\!\cO$,
  $\Proj^{0\alpha}_{(\caX,\cO,\emptyset,\cA)}(7^{1}, \{\BIND{n}{a},
  \BIND{s}{b{+}a(a{+}1)}, \BIND{i}{a{+}1} \})$ returns
  $\tuple{7^1,\ABSMEM{\memory}{\caX}{\cA}} = \tuple{7^1, \ABSMEM{\{
      \BIND{n}{a}, \BIND{s}{b{+}a(a{+}1)}, \BIND{i}{a{+}1}
      \}}{\{\langle s \rangle\}}{\cA}}$.  The abstract memory
  restricts the domain of $\memory$ to variables of interest, so we
  consider only the part of $\memory$ regarding $s$, i.e.,
  $b\!+\!a(a\!+\!1)$.  Hence, we have $\tuple{7^1, \PARDOM(\langle
    b\!+\!a(a\!+\!1)\rangle)}$.  Since the parity of $b{+}a(a{+}1)$
  only depends on the parity of $b$, being either $a$ or $a+1$ even,
  the final result is $\tuple{7^1, \PARDOM(b)}$.  Consider now the
  execution of $\progq$ from $\memory$, which corresponds to the
  following state trajectory:
  \begin{center}
\includegraphics[scale=.4]{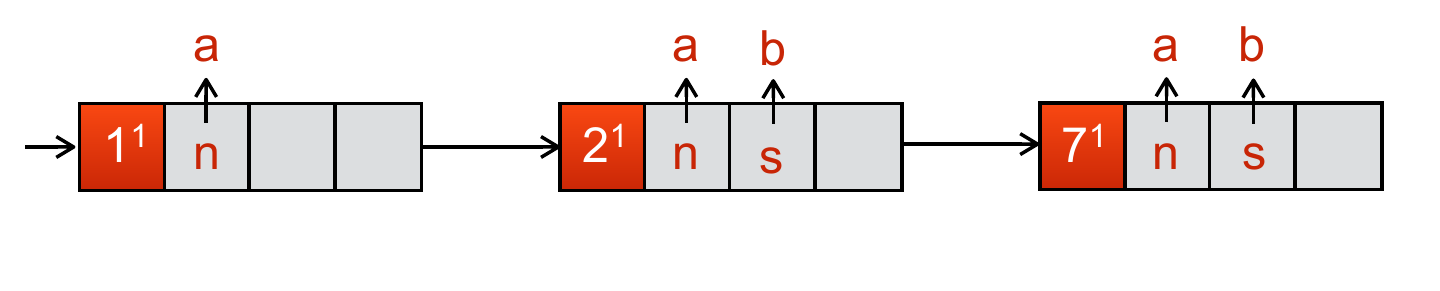}
\end{center}
  Applying $\Proj^\alpha$ to $\trace_\progq^\memory$
  gives: \[ \begin{array}{rcl}
    \Proj^\alpha_{(\caX,\cO,\emptyset,\cA)}(\trace_\progq^\memory)
   =\tuple{7^1, \ABSMEM{\{\BIND{n}{a}, \BIND{s}{b}\}}{\{\langle s
       \rangle\}}{\cA}} = \tuple{7^1, \PARDOM(b)} \end{array} \]
  Therefore,
  $\Proj^\alpha_{(\caX,\cO,\emptyset,\cA)}(\trace_\prog^\memory)$ is
  equal to
  $\Proj^\alpha_{(\caX,\cO,\emptyset,\cA)}(\trace_\progq^\memory)$.
  As $\memory$ is an arbitrary input, this equation holds for each
  $\memory \in \memories$, so that $\tuple{\prog,\progq} \in
  \cE_{\cA}(\crit_{\cA})$, and this implies that $\progq$ is an
  abstract static slice of $\prog$ w.r.t.~$\crit_{\cA}$.
\end{example}

\begin{figure}
  \begin{center}
    \begin{tabular}[t]{c@{\quad}|@{\quad}c}
      \begin{lstlisting}
   read(n);
   s := 0;
   i := 1;
   while (i<=n) {
     s := s+i;
     i := i+1; }
   write(i,n,s);
      \end{lstlisting}
      &
      \begin{lstlisting}
   read(n);
   s := 0;
   
   write(n,s);
      \end{lstlisting}
    \end{tabular}
  \end{center}
  \caption{Programs \progr and \progss}\label{fig:RandS}
\end{figure}
\begin{example}
  Consider the programs $\progr$ and $\progss$ in
  Figure~\ref{fig:RandS}, and let $\crit_{\cA}$ be $(\cI,
  \{\mbox{\ss}\},\{\tuple{7,\NATURALS}\},\false,\tuple{\ss:\PARDOM})$,
  where $\cI =
  \{~\memory\mid\memory(\mbox{\nn})\!\in\!4\mathbb{Z}~\}$; i.e., we
  are interested in the parity of \ss at the end of the execution for
  all inputs where \nn is a multiple of $4$.  Since
  $\progss{\sqsubseteq}\progr$, in order to show that $\progss$ is an
  abstract conditioned slice of $\progr$ w.r.t.~$\crit_{\cA}$, we have
  to show that $\tuple{\progr,\progss} \in \cE_{\cA}(\crit_{\cA})$
  holds, namely that they have the same abstract projection.  Let
  $\memory \in \cI$ be an initial memory, and suppose
  $\memory(\mbox{\nn}) = a = 4m$.  The trajectory
  $\trace_\progr^\memory$ of $\progr$ from $\memory$ contains the
  following computations:
   \begin{center}
\includegraphics[scale=.4]{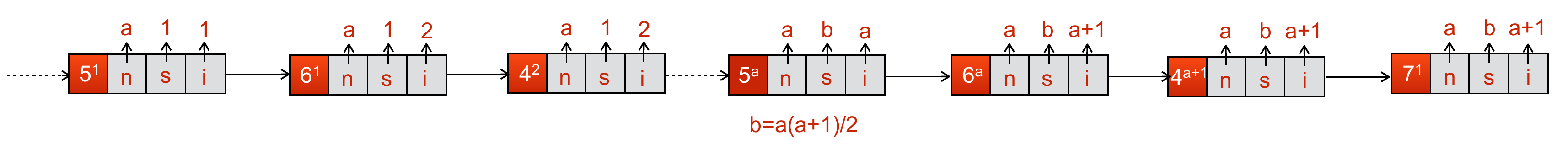}
\end{center}
While executing $\progss$ from $\memory$ gives the state
  trajectory $\trace_\progss^\memory$
   \begin{center}
\includegraphics[scale=.4]{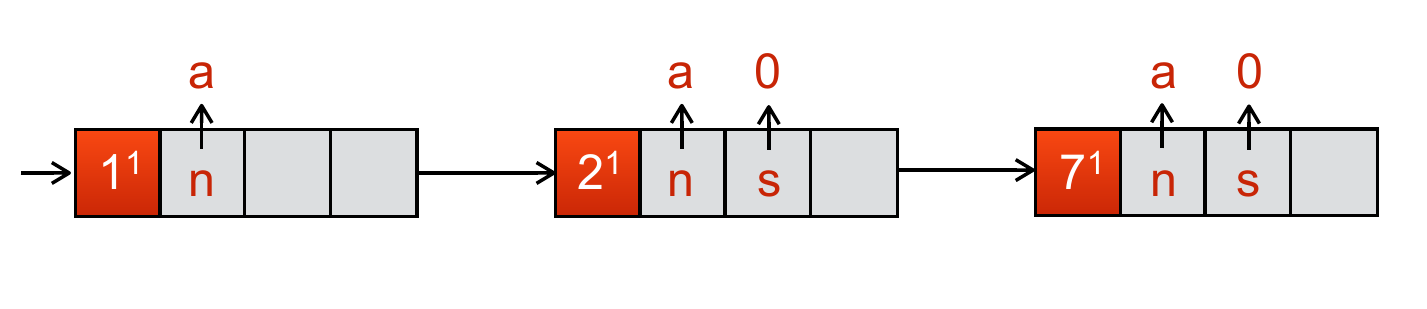}
\end{center}
  Applying $\Proj^\alpha$ to both state trajectories we have:
  \[
  \begin{array}{rcl}
    \Proj^\alpha_{(\caX,\cO,\emptyset,\cA)}(\trace_\progr^\memory) 
    & = &\tuple{7^1, \ABSMEM{\{ \BIND{n}{a}, \BIND{s}{a(a{+}1)/2},
      \BIND{i}{a{+}1} \}}{\tuple{\{s\}}}{\cA}} 
     =  \tuple{7^1, \PARDOM(a(a\!+\!1)/2)} \\
     &=&  \tuple{7^1, \PARDOM(2m(4m{+}1))} = \tuple{7^1, 2\mathbb{Z}} \\
     &=& \tuple{7^1, \PARDOM(0)} =\tuple{7^1, \ABSMEM{\{ \BIND{n}{a}, \BIND{s}{0}
    \}}{\tuple{\{s\}}}{\cA}} 
    =\Proj^\alpha_{(\caX,\cO,\emptyset,\cA)}(\trace_\progss^\memory)
  \end{array}
  \]
  Therefore, we have
  $\Proj^\alpha_{(\caX,\cO,\emptyset,\cA)}(\trace_\progr^\memory) =
  \Proj^\alpha_{(\caX,\cO,\emptyset,\cA)}(\trace_\progss^\memory)$ .
  As $\memory$ is an arbitrary input from $\cI$, this equation holds
  for each $\memory \in \cI$, so that $\tuple{\progr,\progss} \in
  \cE_{\cA}(\crit_{\cA})$, and this implies that $\progss$ is an
  abstract conditioned slice of $\progr$ w.r.t.~$\crit_{\cA}$. It is
  worth noting, that $\progss$ is not a {\em static} abstract slice of
  $\progr$ since for all the input values $a\notin 4\mathbb{Z}$ for
  $n$ the parity of the final value of \ss is not necessarily even.
\end{example}

\subsection{Comparing forms of Abstract Slicing}
\label{subsection:Comparing}
 % comparingFormsOfAbstractSlicing
This section provides a formal theory allowing us to compare abstract
forms of slicing between themselves, and with non-abstract ones.
First of all, we show under which conditions an abstract semantic
equivalence relation \emph{subsumes} another one; analogously, we show
when the form of (\emph{abstract}) slicing, corresponding to the
former equivalence relation, subsumes the form of (\emph{abstract})
slicing corresponding to the latter one.  Such results are necessary
in order to obtain a precise characterization of the extension of the
\emph{weaker than} relation (whose original definition is recalled in
the Appendix) to the abstract forms of slicing.

\COMMENT{
  \begin{figure}[tbp]
    \centering
    \includegraphics[scale=.65,viewport=3.7in 6.7in 5.5in 8.45in]{SlicingTutto1.pdf}
    \caption{ \label{fig:SlcTto}}
  \end{figure}}

The following lemma shows under which conditions on the slicing
criteria there is a relation of subsumption between two semantic
equivalence relations. In the following, we denote
$\overline{\sqsubseteq}$ the relation "more concrete than" in the
lattice of abstract interpretations between tuples of abstractions.
Formally, let us consider
$\cA^1=\tuple{X_1^1:\uco_1^1,\ldots,X_{k^1}^1:\uco_{k^1}^1}$ defined
on the variables $\cX^1$ and
$\cA^2=\tuple{X_1^2:\uco_1^2,\ldots,X_{k^2}^2:\uco_{k^2}^2}$ defined
on the variables $\cX^2$, such that $\cX^1\subseteq\cX^2$, $k^1\leq
k^2$ and $\forall i\leq k^1$ we have $X^1_i=X^2_i$, namely the
variables in common are partitioned in the same way.  Then
$\cA^{2}\overline{\sqsubseteq}\cA^{1}$ iff $\forall
X_i\in\cX^1.\:\uco_i^2\sqsubseteq\uco^1_i$. Note that, for all the
variables in $\ok{\cX^2\smallsetminus\cX^1}$, the abstraction
$\ok{\cA^1}$ does not require any particular observation, hence on
these variables surely $\ok{\cA^2}$ is more precise. The following
relation is such that, when both $\cA^1$ and $\cA^2$ are the identity
on all the variables of interest, then the resulting criterion
relation is the same proposed in \cite{TheoFoun} (see the Appendix for
details) for characterizing the original formal framework.

\begin{lemma}
  \label{the:AbstractUnifiedEquivalenceLemma}
  Let two abstract slicing criteria $\crit_{\cA}^1 =
  (\cI^1,\caX^1,\cO^1,\psi^1,\cA^1)$ and
  $\crit_{\cA}^2=(\cI^2,\caX^2,\cO^2,\psi^2,\cA^2)$ be given.  If (1)
  $\cI^{1}\subseteq \cI^{2}$; (2) $\cO^{1}\subseteq\cO^{2}$; (3)
  $\cX^{1}\subseteq\cX^{2}$; (4) $\psi^1\Ra\psi^2$;
  and (5) $\cA^{2}\overline{\sqsubseteq}\cA^{1}$ (denoted $\crit_\cA^1\ra_{\mbox{\tiny $\cA$}}\crit^2_\cA$), then 
  $(\sqsubseteq,\cE(\crit_{\cA}^1))$ \emph{subsumes} $(\sqsubseteq,\cE(\crit_{\cA}^2))$, i.e., for
  every $\prog$ and $\progq$ such that $\progq\sqsubseteq\prog$, $\tuple{\prog,\progq} \in \cE(\crit_{\cA}^2)$
  implies $\tuple{\prog,\progq} \in \cE(\crit_{\cA}^1)$.
\end{lemma}

\begin{proof}
  First of all, note that, if $\cA^1=\cA^2=\uco_{\iid}$, namely if we
  are considering concrete criteria, then $\ra_{\mbox{\tiny $\cA$}}$
  collapses to the concrete relation defined in \cite{AForm} (which is
  the $\ra$ defined in Equation~\ref{eq:defrelcrit} in the
  Appendix). Hence, in this case, the results holds by
  \cite{AForm}.\\ Suppose $\tuple{\prog,\progq} \in
  \cE(\crit_{\cA}^2)$ with $\progq\sqsubseteq \prog$, namely $\progq$
  slice of $\prog$ w.r.t.\ $\cE(\crit_{\cA}^2)$.  This means that, for
  each $\memory_0\in \cI^2$ $\Proj^{\alpha}_{(\cX^2, \cO^2,
    \cL^2,\cA^2)}(\trace_\prog^{\memory_0}) =
  \Proj^{\alpha}_{(\cX^2,\cO^2,\cL^2,\cA^2)}(\trace_\progq^{\memory_0})$,
  where $\cL^2$ is defined as in
  Definition~\ref{def:AbstractUnifiedEquivalence}.  This means that,
  for each state $\tuple{n^{k},\memory}$ in the trajectory
  $\trace_\prog^{\memory_0}$, whose projection
  $\Proj^{0\alpha}_{(\caX^2, \cO^2, \cL^2, \cA^2)}(n^k,
  \memory)_\prog$\footnote{The notation $\Proj^{0\alpha}_{(\caX^2,
      \cO^2, \cL^2,\cA^2)}(n^k, \memory)_\prog$ means that we are
    projecting a state of the computation of $\prog$.}  is not empty,
  there exists a state $\tuple{n^{k},\memory'}$ in
  $\trace_\progq^{\memory_0}$ with the same projection.  Let us
  consider now, $\memory\in\cI^1\subseteq\cI^2$. We prove that on
  these states $\Proj^{0\alpha}_{(\caX^1, \cO^1, \cL^1, \cA^1)}(n^k,
  \memory)=\Proj^{0\alpha}_{(\caX^2, \cO^2, \cL^2, \cA^2)}(n^k,
  \memory)$ (and in this case there is a corresponding state in the
  trajectory of $\progq$), or it is empty (and in this case also the
  state in $\progq$ has empty projection).  Recall that
  \[ \Proj^{0\alpha}_{(\caX^1, \cO^1, \cL^1,
    \cA^1)}(n^k, \memory) \defi \left\{
  \begin{array}{l l}
    \tuple{n^k, \ABSMEM{\memory}{\caX^1}{\cA^1}} & \mbox{if }
      \exists\tuple{n,K}\in \cO^1.\:k \in K \\ \tuple{n^k,\bot}
      & \mbox{if }
      \nexists(n,K) \in \cO^1.\:k \in K\mbox{ and }n\in \cL^1 \\
    \varepsilon &
    \mbox{otherwise}
  \end{array}
  \right.
  \]
  where also $\cL^1$ is defined as in
  Definition~\ref{def:AbstractUnifiedEquivalence}.  Note that, since
  we are considering both the criteria on the same pair of programs,
  we have also that $\psi^1\Ra\psi^2$ corresponds to saying that
  $\cL^1\subseteq\cL^2$. At this point
  \begin{itemize}
  \item[$\bullet$] If $\Proj^{0\alpha}_{(\caX^1, \cO^1, \cL^1,
    \cA^1)}(n^k, \memory)_\prog=\tuple{n^k,
    \ABSMEM{\memory}{\caX^1}{\cA^1}}$ then $\exists\tuple{n,K}\in
    \cO^1.\:k \in K$, but $\cO^1\subseteq\cO^2$, hence $\tuple{n,K}\in
    \cO^2.\:k \in K$. This mean that $\Proj^{0\alpha}_{(\caX^2, \cO^2,
      \cL^2, \cA^2)}(n^k, \memory)_\prog=\tuple{n^k,
      \ABSMEM{\memory}{\caX^2}{\cA^2}}$, which by hypothesis is equal
    to $\Proj^{0\alpha}_{(\caX^2, \cO^2, \cL^2,\cA^2)}(n^k,
    \memory')_\progq$, for a memory $\memory'$. By definition and
    hypothesis, $\ABSMEM{\memory}{\caX^2}{\cA^2}=\cA^2 \circ
    \memory(\caX)=\cA^2 \circ \memory'(\caX)$. Namely,
    $\tuple{\uco_{1}^2(\memory(X_{1})),\ldots,
      \uco_{k^2}^{2}(\memory(X_{k^2}))}=\tuple{\uco_{1}^2(\memory'(X_{1})),\ldots,
      \uco_{k^2}^2(\memory'(X_{k^2}))}$.  Therefore, in particular,
    $\forall i\in[1,k^1]\subseteq[1,k^2]$ we have
    $\uco_i^2(\memory(X_i))=\uco_i^2(\memory'(X_i))$, but by
    hypothesis $\uco^2_i\sqsubseteq\uco^1_i$, hence we also have
    $\uco_i^1(\memory(X_i))=\uco_i^1(\memory'(X_i))$ (by properties of
    ucos). But then $\tuple{\uco_{1}^1(\memory(X_{1})),\ldots,
      \uco_{k^1}^{1}(\memory(X_{k^1}))}=\tuple{\uco_{1}^1(\memory'(X_{1})),\ldots,
      \uco_{k^1}^1(\memory'(X_{k^1}))}$, namely $\cA^1 \circ
    \memory(\caX)=\cA^1 \circ \memory'(\caX)$. Hence
    \[
    \begin{array}{lll}
    \Proj^{0\alpha}_{(\caX^1, \cO^1, \cL^1,
    \cA^1)}(n^k, \memory)_\prog&=&
    \tuple{n^k,\ABSMEM{\memory}{\caX^1}{\cA^1}}\\
    &=&\tuple{n^k,\cA^1 \circ \memory(\caX)}=\tuple{n^k,\cA^1 \circ \memory'(\caX)}\\
    &=&\tuple{n^k,\ABSMEM{\memory'}{\caX^1}{\cA^1}}=
        \Proj^{0\alpha}_{(\caX^1, \cO^1, \cL^1,
    \cA^1)}(n^k, \memory')_\progq
    \end{array}
    \]
  \item[$\bullet$] If $\Proj^{0\alpha}_{(\caX^1, \cO^1, \cL^1,
    \cA^1)}(n^k, \memory)_\prog=\tuple{n^k,\bot}$ then $\nexists(n,K)
    \in \cO^1.\:k \in K\mbox{ and }n\in \cL^1$. If $\nexists(n,K) \in
    \cO^2.\:k \in K$ then $n\in\cL^1\subseteq\cL^2$, then also
    $\Proj^{0\alpha}_{(\caX^2, \cO^2, \cL^2, \cA^2)}(n^k,
    \memory)_\prog=\tuple{n^k,\bot}$ but then by hypothesis we have
    that there exists a memory $\memory'$ such that
    $\Proj^{0\alpha}_{(\caX^2, \cO^2, \cL^2, \cA^2)}(n^k,
    \memory')_\progq=\tuple{n^k,\bot}$. But then we also have
    $\Proj^{0\alpha}_{(\caX^1, \cO^1, \cL^1, \cA^1)}(n^k,
    \memory)_\progq=\tuple{n^k,\bot}$.\\ If $\exists(n,K) \in
    \cO^2.\:k \in K$, then $\Proj^{0\alpha}_{(\caX^2, \cO^2, \cL^2,
      \cA^2)}(n^k,
    \memory)_\prog=\tuple{n^k,\ABSMEM{\memory}{\caX^2}{\cA^2}}$, but
    then there exists $\memory'$ such that also in $\progq$ we have
    $\Proj^{0\alpha}_{(\caX^2, \cO^2, \cL^2, \cA^2)}(n^k,
    \memory')_\progq=\tuple{n^k,\ABSMEM{\memory'}{\caX^2}{\cA^2}}$. But
    then, the same memory, in $\crit_\cA^1$ keep the program point but
    loses the state observation because $\nexists(n,K) \in \cO^1.\:k
    \in K\mbox{ and }n\in \cL^1$, hence $\Proj^{0\alpha}_{(\caX^1,
      \cO^1, \cL^1, \cA^1)}(n^k, \memory')_\progq=\tuple{n^k,\bot}$.
   \item[$\bullet$] Finally, if $\Proj^{0\alpha}_{(\caX^1, \cO^1,
     \cL^1, \cA^1)}(n^k, \memory)_\prog= \varepsilon$, then
     $\nexists(n,K) \in \cO^1.\:k \in K\mbox{ and }n\notin \cL^1$. But
     this means that, even if there exists $\memory'$ such that we
     have the state $\tuple{n^k,\memory'}$ in the trajectory of
     $\progq$, also in this case we have $\Proj^{0\alpha}_{(\caX^1,
       \cO^1, \cL^1, \cA^1)}(n^k, \memory')_\progq= \varepsilon$.
  \end{itemize}
\end{proof}

\COMMENT{
\begin{lemma}
  \label{the:AbstractUnifiedEquivalenceLemma}
  \defthename{$\UEA$-lemma} Let $\SSS_1$ and $\SSS_2$ be sets of
  initial memory states, $\VV_1$ and $\VV_2$ be sets of variables,
  $\PP_1$ and $\PP_2$ be sets of occurrences, $\XX_1$ and $\XX_2$ be
  functions from pairs of sets of line numbers to sets of line
  numbers, and $\AAA_1$ and $\AAA_2$ be partial functions from line
  numbers and variables to properties on variables, such that $\SSS_1
  \subseteq \SSS_2$, $\PP_1 \subseteq \PP_2$, $\VV_1 \subseteq \VV_2$,
  $\forall P, Q.\:\XX_1(\overline{P}, \overline{Q}) \subseteq
  \XX_2(\overline{P}, \overline{Q})$ and $\forall i \in C_1, x \in
  \VV_1.\:\AAA_2(i, x) \overline{\sqsubseteq} \AAA_1(i, x)$, where
  $C_1 = \{n \in \LLL \mid \exists k \in \mathbb{N}.(n, k) \in
  \PP_1\}$, and $\overline{\sqsubseteq}$ is the relation "more
  concrete than" in the lattice of abstract interpretations
  \footnote{Given the abstractions $\phi_{1}$ and $\phi_{2}$ modeled
    as \uco{}s, $\phi_{1}\overline{\sqsubseteq}\phi_{2}$ iff $\forall
    x.\:\phi_{1}(x)\leq\phi_{2}(x)$.}, then $P\quad\!\!\!\UEA(\SSS_2,
  \VV_2, \PP_2, \XX_2, \AAA_2)\quad\!\!\!Q \Rightarrow
  P\quad\!\!\!\UEA(\SSS_1, \VV_1, \PP_1, \XX_1, \AAA_1)\quad\!\!\!Q$.
\end{lemma}}

This lemma tells us how it is possible to find the relationship (in
the sense of subsumption) between two semantic equivalence relations
determined by two abstract slicing criteria.  In the following,
abstract notions of slicing will be denoted by adding an $\cA$; e.g.,
$\mathcal{AS}$ denotes static abstract slicing, whereas $\mathcal{AD}$
denotes dynamic abstract slicing.  By using this lemma we can show
that, given a slicing criterion $\crit_{\cA}$, all the abstract
equivalence relations introduced in
Sec.~\ref{subsection:AbstractUnifiedEquivalence} subsume the
corresponding non-abstract equivalence relations
$\mathcal{S}(\crit_{\cA})$, $\mathcal{D}(\crit_{\cA})$ and
$\mathcal{C}(\crit_{\cA})$.  Furthermore, by using this lemma we can
show that $\mathcal{AD}(\cC_{\cA})$ subsumes
$\mathcal{AC}(\cC_{\cA})$, which in turns subsumes
$\mathcal{AS}(\cC_{\cA})$.

\begin{theorem}
  \cite{TheoFoun} Let $\cR_{1}$ and $\cR_{2}$ be semantic equivalence
  relations such that $\cR_{2}$ subsumes $\cR_{1}$.  Then, for every
  $\prog$ and $\progq$, we have $\tuple{\prog,\progq} \in
  (\sqsubseteq,\cR_{1})~ \Ra~\tuple{\prog,\progq} \in
  (\sqsubseteq,\cR_{2})$.
\end{theorem}

\begin{figure}[htbp]
 \centering
    \includegraphics[scale=.6,viewport=3in 6.35in 4.5in 9in]{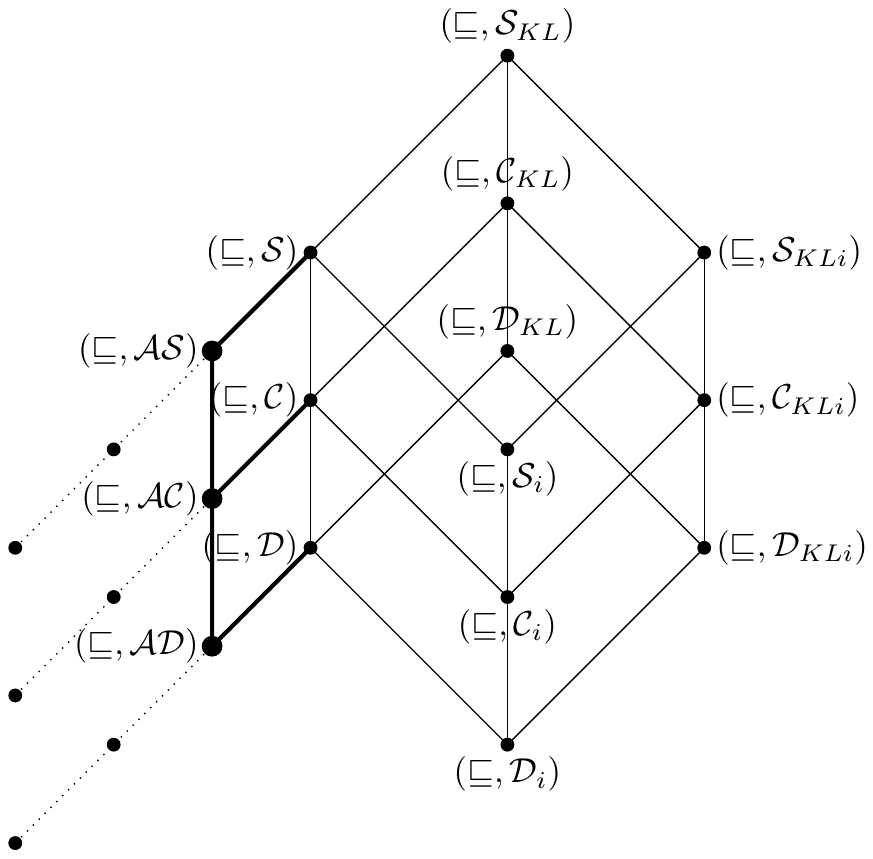}
  \caption{Extended hierarchy. \label{fig:Ret2}}
\end{figure}
Fig.~\ref{fig:Ret2} shows the non-\SIM hierarchy obtained by enriching
the hierarchy in Fig.~\ref{fig:Ret1} with standard forms of
\absstcslcg, \absdynslcg{}, and \absconslcg.  In general, we can
enrich this hierarchy with any abstract form of slicing simply by
using the comparison notions defined above.  Non-abstract forms are
particular cases of abstract forms of slicing, as they can be
instantiated by choosing the identity property, $\IDDOM$, for each
variable of interest.  Hence, non-abstract forms are the "strongest"
forms, since, for each property $\uco$, we have $\IDDOM \sqsubseteq
\uco$.  Moreover, if parameters $M, \caX, \cO, \psi$ are fixed, and
$\cA$ is made less precise or more abstract (i.e., the information
represented by the property is reduced), then the abstract slicing
form becomes weaker, as suggested by dotted lines in Figure
\ref{fig:Ret2}.

\section{Program Slicing and Dependencies}
\label{subsection:dependencies}
 % slicingVsDependencies
In the previous sections we introduced a formal framework of different
notions of program slicing.  In particular, we observed that a kind of
slicing is a pair: a syntactic preorder and a semantic equivalence
relation \cite{AForm}.  After discussing how the notion of ``to be a
slice of'' can be formally \emph{defined}, the focus will shift to how
to \emph{compute} a slice given a program and a slicing criterion.
Again, among all the possible definitions of slicing, we are
interested in slices obtained by erasing statements from the original
program, i.e., the slice is related to the original program by the
syntactic ordering relation $\sqsubseteq$.  Given a slicing criterion,
the idea is keeping all the statements affecting the \emph{semantic
  equivalence relation} defined by the criterion.  In other words, we
should have to \emph{translate} the formal definition into a
characterization of which statements has to be kept in a slice, or
vice versa which statements can be erased, in order to preserve the
semantic equivalence defining the chosen notion of slicing.
Intuitively, we have to keep all the statements \emph{affecting} the
semantics defined by the chosen slicing criterion.

The standard approach for characterizing slices and the corresponding
relation \emph{being slice of} is based on the notion of Program
Dependency Graph \cite{horPR89,RY88}, as described by Binkley and
Gallagher \cite{BinGalla96}.  \emph{Program Dependency Graphs} (PDGs)
can be built out of programs, and describe how data propagate at
runtime.  In program slicing, we could be interested in computing
dependencies on statements: $s''$ depends on $s'$ if some variables
which are used inside $s''$ are defined inside $s'$, and definitions
in $s'$ reach $s''$ through at least one possible execution path.
Also, $s$ depends \emph{implicitly} on an if-statement or a loop if
its execution depends on the boolean guard.
\begin{example}
  \label{example:statementDependency}
  Consider the program in Figure~\ref{pdgfig} and the derived PDG
  (edges which can be obtained by transitivity are omitted).
  \begin{figure}[h]
    \center{ % Figure/depGraph.pdftex_t
\begin{picture}(0,0)%
\includegraphics{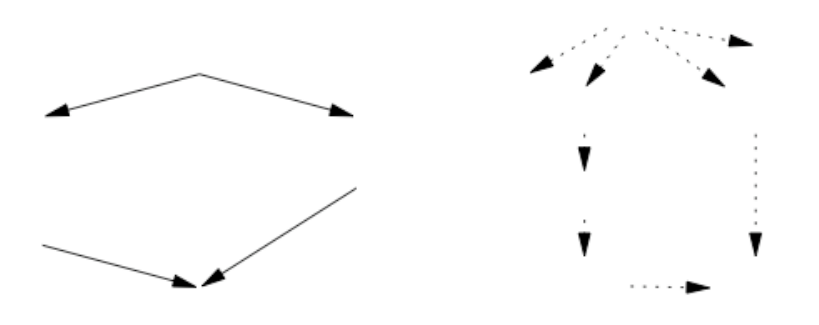}%
\end{picture}%
\setlength{\unitlength}{3947sp}%
\begingroup\makeatletter\ifx\SetFigFont\undefined%
\gdef\SetFigFont#1#2#3#4#5{%
  \reset@font\fontsize{#1}{#2pt}%
  \fontfamily{#3}\fontseries{#4}\fontshape{#5}%
  \selectfont}%
\fi\endgroup%
\begin{picture}(3918,1557)(1651,-1221)
\put(2540,-1185){\makebox(0,0)[lb]{\smash{{\SetFigFont{8}{9.6}{\familydefault}{\mddefault}{\updefault}{\color[rgb]{0,0,0}\IMPASSIGN{y}{v+1}}%
}}}}
\put(4593,252){\makebox(0,0)[lb]{\smash{{\SetFigFont{8}{9.6}{\familydefault}{\mddefault}{\updefault}{\color[rgb]{0,0,0}$s_1$}%
}}}}
\put(5277, 47){\makebox(0,0)[lb]{\smash{{\SetFigFont{8}{9.6}{\familydefault}{\mddefault}{\updefault}{\color[rgb]{0,0,0}$s_6$}%
}}}}
\put(5140,-227){\makebox(0,0)[lb]{\smash{{\SetFigFont{8}{9.6}{\familydefault}{\mddefault}{\updefault}{\color[rgb]{0,0,0}$s_7$}%
}}}}
\put(4388,-227){\makebox(0,0)[lb]{\smash{{\SetFigFont{8}{9.6}{\familydefault}{\mddefault}{\updefault}{\color[rgb]{0,0,0}$s_3$}%
}}}}
\put(4388,-637){\makebox(0,0)[lb]{\smash{{\SetFigFont{8}{9.6}{\familydefault}{\mddefault}{\updefault}{\color[rgb]{0,0,0}$s_4$}%
}}}}
\put(4388,-1048){\makebox(0,0)[lb]{\smash{{\SetFigFont{8}{9.6}{\familydefault}{\mddefault}{\updefault}{\color[rgb]{0,0,0}$s_5$}%
}}}}
\put(5140,-1048){\makebox(0,0)[lb]{\smash{{\SetFigFont{8}{9.6}{\familydefault}{\mddefault}{\updefault}{\color[rgb]{0,0,0}$s_8$}%
}}}}
\put(4046,-158){\makebox(0,0)[lb]{\smash{{\SetFigFont{8}{9.6}{\familydefault}{\mddefault}{\updefault}{\color[rgb]{0,0,0}$s_2$}%
}}}}
\put(1925,-364){\makebox(0,0)[lb]{\smash{{\SetFigFont{8}{9.6}{\familydefault}{\mddefault}{\updefault}{\color[rgb]{0,0,0}\IMPASSIGN{w}{3}}%
}}}}
\put(1651,-364){\makebox(0,0)[lb]{\smash{{\SetFigFont{8}{9.6}{\familydefault}{\mddefault}{\updefault}{\color[rgb]{0,0,0}$s_2$}%
}}}}
\put(3156,-364){\makebox(0,0)[lb]{\smash{{\SetFigFont{8}{9.6}{\familydefault}{\mddefault}{\updefault}{\color[rgb]{0,0,0}$s_6$}%
}}}}
\put(3430,-364){\makebox(0,0)[lb]{\smash{{\SetFigFont{8}{9.6}{\familydefault}{\mddefault}{\updefault}{\color[rgb]{0,0,0}\IMPASSIGN{z}{3}}%
}}}}
\put(1651,-500){\makebox(0,0)[lb]{\smash{{\SetFigFont{8}{9.6}{\familydefault}{\mddefault}{\updefault}{\color[rgb]{0,0,0}$s_3$}%
}}}}
\put(1651,-774){\makebox(0,0)[lb]{\smash{{\SetFigFont{8}{9.6}{\familydefault}{\mddefault}{\updefault}{\color[rgb]{0,0,0}$s_5$}%
}}}}
\put(1651,-637){\makebox(0,0)[lb]{\smash{{\SetFigFont{8}{9.6}{\familydefault}{\mddefault}{\updefault}{\color[rgb]{0,0,0}$s_4$}%
}}}}
\put(1925,-774){\makebox(0,0)[lb]{\smash{{\SetFigFont{8}{9.6}{\familydefault}{\mddefault}{\updefault}{\color[rgb]{0,0,0}\IMPASSIGN{v}{z+w}}%
}}}}
\put(3430,-500){\makebox(0,0)[lb]{\smash{{\SetFigFont{8}{9.6}{\familydefault}{\mddefault}{\updefault}{\color[rgb]{0,0,0}\IMPASSIGN{v}{4}}%
}}}}
\put(1925,-637){\makebox(0,0)[lb]{\smash{{\SetFigFont{8}{9.6}{\familydefault}{\mddefault}{\updefault}{\color[rgb]{0,0,0}\IMPASSIGN{w}{z+4}}%
}}}}
\put(1925,-500){\makebox(0,0)[lb]{\smash{{\SetFigFont{8}{9.6}{\familydefault}{\mddefault}{\updefault}{\color[rgb]{0,0,0}\IMPASSIGN{z}{1}}%
}}}}
\put(3156,-500){\makebox(0,0)[lb]{\smash{{\SetFigFont{8}{9.6}{\familydefault}{\mddefault}{\updefault}{\color[rgb]{0,0,0}$s_7$}%
}}}}
\put(2267, 47){\makebox(0,0)[lb]{\smash{{\SetFigFont{8}{9.6}{\familydefault}{\mddefault}{\updefault}{\color[rgb]{0,0,0}$s_1$}%
}}}}
\put(2540, 47){\makebox(0,0)[lb]{\smash{{\SetFigFont{8}{9.6}{\familydefault}{\mddefault}{\updefault}{\color[rgb]{0,0,0}$(x \leq y)?$}%
}}}}
\put(2267,-1185){\makebox(0,0)[lb]{\smash{{\SetFigFont{8}{9.6}{\familydefault}{\mddefault}{\updefault}{\color[rgb]{0,0,0}$s_8$}%
}}}}
\end{picture}%
}
\caption{PDG example.}  \label{pdgfig}
  \end{figure}
  $s_8$ depends on both $s_5$ and $s_7$ (and, by transitivity, $s_1$)
  since $v$ is not known statically when entering $s_8$.  On the other
  hand, there is \emph{no} dependency of $s_8$ on either (i) $s_6$,
  since $z$ is not used in $s_8$; or (ii) $s_2$, since $w$ is always
  redefined before $s_8$.  The dependency of $s_7$ on $s_1$ is
  implicit since $4$ does not depend on $x$ nor $y$, but $s_7$ is
  executed conditionally on $s_1$.
\end{example}
Formally, a \emph{Program Dependence Graph} \cite{GL91ieee}
$\cG_\prog$ for a program $\prog$ is a directed graph with nodes
denoting program components and edges denoting \emph{dependencies}
between components.  The nodes of $\cG_\prog$ represent the assignment
statements and control predicates in $\prog$.  In addition, nodes
include a distinguished node called \emph{Entry}, denoting where the
execution starts.  An edge represents either a \emph{control
  dependency} or a \emph{flow} (\emph{data}) {\em dependency}.
Control dependency edges $u\lra_c v$ are such that (1) $u$ is the
entry node and $v$ represents a component of $\prog$ that is not
nested within any control predicate; or (2) $u$ represents a control
predicate and $v$ represents a component of $\prog$ immediately nested
within the control predicate represented by $u$.  Flow dependency
edges $u\lra_f v$ are such that (1) $u$ is a node that defines the
variable $x$ (usually an assignment), (2) $v$ is a node that
\emph{uses} $x$, and (3) control can reach $v$ from $u$ via an
execution path along which there is no intervening re-definition of
$x$.

Unfortunately, there is a clear \emph{gap} between the definition of
slicing given in Definition~\ref{defSlice} and the standard
implementation based on program dependency graphs (PDG)
\cite{horPR89,reps91}.  This happens because slicing and dependencies
are usually defined at \emph{different levels} of approximation.  In
particular, we can note that the slicing definition in the formal
framework defines slicing by requiring the same \emph{behavior}, with
respect to a criterion, between the program and the slice, i.e., we
are specifying what is \emph{relevant} as a \emph{semantic}
requirement.  On the other hand, dependency-based approaches consider
a notion of dependency between statements which corresponds to the
\emph{syntactic} presence of a variable in the definition of another
variable.  In other words, slices are usually defined at the
\emph{semantic} level, while dependencies are defined at the
\emph{syntactic} level.  The idea presented in this paper consists,
first of all, in identifying a notion of \emph{semantic} dependency
corresponding to the slicing definition given above, in order to
characterize the implicit parametricity of the notion of slicing on a
corresponding notion of dependency.  This way, we can precisely
identify the semantic definition of slicing corresponding to a given
dependency-based algorithm, characterizing so far the loss of
precision of a given algorithm w.r.t.~the semantic definition.
\begin{figure}
  \begin{center}
    \includegraphics[scale=.3]{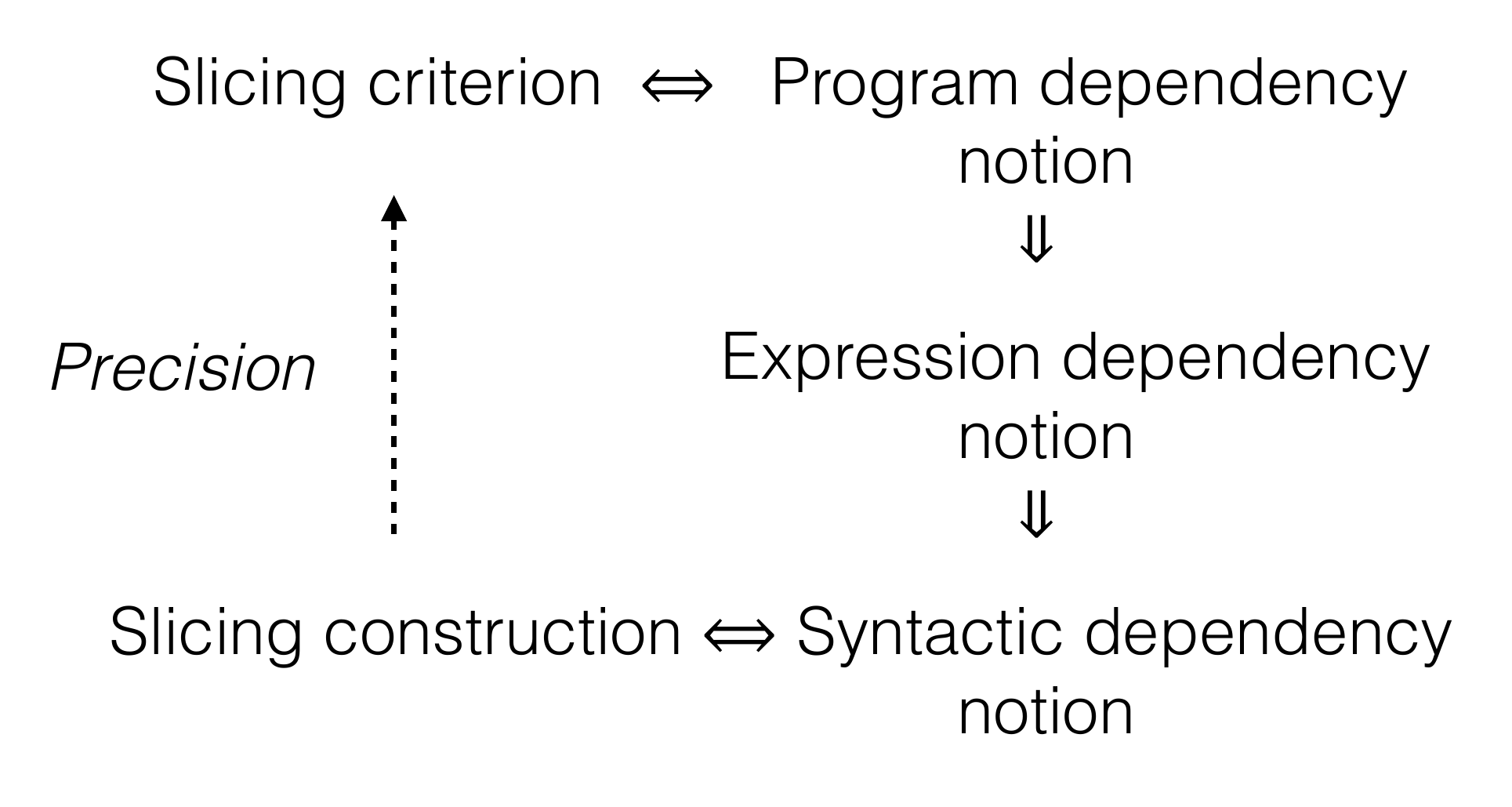}
  \end{center}
  \caption{Schema of dependency-based slicing notions.}\label{fig:schema}
\end{figure}
In Figure~\ref{fig:schema} we show these relations. In particular,
starting from the criterion, we can define an equivalent notion of
dependency which allows us to identify which variables should be kept
in a slice, affecting the whole program semantics (program dependency
notion).
\begin{mydefinition}[(Semantic) Program dependency]
Let $\crit=(\cI,\cX,\cO,\psi)$ be a slicing criterion, and $\prog$ be
a program. The program depends on $x$, denoted
$\CDEPENDS{\crit}{x}{\prog}$ iff
\[
\exists\memory_1,\memory_2\in\cI.\:\forall y\neq x.\:\memory_1(y)=\memory_2(y)\ \wedge\ 
\Proj_{(\cX, \cO, \cL)}(\trace_\prog^{\memory_1})\neq
\Proj_{(\cX, \cO, \cL)}(\trace_\prog^{\memory_2})
\]
\end{mydefinition}
This means that the variable $x$ affects the observable semantics of
$\prog$.

Unfortunately, this characterization is not effective due to
undecidability of the program semantics.
In particular, the amount of traces to compare could be infinite, and
also the traces themselves could be infinite. Hence, we consider a
stronger notion of dependency that looks for {\em local} semantic
dependencies, identifying all the variables affecting at least one
expression used in the program (expression dependency notion), and
this is precisely the {\em semantic} generalization of the syntactic
dependency notion used, for instance, in PDG-based algorithm for
slicing.  In other words we characterize when a variable affects the
semantics of an expression in $\prog$.

Our idea is to make semantic the standard notion of syntactic
dependency, by substituting the notion of \emph{uses} with the notion
of \emph{depends on} \cite{GiacobazziJM12}.  In order to obtain this
characterization, we have to find which variables might affect the
evaluation of the expression $\exp$ in the assignment \CODE{z:=e} or
in a control statement guarded by $\exp$, i.e., which variables belong
to the set $\RELEVANT{\exp}$ of the variables \emph{relevant} to the
evaluation of $\exp$.  As already pointed out, standard syntactic
dependency calculi compute $\RELEVANT{\exp}$ as $\VARS{\exp}$.
\begin{mydefinition}[(Semantic) Expression dependency]
  \label{def:concreteDependencies}
  Let $\crit=(\cI,\cX,\cO,\psi)$ be a slicing criterion.  Let $x\in
  \variables$, $\cY \subseteq \variables$.
    \[
    \begin{array}{rcl} \CDEPENDS{\crit}{x}{\exp} & \Leftrightarrow & \exists
      \state_1,\state_2\in\states.\ \forall y\neq
      x.\:\state_1(y)=\state_2(y) 
       \wedge\ \EVAL{\exp}{\state_1} \neq \EVAL{\exp}{\state_2} \\
      \CDEPENDS{\crit}{\cY}{\exp}& \Leftrightarrow &\exists y \in \cY.\
      \CDEPENDS{\crit}{y}{\exp} \\
      & \Leftrightarrow & \exists
      \state_1,\state_2\in\states.\ \forall w\notin
      \cY.\:\state_1(w)=\state_2(w) 
       \wedge\ \EVAL{\exp}{\state_1} \neq \EVAL{\exp}{\state_2} \\
    \end{array}
    \]
  The formulation of $\CDEPENDS{\crit}{x}{\exp}$ can be rewritten as
  \[
  \exists\state\in\states,
  v_1,v_2\in\values.\: \EVAL{\exp}{\UPDATE{\state}{x}{v_1}} \neq
  \EVAL{\exp}{\UPDATE{\state}{x}{v_2}}
  \]
\end{mydefinition}

\begin{proposition}\label{prop:approxdep}
  Let $\crit$ be a slicing criterion. If $\CDEPENDS{\crit}{x}{\prog}$
  then there exists $\exp$ in $\prog$ such that
  $\CDEPENDS{\crit}{x}{\exp}$.
\end{proposition} 

\begin{proof}
  Let us reason by contradiction. If for each $\exp$ in $\prog$ we
  have $\NCDEPENDS{\crit_\cA}{x}{\exp}$, then all the expressions in
  $\prog$ do not depend on $x$ and therefore, independently from $x$,
  $\prog$ provides precisely the same results.
\end{proof}

By using this notion of dependency, we can characterize the subset
$\RELEVANT{\exp}\subseteq\VARS{\exp}$ containing exactly those
variables which are \emph{semantically} relevant for the evaluation of
$\exp$.  This way, we obtain a notion of dependency which allows us to
derive more precise slices, i.e., to remove statements that a merely
syntactic analysis would leave.

\begin{example}
  \label{ese1} Consider the program $\prog$: \vspace{-.3cm}
  \begin{lstlisting}
   x:=$e_x$;
   y:=$e_y$; 
   w:=$e_w$;
   z:=w+y+2(x$^2$)-w;
  \end{lstlisting}
  \noindent where $\exp_x$, $\exp_y$ and $\exp_w$ are expressions.  We
  want to compute the static slice $\progq$ of $\prog$ affecting the
  final value of \zz (i.e., the \emph{slicing criterion} $\crit =
  (\memories,\{\mbox{\zz}\},\tuple{4,\NATURALS}, \false)$ is
  interested in the final value of \zz).  If we consider the
  traditional notion of slicing, then it is clear that we can erase
  line 3 without changing the final result for \zz.  However, in the
  usual syntactic approach, we would have a dependency between \zz and
  \ww, since \ww is used where \zz is defined.  Consequently, the
  slice obtained by applying this form of dependency would leave the
  program unchanged.  On the other hand, if the semantic dependency is
  considered, then the evaluation of
  $\mbox{\ww}+\mbox{\yy}+2(\mbox{\xx}^2)-\mbox{\ww}$ does not depend
  on the possible variations of \ww, which implies that we are able to
  erase line 3 from the slice.
\end{example}

Next we show how the PDG-based approach to slicing can be modified in
order to cope with semantic slicing.  The PDG approach is based on the
computation of the set of variables \emph{used} in a expression
$\exp$.  In the following, we wonder if this set can be rewritten by
considering a semantic form of dependency.

Hence, let us define the new notion of \emph{semantic} PDG, where all
the flow dependencies are {\em semantic}, i.e., we substitute the flow
edges defined above with semantic flow dependency edges $u\lra_{sf} v$
which are such that (1) $u$ is a node that defines the variable $x$
(usually an assignment), (2) $v$ is a node containing an expression
$\exp$ such that $\CDEPENDS{\crit}{x}{\exp}$ (where $\crit$ is the
criterion with respect with we are computing the slice), and (3)
control can reach $v$ from $u$ via an execution path along which there
is no intervening re-definition of $x$. A (semantic) flow path is a
sequence of (semantic) flow edges.

\begin{proposition}\label{prop:flowdep}
  Let $\prog$ be a program and $\crit$ be a slicing criterion. Let
  $\cG_\prog$ the PDG with flow dependency edges $\lra_f$, and
  $\cG_\prog^{s}$ be the semantic PDG where the flow dependency edge
  are semantic $\lra_{sf}$. If $u\lra_{sf}v$ then $u\lra_{f}v$.
\end{proposition}

\begin{proof}
  Trivially, since if $\CDEPENDS{\crit}{x}{\exp}$, then $\exp$ must
  use the variable $x$.
\end{proof}

In principle, a (backward) slice is composed by all the statements
(i.e., nodes) such that there exists a path from the corresponding
node to the relevant (according to the slicing criterion) use of a
variable of interest (in the criterion) \cite{RY88}.
In other words, we follow backward the (semantic) flow edges from the
nodes identified by the criterion, and we keep all the
nodes/statements we reach.  Hence, the criterion, and therefore the
dependency notion, defines the edges that we can follow for computing
the slice.
 
By using the semantic flow dependency edges, we can draw a new
\emph{semantic} PDG containing less flow edges, i.e., only those
corresponding to semantic dependencies.  At this point, the type of
slicing (either static, dynamic or conditional) characterized by the
criterion decides which nodes can be kept in the PDG.

\begin{theorem}\label{th:slispdg}
  Let $\cC=\tuple{\cI,\cX,\cO,\psi}$ be a slicing criterion. Let
  $\prog$ be a program and $\cG^s_\prog$ its \emph{semantic PDG},
  i.e., a PDG whose flow edges are $\lra_{sf}$.  Let $\progq$ the
  subprogram of $\prog$ containing all the statements corresponding to
  nodes such that there exists a semantic flow path in $\cG^s_\prog$
  from them to a node in
  $\cN_{\cO}\defi\sset{n}{\exists\tuple{n,K}\in\cO}$. Then $\progq$ is
  a slice w.r.t.\ the criterion $\cC$.
\end{theorem}

\begin{proof}
  Note that, the PDG construction is syntactic, and therefore
  independent from the input set $\cI$, hence any slice computed by
  using the PDG holds for any possible input memory in
  $\memories$. Moreover, since we simply collects the statements
  potentially affecting the program observation, we cannot decide
  which iteration to observe, for this reason each program point is
  taken in the slice independently from the iteration to observe, for
  this reason the obtained slice will provide the same result for any
  possible iteration, i.e., the set of interesting points to observe
  are
  $\cN_\cO\times\{\NATURALS\}=\sset{\tuple{n,\NATURALS}}{n\in\cN_{\cO}}$. Finally,
  by construction we cannot have statements executed in the slice
  which are not executed in the original programs. Hence, the
  criterion enforced by the PDG slice construction is
  $\crit^{\mbox{\tiny
      PDG}}=(\memories,\cX,\cN_\cO\times\{\NATURALS\},\true)$ (direct
  consequence of the \emph{Slicing theorem} in \cite{RY88}), which for
  each $\cI\subseteq\memories$ and
  $\{K_n\}_{n\in\cN_\cO}\subseteq\wp(\NATURALS)$ is a slice also
  w.r.t.\ the criterion
  $(\cI,\cX,\{\tuple{n,K_n}~|~n\in\cN_\cO\},\psi)$, by \cite{AForm}
  (Equation~\ref{eq:defrelcrit} in the Appendix). Hence, we have that
  the results is a slice w.r.t.\ $\crit$ \cite{AForm}.
\end{proof}

We have pointed out so far the difference between syntactic and
semantic dependencies: it can be the case that a variable
syntactically appears in an expression without affecting its value.
Actually, one could argue that the case is not very likely to happen:
the possibility to find an assignment like \CODE{x:=y-y} in code
written by a professional software engineer is remote to the very
least.  However, when it comes to abstract dependencies, the picture
is quite different, and we could even say that, in the present work,
(concrete) semantic dependencies have been mainly introduced to
prepare the discussion about their abstract counterpart.  Indeed, it
is much more likely that some variables are not semantically relevant
to an expression if the value of interest is an abstract one, e.g.,
the parity or the sign of a numeric expression, or the nullity of a
pointer.  This justifies the definition of a semantic notion of
dependency at the abstract level.

\section{Abstract Dependencies}
\label{sec:abstractDependencies}
 % abstractDependencies
This section discusses the problem of defining and computing abstract
dependencies allowing us to capture the dependency relation between
variables w.r.t.\ a given abstract criterion.  In the previous
section, we formalized this relation in the concrete semantic case;
the following example takes it to the abstract level.

\begin{example}
  Consider the expression $\exp=$~\CODE{2x$^2$+y}: although both
  variables are semantically relevant to the result, only \yy can
  affect its parity, since \CODE{2x$^2$} will always be even.  On the
  other hand, note that both variables are relevant to the sign of
  $\exp$, in spite of the positivity of \CODE{x$^2$}.  In fact, given
  a negative value for \yy, a change in the value of \xx can alter the
  sign of the entire expression.
\end{example}

First of all, it is worth noting that the notion of \emph{semantic
  program dependency} can be easily extended to abstract criteria,
simply by changing the projection considered.  In this case we will
write that $\CDEPENDS{\crit_\cA}{x}{\prog}$, meaning that $x$ has
effect on the abstract projection of $\prog$ determined by
$\crit_\cA$.  Also in the abstract case we inherit the undecidability
of the concrete semantics of $\prog$; hence, again, we have to
approximate the semantic program dependency with a local notion of
\emph{abstract semantic expression dependency}.  Unfortunately, when
dealing with abstract criteria $\crit_\cA$, some aspects become more
complicated.

\subsection{Abstract slicing and dependencies}

In the previous section, we defined the concrete semantic dependency
by identifying those variables that do not \emph{interfere} with the
final observation of each expression.  Analogously, in order to define
a general notion of abstract semantic dependency, we need to consider
the \emph{abstract} interference between a property of a variable and
a property of an expression.

The definition below follows the same philosophy as \emph{narrow
  abstract non-interference} \cite{GM04popl,Mastroeni13}, where
abstractions for observing input and output are considered, but these
abstractions are observations of the \emph{concrete} executions.
\begin{mydefinition}[(Abstract) Semantic expression dependency, \emph{Ndep}]
  \label{def:abstractDependencies}\label{def:narrowDependencies}
  Consider the abstractions $\uco\in\ucos(\wp(\values))$ and
  $\ov{\eta}\in\ucos(\wp(\val))^n$, where $n$ is the number of
  variables, i.e., $\ov{\eta}$ is a tuple of properties such that
  $\eta_y$ is the property on the variable $y$.
    \[
    \begin{array}{rcl} \ANARROWDEPENDS{x}{}{\uco}{\exp}{\ov{\eta}} & \Leftrightarrow & \exists \state_1,\state_2\in\states\:.\ 
    \left(\forall y\neq x.\:\eta_y(\state_1(y))=\eta_y(\state_2(y))\ \wedge\ 
    \uco(\EVAL{\exp}{\state_1}) \neq \uco(\EVAL{\exp}{\state_2})\right) \\
    \end{array}
    \]
  \end{mydefinition}

This notion is a generalization of
Definition~\ref{def:concreteDependencies} where we abstract the
observation of the result ($\uco$) and the information that we fix
about all the variables different from \xx ($\ov{\eta}$).  Still, this
notion characterizes whether the variation of the value of \xx affects
the abstract evaluation in $\uco$ of $\exp$.
 
As an important result, we have
$\ANARROWDEPENDS{x}{}{\uco^1}{\exp}{\ov{\eta}} \Leftrightarrow
\ANARROWDEPENDS{x}{}{\uco^2}{\exp}{\ov{\eta}}$ whenever $\uco^1$ and
$\uco^2$ induce the same partitions (either on values, or tuples of
values).  This happens because, in
Definition~\ref{def:abstractDependencies}, both abstractions are only
applied to singletons.  In the following, only partitioning domains
will be considered since it is straightforward to note that
$\ANARROWDEPENDS{x}{}{\uco}{\exp}{\ov{\eta}}$ is affected only by
$\Pi(\uco)$, rather than by $\uco$ itself.

When dealing with abstractions, and therefore with abstract
computations, some more considerations have to be taken into account.
Consider the program in Example~\ref{ese1}, and consider the $\PARDOM$
property (Section \ref{sec:basicAbstractInterpretation}) for all
variables on both input and output .  If we compute the set of
variables on which the parity of $\exp =$ \CODE{w+y+2(x$^2$)-w}
depends on, then we can observe that $\exp$ is still independent from
\ww, but is also independent from any possible variation of \xx.
At a first sight, the parity of \CODE{w+y+2(x$^2$)-w} is independent
from \xx just because \CODE{2(x$^2$)} is constantly even.  However, it
is not only a matter of constancy: a deeper analysis would note that
we can look simply at the abstract value of \xx only because the
operation involved (the sum) in the evaluation is \emph{complete}
(Section \ref{sec:basicAbstractInterpretation}), i.e., precise,
w.r.t.~the abstract domain considered ($\PARDOM$).  In particular,
when we deal with abstract domains which are complete for the
considered operations, then it is enough to look at the abstract value
of variables in order to compute dependencies.  Indeed, consider
the $\SIGNDOM$ domain (Section \ref{sec:basicAbstractInterpretation}).
In this case, even if the sign of \CODE{2(x$^{2}$)} is constantly
positive, the final sign of \zz might be affected by a concrete
variation of \xx (e.g., consider $\mbox{\yy} = -4$ and two executions
in which \xx is, respectively, 1 and 5).  Therefore, \xx has to be
considered relevant, although the sign of \CODE{2(x$^{2}$)} (the only
sub-expression containing \xx) is constant.  This can be also derived
by considering the logic of independencies from \cite{AB07} since, by
varying the value of \xx, we can change the sign of $\exp$.

Unfortunately, the notion of abstract dependencies given in
Definition~\ref{def:abstractDependencies} is not suitable for
weakening the PDG approach, as we have done in the concrete semantic
case.  Let us explain the problem in the following example.
\begin{example}\label{ex:probabs}
  Consider the program $\prog=$\CODE{C; x:=y>0?0:1;}, where \CODE{C}
  is come code fragment and the expression \CODE{b?e1:e2} evaluates to
  \CODE{e1} if \CODE{b} is true, and to \CODE{e2} otherwise.

  \noindent Suppose the criterion requires the observation of the
  parity of $x$ at this program point, i.e., $\cA=\tuple{x:\PARDOM}$.
  Then it is straightforward to observe that the expression depends on
  $y$, but we would like to be more specific (being in the context of
  abstract slicing), and we can observe that it is the sign of $y$
  that affects the parity of the expression, and therefore of $x$. At
  this point, in the code \CODE{C} we should look for the variables
  affecting not simply $y$ (as expected in standard slicing
  approaches), but more specifically the {\em sign} of $y$, a
  requirement not considered in the abstract criterion.
\end{example}

This example shows that, if we aim at computing abstract dependencies
without losing too much information, we would need an algorithm able
to keep trace backwards, not only of the different variables that
become of interest (affecting the desired criterion), but also of the
different properties to observe on variables affecting the desired
property of the criterion. This means that, while for the concrete
semantic program dependency we can provide a definition depending only
on the criterion, this is not possible in a more abstract context,
where each flow edge should be defined depending on abstract
properties potentially different from those in the abstract criterion,
and which should be characterized dynamically backward starting from
the criterion. Unfortunately, this is not possible in
Definition~\ref{def:abstractDependencies}, where we always look for
the variation of the value and not of an abstract property of \xx,
hence if we would use this notion for substituting the semantic
dependency in PDGs we would not have so much advantage.

These observations make clear that, if we aim at constructively
characterize abstract slicing by means of the abstract dependency
notion provided in Definition~\ref{def:abstractDependencies}, we need
to build from scratch a systematic approach for characterizing
abstract slicing. Towards this direction, the first step we propose is
a computable approximation of the abstract dependencies of
Definition~\ref{def:abstractDependencies}.

\subsection{A constructive approach to Abstract Dependencies}
\label{sec:constructiveApproachesToAbstractProgramSlicing}
 % constructiveApproachesToAbstractProgramSlicing
By means of the (uco-dependent) definition of operations on abstract
values, it is possible to automatically obtain (an over-approximation
of) the set of relevant variables.  The starting point is the
\emph{brute-force} approach which uses the abstract version of
concrete operations, and explicitly \emph{goes into} the quantifiers
involved in Definition \ref{def:narrowDependencies}.

\begin{example}
  Consider the program in Example~\ref{ese1} and let $\eta_y=\PARDOM$.
  In order to decide whether
  $\ANARROWDEPENDS{\mbox{\xx}}{}{\uco}{\exp}{\ov{\eta}}$ holds, the
  brute-force approach considers the abstract evaluation of $\exp$ in
  all contexts where all variables different from \xx does not change,
  up to $\ov{\eta}$, while \xx may change.  In this example, this
  boils down to consider pairs of memories where \yy has the same
  parity (with no information about sign), and we take memories where
  \xx changes value, and see whether the final values of $\exp$ agree
  on the expression observation $\uco$.  Suppose $\PARDOM(\yy)=\EVEN$
  (meaning that it is even but the sign is unknown) and suppose
  $\uco=\SIGNDOM$, then we should have to compute the abstract value
  of $\exp$ for each possible value for \xx. It is clear that we can
  easily find $\sigma_1$ and $\sigma_2$ such that
  \[ 2 * \sigma_1(\xx)^2 + \sigma_1(\yy) = \NEG, \qquad 2 * \sigma_2(\xx)^2 + \sigma_2(\yy)
  = \POS \] even if $\PARDOM(\sigma_1(\yy))=\PARDOM(\sigma_2(\yy))$
  (for instance $\sigma_1(\yy)=\sigma_2(\yy)=-4$ while
  $\sigma_1(\xx)=1$ and $\sigma_2(\xx)=5$).  It is clear that the sign
  of $\exp$ may depend on the value of \xx since to ``fix'' the
  abstract property of the other variables is not enough to ``fix''
  the final value of $\exp$ w.r.t.~$\uco$. On the other hand, the
  parity of $\exp$ does not depend on \xx, in particular if we fix the
  property $\PARDOM$ of \yy, for instance to $\EVEN$ (but it holds
  also for the other abstract values) then
  \[ 2 * \NEGODD^2 + \EVEN = \POSEVEN, \qquad 2 * \POSEVEN^2 + \EVEN
  = \EVEN \qquad\ldots \] Namely, if we fix all the (abstract values
  of the) variables but \xx, then the parity of the result does not
  change, hence the variation of \xx does not affect the parity of the
  expression $\exp$.
\end{example}

In the following, we introduce an algorithm able to improve the
computational complexity of the brute-force approach, especially on
bigger ucos, and when (1) several variables are involved in
expressions, and (2) a significant part of them is irrelevant.

\subsubsection{Checking \emph{Ndep}}
\label{section:algorithmicIdeasForCheckingNdep}

In the following, we discuss how we can constructively compute narrow
dependencies.  Unfortunately, in static analysis, the concrete
semantics cannot be used directly as it appears in
Definition~\ref{def:narrowDependencies}, hence we need to approximate
this abstract notion. The following definition introduces a stronger
notion of dependency based on a sound abstract semantics
$\ABSEVAL{\cdot}{}$ (Section \ref{sec:abstractSemantics}), which
approximates narrow dependencies.
\begin{mydefinition}[\emph{Atom-dep}]
  \label{def:atomicAbstractDependencies}
  An expression $\exp$ \emph{atom-depends} on $x$ (written
  $\ATOMDEPENDS{x}{}{\uco}{\exp}{\ov{\eta}}$) with respect to
  $\uco\in\ucos(\wp(\values))$ and $\ov{\eta}\in\ucos(\wp(\values))^n$
  ($n$ number of variables) if and only if there exist $\state_1,
  \state_2 \in \states$ such that
  \[ \forall y\neq x.\eta_y(\state_1(y))=\eta_y(\state_2(y))\ \quad \wedge\ \quad \neg
  \ISATOM{\ABSEVAL{\exp}{\{\state_1,\state_2\}}}{\uco} \]
\end{mydefinition}
Being domains partitioning, the \emph{non-atomicity} requirement $\neg
\ISATOM{\cdot}{\uco}$ amounts to say that all $\uco$-abstract
evaluations of $\exp$, starting from different values for $x$, may not
be abstracted in the same abstract value (this is the crucial issue in
\emph{Ndep}), i.e., $\state_1$ and $\state_2$ may lead to different
abstract values for $e$. Next results shows that {\em Atom-dep} is an
approximation of {\em Ndep}, since {\em Ndep} implies {\em Atom-dep},
meaning that {\em Atom-dep} may only add false dependencies, but
cannot lose abstract dependencies characterized by {\em Ndep}.

\begin{proposition}\label{theorem:equivalenza}
  Consider the abstractions $\uco\in\ucos(\wp(\values))$ and
  $\ov{\eta}\in\ucos(\wp(\val))^n$, where $n$ is the number of
  variables, i.e., $\ov{\eta}$ is a tuples of properties.  For every
  $\exp$ and $x$, $\ANARROWDEPENDS{x}{}{\uco}{\exp}{\ov{\eta}}$
  implies $\ATOMDEPENDS{x}{}{\uco}{\exp}{\ov{\eta}}$.
\end{proposition}

\begin{proof}
  Suppose $\ANARROWDEPENDS{x}{}{\uco}{\exp}{\ov{\eta}}$, i.e.,
  $\exists \state_1,\state_2\in\states\:.\ (1)\:\forall y\neq
  x.\:\ABSEQUALSPY{\state_1}{\state_2}{\eta_y}{y}\ \wedge\ (2)\:\uco(\EVAL{\exp}{\state_1})
  \neq \uco(\EVAL{\exp}{\state_2})$, then we have to prove
  $\ATOMDEPENDS{x}{}{\uco}{\exp}{\ov{\eta}}$, i.e., there exist
  $\state_1, \state_2 \in \states$ such that $(3)\:\forall y\neq
  x.\ABSEQUALSPY{\state_1}{\state_2}{\eta_y}{y}$ and $\ok{(4)\:\neg
    \ISATOM{\ABSEVAL{\exp}{\{\state_1,\state_2\}}}{\uco}}$. Since
  conditions $(1)$ and $(3)$ are the same, we have to prove that
  $\exists \state_1,\state_2\in\states.\:(2)\Ra(4)$.
  
  Consider $\state_1,\state_2\in\states$ satisfying $(1)$ and such
  that $\uco(\EVAL{\exp}{\state_1}) \neq \uco(\EVAL{\exp}{\state_2})$,
  then $\uco(\EVAL{\exp}{\state_1}),\uco(\EVAL{\exp}{\state_2})\in
  \uco(\EVAL{\exp}{\{\state_1,\state_2\}})\subseteq\uco(\EVAL{\exp}{\uco(\{\state_1,\state_2\})})=
  \ABSEVAL{\exp}{\{\state_1,\state_2\}}$. At this point, since
  $\ABSEVAL{\exp}{\{\state_1,\state_2\}}$ contains two different
  values of $\uco$, then, by definition, it cannot be an atom of
  $\uco$.
\end{proof}

Starting from this new approximate notion, our idea is to provide an
algorithm
over-approximating the set of variables relevant for a given
expression $\exp$ when $\forall y\neq x.\eta_y=\uco$, namely the
abstraction observed in output is the same that we fix on the input
variables.  The idea is to start from an empty set of not relevant
variables $X$ for an expression $\exp$, and incrementally increase
this set adding all those variables that \emph{surely} are not
relevant for the expression (obtaining an under-approximation of
abstract dependencies). Finally, the complement of such set is
returned, which is an over-approximation of relevant variables.

In order to check the dependency relation, we aim at checking whether
a change of the values of a variables makes a difference in the
evaluation of the expression.  Dependencies are computed according to
\emph{Atom-dep}, in order to approximate \emph{Ndep}.  In a
brute-force approach, \emph{Atom-dep} would be verified by checking
for each $\astate$ associating atomic values to variables we have that
$\ok{\ISATOM{\ABSEVAL{\exp}{\astate}}{\uco}}$ is always the same atom.

\begin{example}
  \label{example:dependenciesOnSign}
  Let $e$ be an expression involving variables $x$, $y$ and $z$, and
  $\uco=\PARSIGNDOM$.  In principle, in order to compute the set of
  $\uco$-dependencies on $\exp$, we must compute $\ABSEVAL{\exp}{}$ on
  every possible atomic value\footnote{Remember that atoms in $\uco$
    are $\ZERO$, $\POSEVEN$, $\POSODD$, $\NEGEVEN$, and $\NEGODD$;
    since this is a partition of concrete values, we describe all
    concrete inputs by computing $\ABSEVAL{\exp}{}$ on atoms.} of \xx,
  \yy and \zz, i.e., $\ABSEVAL{\exp}{}$ must be computed $5^3 = 125$
  times.  \yy is \emph{not} relevant to $\exp$ if, for any abstract
  values $\avalue_x, \avalue_z \in \ATOMS{\uco}$, there exists an
  atomic abstract value $\avaluee \in \ATOMS{\uco}$ such that $\forall
  \avalue \in \ATOMS{\uco}.\ \avaluee =
  \ABSEVAL{\exp}{\{\BIND{x}{\avalue_x},\BIND{y}{\avalue},\BIND{z}{\avalue_z}\}}$.
  This amounts to say that changing the value of \yy does not affect
  $\exp$, since we require the same output atomic evaluation for each
  possible abstract value for \yy. Indeed, if the result is not
  atomic, it means that we have at least two different abstract
  results for different values of \yy. Analogously, for different
  (abstract) values for \yy we have different atomic results, then
  again it means that there exists a variation of \yy affecting the
  abstract evaluation of $\exp$.
\end{example}

\noindent
However, it is possible to be smarter:

\begin{itemize}
\item[$\bullet$] \emph{Excluding states:} consider dependencies of
  $\exp$ in Example \ref{example:dependenciesOnSign}, computed at
  program point $n$.  Suppose $\ABSEVAL{\cdot}{}$ (used as a tool to
  infer invariant properties, as discussed in Section
  \ref{sec:abstractSemantics}) is able to infer, at point $n$, that
  the abstract state $\astatei{n}$ is such that $\astatei{n}(\yy) =
  \ABSVAL{posodd}$ correctly approximates the value of variables at
  $n$.  Then, we only need to consider states of the form
  $\{\BIND{x}{\avalue_x},\BIND{y}{\ABSVAL{posodd}},\BIND{z}{\avalue_z}\}$
  as inputs for $\ABSEVAL{\exp}{}$ (now considered as the abstract
  computation of expressions, according to Definition
  \ref{def:atomicAbstractDependencies}) at $n$.
\item[$\bullet$] \emph{Computing on non-atomic states:} let $E = \{
  \POSEVEN, \ZERO, \NEGEVEN \}$ and $O = \{ \POSODD, \NEGODD \}$.  In
  this case,
  \[ \forall \avalue' \in E, \avalue'' \in O. \
  \ABSEVAL{\exp}{\{\BIND{x}{\avalue_x},\BIND{y}{\avalue'},\BIND{z}{\avalue}\}}
  \leq \avaluee \] is implied by the more general result
  \[ \ABSEVAL{\exp}{\{\BIND{x}{\avalue_x},\BIND{y}{\EVEN},\BIND{z}{\ODD}\}} \leq
  \avaluee \] since $E$ and $O$ are partitions, respectively, of
  $\EVEN$ and $\ODD$, and $\ABSEVAL{\exp}{}$ is monotone: $\astatei{1} \leq
  \astatei{2}$ implies $\Rightarrow \ABSEVAL{\exp}{\astatei{1}} \leq
  \ABSEVAL{\exp}{\astatei{2}}$.  This means that results obtained on
  $\astate$ can be used on $\astatei{1} \leq \astate$.
\end{itemize}

In the following, we compute dependencies w.r.t. the abstract state
$\astatei{n}$, which is the abstract state computed at $n$ as an
invariant at that program point. In the worse case, we don't have any
information about the different variables, and therefore $\astatei{n}$
associates $\TOP$ (of the considered abstraction $\uco$) to all the
unknown variables.  At this point, we aim at proving that $\exp$ is
independent from a set of variable $X$, and in order to prove this
fact we need to prove that the evaluation of $\exp$ is always the same
atom in $\uco$, independently from the value of the variables in $X$,
without proving for all these values. Hence, our idea is to prove this
atomicity, if it holds, by iteratively refining the starting abstract
state $\astatei{n}$. Let us explain the intuition in the following
example.

\begin{example}
  \label{example:funesharppMeaning}
  Let $\exp \equiv$ \CODE{x * x + 1} and $\uco = \SIGNDOM$.  
  Let also
  $\ABSEVAL{\exp}{}$ follow the usual rules on $*$ and $+$: $\POS * \POS
  = \POS$, $\NEG * \NEG = \POS$, $\TOP * \TOP = \TOP$, $\POS + \POS =
  \POS$, $\TOP + \POS = \TOP$, etc.  Suppose to start from a memory such that $\{\BIND{\xx}{\TOP}\}$, then we observe that 
  $\ABSEVAL{\exp}{\{\BIND{\xx}{\TOP}\}} \in\ATOMS{\SIGNDOM}$ cannot be proved by using these
  rules, since $\TOP * \TOP + \POS = \TOP$, which is not atomic, and therefore there may be a dependency. Then, consider the possible refinements w.r.t. \xx in $\SIGNDOM$, namely 
  $\{\{\BIND{\xx}{\POS}\},\{\BIND{\xx}{\ZERO}\},\{\BIND{\xx}{\NEG}\}\}$. 
  This is enough to compute
  \[ \ABSEVAL{\exp}{\{\BIND{\xx}{\POS}\}} = \ABSEVAL{\exp}{\{\BIND{\xx}{\ZERO}\}} =
  \ABSEVAL{\exp}{\{\BIND{\xx}{\NEG}\}} = \POS\] meaning that any
  variation of \xx provides the same atomic result in $\uco$.
\end{example}

Given an initial abstract state $\astate$, we define the set of all
the abstract states that refine its abstract value on the variables in
$X$, while leaving unchanged atomic values of other variables.
\[
\SUBS{\astate}{X}=\sset{\astatei{i}}{\forall y\notin X.\:\astate(y)=\uco(\state(y))=\astatei{i}(y)\in\ATOMS{\uco},\:\forall x\in X.\: \astatei{i}(x)\leq\astate(x)}
\]
This is the set of all the possible abstract states that can be
obtained by restricting abstract values of variables in $X$ (the least
values are the atoms of $\uco$) starting from an initial memory
$\astate$, while all the other variables are atomic\footnote{Namely,
  all the other variables are fixed to the smallest possible abstract
  values in the abstract domain $\uco$.} and fixed by $\astate$.
Hence, the idea is that, once that all the variables different from
\xx are specified by atoms, we first compute the expression with
$\BIND{\xx}{\astate(x)}$, if the result is atomic, it means that all
the values for \xx provide the same results, and we can conclude that
surely $\exp$ does not depend on \xx. If the result is not atomic, it
may be because there is a dependency or because the abstraction is
incomplete for the semantics of the expressions, hence, in order not
to be too coarse, we consider the covering of the $\TOP$ in $\uco$,
namely we take the elements under the top and we repeat for all these
values. If all the computations provide the same atomic results then
we surely have $\exp$ independent form \xx, otherwise we continue to
refine the abstract values. If, we reach atomic values for \xx then we
terminate the recursion and we conclude that there may be
dependency. Hence we have to define a recursive predicate computing
this iteration: Given an expression $\exp$ and an atom $\avaluee$, the
\emph{atomicity condition} $\ACu{\avaluee}{\astate}$ holds iff
$\ABSEVAL{\exp}{\astate}$ gives $\avaluee$.
\[
\ACC{\astate}{X}=
\left \{
\begin{array}{ll}
\avaluee & \mbox{if}\ \ACu{\avaluee}{\astate}\ \vee\ \exists\avaluee\in\ATOMS{\rho}.\:\\
& \exists\{\astatei{i}\}_{i\in[1,k]}\in X\mbox{\em -covering}(\astate).\:\forall i\in[1,k].\:\ACC{\astatei{i}}{X}=\avaluee\\
\bot & \mbox{if}\ \forall x\in X.\:\astate(x)\in\ATOMS{\rho}\ \wedge\ \nexists\avaluee\in\ATOMS{\rho}.\:\ACu{\avaluee}{\astate}\\
\end{array}
\right .
\]
with
\[
X\mbox{\em -covering}(\astate)=\sset{\{\astatei{1},\ldots,\astatei{k}\}}{\forall i\in[1,k].\:\astatei{i}\in \SUBS{\astate}{X}\ \wedge\ \forall x\in X.\:\astatei{i}(x)\leq_\iota\astate(x),\\
\mbox{and}\ \bigvee_i \astatei{i}(x)=\astate(x)}
\]
where $\astatei{1}(x)\leq_\iota\astatei{2}(x)$ iff
$\astatei{1}(x)=\astatei{2}(x)$ or $\astatei{1}(x)$ is a direct
sub-value of $\astatei{2}(x)$ in $\uco$. Intuitively,
$\ACC{\astate}{X}$ terminates the iterations either when the
evaluation of the expression is always the same atom independently
from the abstract value of the variables in $X$ (meaning that there is
no dependency from $X$) or when, for all the possible atomic values
for variables in $X$, the evaluation of the expression is not atomic
or may have different atomic values (meaning that there may be
dependency from variables in $X$).  Hence, a judgment
$\ACC{\astatei{n}}{X}$ means that an atomic value for $\exp$ was
obtained without the need of further restricting the variables in $X$,
when all the other variables are atomic; therefore, $X$ only contains
non-relevant variables.

The $\FINDNDEPS$ algorithm (Figure \ref{fig:NdepAlgorithm}) starts by
trying to prove that $\ACC{\astatei{n}}{\variables}$ is atomic.  If it
find a restriction of the abstract values of all the variables making
the $\exp$ evaluation atomic, then $\exp$ depends on no variables.
Otherwise, the set $X$ is decreased non-deterministically (one element
at a time, randomly) until some judgment $\ACC{\astatei{n}}{X}$ is
proved.
\begin{proposition}[Soundness]
  \label{proposition:correctness}
  Let $\uco\in\ucos(\wp(\values))$ be partitioning and
  $\astate$\footnote{This represent the initial information about $X$,
    if there is no information then it works mapping all the variables
    in $X$ to $\top$} an abstract state. Let $\ov{\uco}$ denote the
  tuple of $\uco$, on each variable.  For all $\astatei{n}\in
  \SUBS{\astate}{X}$ we have that if
  $\ACC{\astatei{n}}{X}\in\ATOMS{\uco}$ (namely not $\bot$), then
  there is no $x \in X$ such that
  $\ATOMDEPENDS{x}{}{\uco}{\exp}{\ov{\uco}}$.
\end{proposition} 

\begin{proof}
  Suppose to know that the variables in $X$ may soundly have only the
  values determined by $\astate$, namely for each $x\in X$ the
  concrete value of $x$ in the computation is contained in
  $\astate(x)$. This means that, for each $x\in X$, we can check
  dependency only for the values ranging over $\astate(x)$.  Consider
  $\astatei{1},\astatei{2}\in\SUBS{\astate}{X}$, by definition of
  $\SUBS{\astate}{X}$ we have that $\forall y\notin X$ we have
  $\astatei{1}(y)=\astatei{2}(y)=\astate(x)$, moreover again by
  definition these values are atomic hence
  $\astatei{1}(y)=\uco(\state_1(y))=\uco(\state_2(y))=\astatei{2}(y)$. At
  this point, in order to prove that there is no dependency, we have
  to show that
  $\ISATOM{\ABSEVAL{\exp}{\{\astatei{1},\astatei{2}\}}}{\uco}$
  holds. But, at this point, by hypothesis since
  $\ACC{\astate}{X}\in\ATOMS{\uco}$, meaning that there exists a
  covering of $\astate$
  $\{\astatei{i}\}_{i\in[1,k]}\subseteq\SUBS{\astate}{X}$ and
  $\avaluee\in\ATOMS{\uco}$ such that $\forall i\in[1,k]$ we have
  $\ACu{\avaluee}{\astatei{i}}$, but this implies that
  $\ACu{\avaluee}{\astatei{1}}$ and $\ACu{\avaluee}{\astatei{2}}$,
  which trivially implies the thesis.
\end{proof}

\noindent
Importantly, the assertions $\ACC{\astatei{n}}{X}$ and
$\ACC{\astatei{n}}{Y}$ guarantee $\ATOMDEPENDS{X\cup
  Y}{}{\uco}{\exp}{\ov{\uco}}$ not to hold, even if
$\ACC{\astatei{n}}{X\cup Y}$ cannot be directly proved.  The final
result of $\FINDNDEPS$ is $\VARS{\exp} \smallsetminus Z$, where $Z$ is
the union of all sets $Z_i$ such that $\ACC{\astatei{n}}{Z_i}$ can be
proved.  By Proposition \ref{proposition:correctness}, this set is an
over-approximation of relevant variables.
\begin{figure}
  \begin{pseudocode}
   function $\FINDNDEPS$ { 
     $\mbox{nonDep}$ := $\emptyset$; // $\mbox{{\it can be modified by prove()}}$
     $\PROVE$($\astatei{n}$,$\VARS{e}$);
     return $\VARS{e}\ \smallsetminus\ \mbox{nonDep}$; // $\mbox{{\it relevant variables}}$
   }
   procedure $\PROVE$($\astate$,$X$) {
     if($\ACC{\astate}{X}\neq\bot$) then $\mbox{nonDep}$ := $\mbox{nonDep}\ \cup\ X$;
     else foreach ($x \in X$) { $\PROVE$($\astate$,$X\ \smallsetminus\ \{x\}$); }
   }
  \end{pseudocode}
  \caption{The $\FINDNDEPS$ algorithm}
  \label{fig:NdepAlgorithm}
\end{figure}

The $\FINDNDEPS$ algorithm may deal, in principle, with
\emph{infinite} abstract domains, and in particular with abstract
domains with infinite descending chains, since non-dependency results
can be possibly proved without exploring the entire state-space; in
fact, if $\ACu{\avaluee}{\astate}$ can be proved, then it is not
needed to descend into the (possibly infinite) set of sub-states of
$\astate$.  This is not possible in the brute-force approach.  It is
also straightforward to add \emph{computational bounds} in order to
stop ``refining'' states if some amount of computational effort has
been reached.

\subsubsection{Dependency erasure in the abstract framework}
\label{section:simplifyingDomainsForEliminatingDependencies}

The problem of computing abstract dependencies can be observed from
another point of view: given $e$ and a set $X$ of variables, we may be
interested in soundly approximating the \emph{most concrete} $\uco$
such that $\ANARROWDEPENDS{X}{}{\uco}{\exp}{\ov{\uco}}$ does not
hold. Namely, the most concrete observation guaranteeing the
non-interference of the variables in $X$ on the evaluation of $\exp$
\cite{GM04popl,Mastroeni13}.  This can be accomplished by repeatedly
simplifying an initial domain $\uco_0$ in order to eliminate abstract
values which are \emph{responsible} for dependencies.  In order to
avoid dependencies on $X$, we should have $\ACC{\astatei{n}}{X}$,
i.e., $\ACu{\avalue}{\astatei{X}}$ should hold for any $\astatei{X}
\in \SUBS{\astatei{n}}{X}$.  If this does not hold for some $\astate$,
then $\uco$ is modified to obtain the atomicity of $\avalue =
\ABSEVAL{\exp}{\astate}$.

We design a simple algorithm $\EDEP{\exp}{\uco_0}{X}$ (Figure
\ref{edep}) which repeatedly checks if there exists $\avalue$ such
that $\ACu{\avalue}{\astate}$.  Initially, the current state $\astate$
is $\astatei{n}$ (remember that $\astatei{n}$ is the abstract state
correctly describing invariant properties of variables at the current
program point $n$); then, it is progressively specialized to states
belonging to one of its $X$-coverings, until one of the following
holds:
\begin{itemize}
\item[$\bullet$] $\ACu{\avalue}{\astate}$; in this case, $\uco$ is
  precise enough to exclude dependencies on $X$ in $\astate$, and is
  not further modified;
\item[$\bullet$] $\astate$ cannot be refined anymore (it is atomic on
  $\VARS{\exp} \smallsetminus X$) but $\avalue$ is non-atomic; in this
  case, $\uco$ needs to be simplified in order to obtain
  $\ISATOM{\avalue}{\uco}$.
\end{itemize}
States are processed by means of a queue; the algorithm stops when all
states have been consumed without any modification to $\uco$, i.e.,
when no non-atomic $\avalue$ has been found.  As in $\FINDNDEPS$,
states are progressively restricted, and computations on
$\ABSEVAL{\exp}{\astate}$ are avoided if the desired property already
holds for $\astate' > \astate$.

\begin{figure}
  \begin{pseudocode}
   $\uco$ := $\uco_0$; // $\mbox{{\it the initial domain}}$
   repeat {
     $\mbox{inputQueue}$ := [$\astatei{p}$]; // $\mbox{{\it one-element queue}}$
     while ($\mbox{notEmpty}$($\mbox{inputQueue}$)) {
       $\astate$ := $\mbox{extract}$($\mbox{inputQueue}$);
       if ($\nexists V.\ \ACu{V}{\astate}$) then {
         if ($\ISATOM{\astate}{\uco}$) then { // $\mbox{{\it on }}\VARS{e} \smallsetminus X$
           // $\mbox{{\it at this point,}}V \mbox{ {\it is not atomic}}$
           $V$ := $\ABSEVAL{e}{\astate}$;
           $\uco$ := $\ATOMIZE{\uco}{V}$;
           // $\mbox{{\it the queue still has 1 element}}$
           $\mbox{inputQueue}$ := [$\astatei{p}$];
         } else { // $\{\astatei{1}..\astatei{k}\}\ \mbox{{\it is an }}X\mbox{{\it -covering of }}\astate$
           foreach($i$) {
             $\mbox{insertInQueue}$($\mbox{inputQueue}$,$\astatei{i}$)}}}}
   } until ($\uco\ \mbox{has not been modified in the while loop}$);
   return $\uco$; // $\mbox{{\it the domain s.t. }}\ANARROWDEPENDS{X}{}{\uco}{e}{\ov{\uco}}\ \mbox{{\it does not hold}}$
  \end{pseudocode}
  \caption{The $\EDEPALG$ algorithm.}\label{edep}
\end{figure}

The simplifying operator $\ATOMIZE{\uco}{\avalue}$ is a \emph{domain
  transformer}, and works by removing abstract values in order to
obtain $\ISATOM{\avalue}{\uco}$.  Formally,
\[ \uco'\ =\ \ATOMIZE{\uco}{\avalue}\ \EQDEF\ \{ \avaluee \in \uco\ |\
\avalue{\sqcap}\avaluee{=}\bot\ \vee\ \avalue{\leq}\avaluee \} \]

\noindent
The final $\uco$ is an approximation of the most precise $\uco'$
s.t.\ $\ANARROWDEPENDS{X}{}{\uco'}{\exp}{\ov{\uco'}}$ is false:

\begin{theorem}
  \label{theorem:nonDependency}
  $\uco\defi\EDEP{\exp}{\uco_0}{X}$ makes $e$ not narrow-dependent on
  $X$.  In other words: the final $\uco$ satisfies
  non-narrow-dependency of $\exp$ on $X$, that is, for every $\astate$
  which is atomic on $\VARS{\exp} \smallsetminus X$,
  $\ABSEVAL{\exp}{\astate}$ is atomic.
\end{theorem}

\begin{proof}
  The algorithm halts if, in processing $\astatei{n}$, $\uco$ is not
  changed.  Processing $\astatei{n}$ involves computing
  $\ABSEVAL{\exp}{}$ on sub-states when required, in order to prove
  the atomicity property on every concrete state represented by
  $\astatei{n}$ (to this end, we exploit monotonicity of
  $\ABSEVAL{\exp}{}$ on states).  This is precisely obtained if every
  state is removed from the queue before any modification to $\uco$
  occurs.
\end{proof}

\noindent
On the practical side, the loss of precision in abstract computations
may lead to remove more abstract values than strictly necessary from
the semantic point of view.  It is important to note that $\EDEPALG$
works as long as $\AC{}$ can be computed on the initial domain (in
this case, no problems arise in subsequent computations, since the
``complexity'' of $\uco$ can only decrease).  This can possibly happen
even if $\uco_0$ is infinite (see the end of Section
\ref{section:algorithmicIdeasForCheckingNdep}).  Moreover, unlike
$\FINDNDEPS$, there is no reasonable trivial counterpart, since any
brute-force approach would be definitely impractical.

\section{The quest for Abstract Slices}
\label{sec:theQuestForAbstractSlices}
 % computingAbstractSlices
This section introduces an algorithm for computing conditioned
abstract slices, based on abstract dependencies and the notion of
\emph{agreement} between states.

As explained in Section \ref{sec:practicalIssuesAndOptimizations},
this way to compute slices relies on \emph{a priori} knowledge of the
properties which will be of interest for the analysis.  In most cases,
the majority of the abstract domains to be taken into account are
quite simple (e.g., nullity).  On the other hand, Section
\ref{sec:constructiveApproachesToAbstractProgramSlicing} presents a
general way to compute abstract dependencies on more complex domains.
Indeed, those algorithms can be used here, and their complexity is
acceptable if small domains are dealt with.

The slices obtained by following this methodology will have the
standard form of backward slicing; \emph{predicates} or
\emph{conditions} on states can be specified, in the style of
conditioned slicing \cite{Conditioned}.  In this case, for a predicate
$\PRED$, the judgment $\state \models \PRED$ means that the state
$\state$ satisfies $\PRED$.  A predicate $\PRED$ at a certain program
point may include either user-provided or statically inferred
information: for example, after a \CODE{x:=new C()} statement,
judgments like ``\CODE{x} is not null'', ``\CODE{x} is not cyclic'',
``\CODE{x} is not sharing with \CODE{y}'' could be provided depending
of the kind of static analyses available (nullity, sharing, cyclicity,
etc.).  A way to decide which statements have to be included in an
abstract slice consists of two main steps:
\begin{itemize}
\item for each statement $s$, a specific static-analysis algorithm
  provides information about the relevant data \emph{after} that
  statement (below, the \emph{agreement}), according to the slicing
  criterion and the program code;
\item if the execution of $s$ \emph{does not affect} its corresponding
  agreement (i.e., some condition on states which must hold after
  $s$), then $s$ can be removed from the slice.
\end{itemize}

\begin{example}
  \label{ex:conditionsToSlice}
  Consider the following code fragment, and suppose that the slicing
  criterion is the nullity of \xx at the end:

  \vspace{-5mm}

  \begin{lstlisting}[firstnumber=21]
   ...
   y.f := $\mathit{exp}$;
   x := y;
  \end{lstlisting}

  \noindent
  The field update on \yy (line 22) can be removed from the slice
  because
  \begin{itemize}
  \item the question about the nullity of \xx after line $23$ is
    equivalent to the question about the nullity of \yy after line
    $22$; and
  \item the field update at line $22$ does not affect the nullity of
    \yy.
  \end{itemize}
\end{example}

\noindent
The rest of this section formalizes how these two main steps are
carried out.

\subsection{The logic for propagating agreements}
\label{section:theLogicForPropagatingAgreements}

This section describes how agreements are defined and propagated via a
system of logical rules: the \GSYSTEM\footnote{A version of this
  system was introduced in previous work \cite{Zanardini} as the
  $\RULENAME{a}$-system.  However, there are many differences between
  both systems, mainly due to the changes in the language under study
  (for example, variables are taken into account here instead of a
  more involved notion of \emph{pointer expression}).}.  Hoare-style
\emph{triples} \cite{Hoa69} are used for this purpose, in the spirit
of the \emph{weakest precondition calculus} \cite{Dij75}.

\begin{mydefinition}
  \label{def:agreement}
  An \emph{agreement} $\AGREEM$ is a set of conditions
  $\AGRS{\uco}{\bar{x}}$ where each uco $\rho$ involves a sequence of
  variables $\bar{x}$ (most usually, just one variable), and all
  conditions involve mutually disjoint sets of variables.  Two states
  $\state_1$ and $\state_2$ are said to \emph{agree} on $\AGREEM$,
  written $\AGREEM(\state_1,\state_2)$, iff, for every
  $\AGRS{\uco}{\bar{x}}$ in $\AGREEM$, $\uco(\state_1(\bar{x})) =
  \uco(\state_2(\bar{x}))$, where notation is abused by taking
  $\state_i(\bar{x})$ as the sequence of values of variables $\bar{x}$
  in $\state_i$, and $\uco(\state_i(\bar{x}))$ as the application of
  $\uco$ (which could be a relational domain) to the elements of such
  a sequence.

  Agreements are easily found to form a lattice, and a partial order
  $\sqsubseteq$ can be defined: $\AGREEM' \leq \AGREEM''$ iff, for
  every $\state_1$, $\state_2$ such that
  $\AGREEM'(\state_1,\state_2)$, then $\AGREEM''(\state_1,\state_2)$.
  Moreover, an intersection operator is induced by the partial order:
  $\AGREEM' \sqcap \AGREEM''$ is the greatest agreement which is less
  than or equal to both.
\end{mydefinition}

In the following, $\AGREEM(x)$ will be the uco corresponding to the
condition $\AGRS{\uco}{x}$ in $\AGREEM$, or $\TOPDOM$ if no condition
on $x$ belongs to $\AGREEM$.  For the sake of simplicity, the
discussion will be limited to domains each involving one single
variable.  In this case, ordering amounts to the following: $\AGREEM_1
\sqsubseteq \AGREEM_2$ if $\forall x.\AGREEM_1(x) \sqsubseteq
\AGREEM_2(x)$, where $\AGREEM_1(x) \sqsubseteq \AGREEM_2(x)$ is the
comparison on the \emph{precision} of abstract domains, meaning that
$\AGREEM_1(x)$ is \emph{more precise} than $\AGREEM_2(x)$.

\begin{example}
  \label{ex:agreement} Let $\state_1
  = \{\BIND{\mbox{\nn}}{2}, \BIND{\mbox{\ii}}{3}, \BIND{\mbox{\xx}}{\mbox{\CODE{null}}}\}$
  and $\state_2
  = \{\BIND{\mbox{\nn}}{0}, \BIND{\mbox{\ii}}{4}, \BIND{\mbox{\xx}}{\mbox{\CODE{null}}}\}$
  be two states.  Then, they agree on $\AGREEM
  = \{ \AGRS{\PARDOM}{\langle \mbox{\nn}\rangle}, \AGRS{\NULLDOM}{\langle \mbox{\xx}\rangle} \}$
  since
  \begin{itemize}
  \item \nn has the same parity in both states;
  \item \xx is \CODE{null} (therefore, it has ``the same nullity'')
    in both states; and
  \item there is no condition on \ii.
  \end{itemize}
  On the other hand, these states do not agree on $\AGREEM' = \{
  \AGRS{\PARDOM}{\langle \mbox{\ii}\rangle}, \AGRS{\NULLDOM}{\langle
    \mbox{\xx}\rangle} \}$ because \ii is odd in $\state_1$ whereas it
  is even in $\state_2$.
\end{example}

In a triple $\TRIPLEB{\AGREEM}{}{s}{\AGREEM'}$, the pre-condition
$\AGREEM$ is the weakest agreement on two states before a statement
$s$ such that the agreement specified by the post-condition $\AGREEM'$
holds after the statement.  Predicates on states can be used, so that
triples are, actually, 4-tuples which only take into account a subset
of the states.  Formally, the 4-tuple (or \emph{augmented triple})
$\TRIPLEB{\AGREEM}{\PRED}{s}{\AGREEM'}$ (where the $\true$ predicate
is often omitted) holds if, for every $\state_1$ and $\state_2$,
\[
  \begin{array}{rll}
    \state_1 \models \PRED\ \ \wedge\ \ \state_2 \models \PRED
    \ \ \wedge\ \ \AGREEM(\state_1,\state_2) & \Rightarrow &
    \AGREEM'(\SEMANTICS{s}(\state_1),\SEMANTICS{s}(\state_2))
  \end{array}
\]
\noindent
The rules of the \GSYSTEM are shown in Figure \ref{fig:gRules}.  The
\emph{transformed predicate} $s(\PRED)$ is one which is guaranteed to
hold after a statement $s$, given that $\PRED$ holds before, in the
style of \emph{strongest post-condition calculus} \cite{DS90,BN98}.
For example, if $\PRED = x\geq 0$, then the condition $s(\PRED) =
x\geq 1 \wedge y\neq \nil$ certainly holds after $s \equiv
\mbox{\CODE{x:=x+1; y:=new C()}}$.  The way predicates are transformed
is outside the scope of this paper; however, the cited works
introduced calculi for computing such strongest post-conditions.  In
the absence of such tools, $\true$ is always a consistent choice
(precision, but not soundness, may be affected since the set of states
to be considered grows larger, and to prove useful results could
become harder).

The \GSYSTEM is tightly related to narrow non-interference
\cite{GM04popl,GM04CSL}.  These works define a similar system of
rules, the $\RULENAME{n}$-rules, for assertions
$\NANI{\eta}{s}{\eta'}$, where $\eta$ and $\eta'$ are basically the
(tuples of) abstract domains corresponding to, resp., $\AGREEM$ and
$\AGREEM'$.  The systems differ in that:
\begin{itemize}
\item the use of pointers requires the rules for assignment to account
  for sharing, while $\RULENAME{n}$-rules only work on integers;
\item in the present approach, domains are supposed to be
  partitioning, so that there is no need to include explicitly the
  $\Pi$ operator (Section \ref{Sect:partit});
\item the \GSYSTEM~does not distinguish between \emph{public} and
  \emph{private} since this notion is not relevant in slicing;
\item the rule for conditional is not included in the
  $\RULENAME{n}$-system; indeed, this is quite a tricky rule, and, in
  general, expressing a conditional with loops and using the rule
  $\RULENAME{n6}$ for loops results in inferring less precise
  assertions;
\item in the $\RULENAME{n}$-system, predicates $\PRED$ on program
  states are not supported.
\end{itemize}

In the following, each rule of the \GSYSTEM is discussed.
Importantly, this rule system relies on the computation of
\emph{property preservation}.  We rely on a rule system which soundly
computes whether executing a statement affects properties of some
variables: the judgment

\centerline{$\PRESERVESB{\PRED}{s}{\AGREEM}$}
\noindent can only obtained by using those rules if it is possible to
prove that executing $s$ in a state $\state$ satisfying the condition
$\PRED$ results in a final state $\state'$ that is equal to the
initial $\state$ with respect to the agreement $\AGREEM$, i.e.,
$\AGREEM(\state,\state')$.  In other word, the statement is equivalent
to \CODE{skip} with respect to the properties of interest.  Note that,
here, the agreement is not used to compare two states at the same
program point; rather, it takes as input the states before and after
executing a statement.  The rule system for proving property
preservation is explained after introducing the \GSYSTEM (Section
\ref{sec:thePSystem}).

 % gsystem
\begin{figure}
  \begin{center}
      \[
      \GRULE{\PRESERVESB{\PRED}{s}{\AGREEM}}
      {\TRIPLEB{\AGREEM}{\PRED}{s}{\AGREEM}}{pp}
      \qquad\qquad
      \GRULE{}{\TRIPLEB{\AGREEM}{\PRED}{\mbox{\CODE{skip}}}{\AGREEM}}{skip}
      \]
      \[
      \GRULE{\TRIPLEB{\AGREEM}{\PRED}{s}{\AGREEM'} \qquad
        \TRIPLEB{\AGREEM'}{s(\PRED)}{s'}{\AGREEM''}}
      {\TRIPLEB{\AGREEM}{\PRED}{s\mbox{\CODE{;}}s'}{\AGREEM''}}{concat}
      \qquad\qquad
      \GRULE{\forall x.~\AGREEM_{ID}(x){=}\IDDOM\quad s\neq\mbox{\CODE{read($\cdot$)}}}{\TRIPLEB{\AGREEM_{ID}}{\PRED}{s}{\AGREEM'}}{id}
      \]
      \[
      \GRULE{ \TRIPLEB{\AGREEM_2}{\PRED_2}{s}{\AGREEM'_2}\qquad \AGREEM_1
        \sqsubseteq \AGREEM_2\qquad \AGREEM'_2 \sqsubseteq\AGREEM'_1
        \qquad \PRED_1\Rightarrow\PRED_2
      }{\TRIPLEB{\AGREEM_1}{\PRED_1}{s}{\AGREEM'_1}}{sub}
      \]
      \[
      \GRULE{\forall y.~\lnot
        \left(\ATOMDEPENDS{y}{}{\AGREEM'(x)}{e}{\AGREEM}\right)^\PRED
        \qquad
        \forall y \neq x.~\AGREEM(y) = \AGREEM'(y)}
           {\TRIPLEB{\AGREEM}{\PRED}{x\mbox{\CODE{:=}}e}{\AGREEM'}}{assign}
      \]
      \[
      \GRULE{
        \begin{array}{c@{~}r@{~}l}
          ({*}) & \forall y \in \DALIAS{x}. & \forall
          \state_1 \models \PRED,\state_2 \models
          \PRED. \\
          & & ~~\AGREEM(\state_1,\state_2)
          \Rightarrow \\
          & & ~~\AGREEM'(\state_1[y.f \leftarrow
            \SEMANTICS{e}(\state_1)],\state_2[y.f \leftarrow
            \SEMANTICS{e}(\state_2)]) \\
          ({*}{*}) & \forall y \in \SHARE{x}. &
          \forall \bar{g}, \forall
          \state_1 \models \PRED,\state_2 \models
          \PRED. \\
          & & ~~\AGREEM(\state_1,\state_2)
          \Rightarrow \\
          & & ~~\AGREEM'(\state_1[y.\bar{g} \leftarrow
            \SEMANTICS{e}(\state_1)],\state_2[y.\bar{g} \leftarrow
            \SEMANTICS{e}(\state_2)]) \\
          ({*}{*}{*}) & \forall y \notin \DALIAS{x}. & \AGREEM(y)
          \sqsubseteq \AGREEM'(y)
        \end{array}
      }{
        \TRIPLEB{\AGREEM}{\PRED}{x.f\mbox{\CODE{:=}}e}{\AGREEM'}
      }{fassign}
      \]
      \[
      \GRULE{\TRIPLEB{\AGREEM}{\PRED}{s_t \diamond s_f}{\AGREEM'}}
      {\TRIPLEB{\AGREEM}{\PRED}{\mbox{\CODE{if ($b$) }~$s_t$~\CODE{else}~$s_f$}}{\AGREEM'}}{if1}
      \]
      \[
      \GRULE{\TRIPLEB{\AGREEM_t} {\PRED\wedge b} {s_t} {\AGREEM'}\qquad
        \TRIPLEB{\AGREEM_f} {\PRED\wedge\lnot b} {s_f}
        {\AGREEM'}}{\TRIPLEB{\AGREEM_b\sqcap\AGREEM_t\sqcap\AGREEM_f}{\PRED}
        {\mbox{\CODE{if ($b$) }~$s_t$~\CODE{else}~$s_f$}} {\AGREEM'}}{if2}
      \]
      \[
      \GRULE{\PRED \Rightarrow s(\PRED) \quad
        \TRIPLEB{\AGREEM\sqcap\AGREEM_b}{\PRED\wedge
          b}{s}{\AGREEM\sqcap\AGREEM_b}}
      {\TRIPLEB{\AGREEM\sqcap\AGREEM_b}{\PRED}{
          \mbox{\CODE{while ($b$) }~$s$}}{\AGREEM\sqcap\AGREEM_b}}{while}
      \]
  \end{center}
  \caption{The \GSYSTEM}
  \label{fig:gRules}
\end{figure}

\subsubsection{Rule $\GRULENAME{pp}$}

This rule makes the relation between property preservation and the
\GSYSTEM more clear.  The triple $\TRIPLE{\AGREEM}{s}{\AGREEM}$
amounts to say that two executions agree \emph{after} $s$, provided
they agree \emph{before} on the same $\AGREEM$.  On the other hand,
the preservation of $\AGREEM$ on $s$ means that any state
\emph{before} $s$ agrees on $\AGREEM$ with the corresponding state
\emph{after} $s$.  Property preservation is a stronger requirement
than the mere propagation $\TRIPLE{\AGREEM}{s}{\AGREEM}$ of
agreements, so that this rule is sound.  In fact, if
$\AGREEM(\state_1,\state_2)$ and both
$\AGREEM(\state_1,\SEMANTICS{s}(\state_1))$ and
$\AGREEM(\state_2,\SEMANTICS{s}(\state_2))$ hold, then
$\AGREEM(\SEMANTICS{s}(state_1),\SEMANTICS{s}(\state_2))$ follows,
which is, equivalent, by definition, to
$\TRIPLE{\AGREEM}{s}{\AGREEM}$.

\begin{example}
  Let the parity of $x$ be the property of interest.  In this case,
  \IMPASSIGN{x}{x+1} does not preserve the parity of $x$, but two
  initial states agreeing on $\AGRS{\PARDOM}{x}$ lead to final states
  which still agree on it.  Therefore,
  $\TRIPLE{\AGRS{\PARDOM}{x}}{\mbox{\IMPASSIGN{x}{x+1}}}{\AGRS{\PARDOM}{x}}$
  holds.  On the other hand, \IMPASSIGN{x}{x+2} also satisfies a
  stronger requirement: that $\PARDOM(x)$ does not change.  Therefore,
  besides having
  $\TRIPLE{\AGRS{\PARDOM}{x}}{\mbox{\IMPASSIGN{x}{x+2}}}{\AGRS{\PARDOM}{x}}$,
  the judgment
  $\PRESERVES{\mbox{\IMPASSIGN{x}{x+2}}}{\AGRS{\PARDOM}{x}}$ is also
  true.
\end{example}

\subsubsection{Rules $\GRULENAME{skip}$, $\GRULENAME{concat}$,
  $\GRULENAME{id}$, $\GRULENAME{sub}$}

The $\GRULENAME{skip}$ rule describes \emph{no-op}.  The assertion
holds for every $\AGREEM$ and $\PRED$ since
$\SEMANTICS{\mbox{\CODE{skip}}}(\state) = \state$.

$\GRULENAME{concat}$ is also easy: soundness holds by transitivity
(note also the use of $s(\PRED)$ to propagate conditions of states).

Rule $\GRULENAME{id}$ can be used when nothing else can be proved: it
always holds because execution is deterministic, so that two execution
starting from two states which are equal\footnote{The notion of
  equality on references and objects is recalled in Example
  \ref{ex:referenceAbstractDomains}.} on all variables will end in a
pair of states agreeing on any abstraction.  Note that, as pointed out
in Section \ref{sec:theProgrammingLanguage}, \CODE{read} statements
are supposed only to appear at the beginning of a program; therefore,

Finally, in $\GRULENAME{sub}$, remember that $\sqsubseteq$ is the
partial order on agreements.

\subsubsection{Rule $\GRULENAME{assign}$}

This rule means that, given a statement $x
\mbox{\CODE{:=}} e$, any agreement $\AGREEM$ which satisfies the two
conditions of the above part of the rule is a sound pre-condition for
the post-condition $\AGREEM'$.  The conditions are (1) that, given two
states which agree on $\AGREEM$, the computed results for the
expression $e$ in both states are abstracted by $\AGREEM'(x)$ to the
same abstract value; and (2) that $\AGREEM$ is as precise as
$\AGREEM'$ on all variables but $x$.  The first condition is
represented in terms of Definition
\ref{def:atomicAbstractDependencies}, and the superscript $\PRED$
indicates that only states satisfying $\PRED$ have to be considered.
Such a condition can be easily shown to imply the formula
\[ F
~~=~~ \forall \state_1 \models \PRED,\state_2 \models \PRED.~\left(\AGREEM(\sigma_1,\sigma_2) \Rightarrow
(\AGREEM'(x))(\SEMANTICS{e}(\state_1)) =
(\AGREEM'(x))(\SEMANTICS{e}(\state_2))\right) \] Note that, by Proposition~\ref{theorem:equivalenza}, the absence of abstract dependencies
w.r.t.~Definition \ref{def:atomicAbstractDependencies}
(\emph{Atom-dep}) implies the absence of abstract dependencies
w.r.t.~Definition \ref{def:abstractDependencies} (\emph{Ndep}) which,
in turn, implies $F$.  This clarifies the relation between abstract
dependencies and the computation of a slice.  Obviously, the second
condition guarantees that, for all variables which are not updated by
the assignment, the agreement required by $\AGREEM$ still holds when
$\AGREEM'$ is considered.

\subsubsection{Rule $\GRULENAME{fassign}$}

This rule accounts for the modification of the data structure pointed
to by a variable by means of a field update.  In the following, that a
variable is \emph{affected} means that the data structure pointed to
by it is updated.  Given a field update on a variable $x$, some other
variables (i.e., the data structures pointed to by them) could be
affected.  There exists a well-known static analysis which tries to
detect which variables point to a data structure which is updated by a
field update on $x$: this analysis is known as \emph{sharing analysis}
\cite{DBLP:conf/sas/SecciS05,ZanardiniG15sh}, and usually comes as
\emph{possible-sharing} analysis, where the set of variables which
\emph{could} be affected by a field update is over-approximated.
Moreover, \emph{aliasing analysis} \cite{Hind2001} can be used in
order to compute the set of variables pointing exactly (and directly)
to the same location as $x$; in this case, \emph{definite-aliasing}
analysis makes sense, which under-approximates the set of variables
which
\emph{certainly} alias with $x$.  According to the result of these
analyses, reference variables can be partitioned in three categories:
(1) variables which certainly alias with $x$, so that they can be
guaranteed to be updated in their field $f$; (2) variables which could
share with $x$, so that they could be affected by the update in many
ways; and (3) variables which certainly do not share with $x$, so that
they are unaffected by the update.  Let $\SHARE{x}$ be the set of
variables possibly sharing with $x$ before the update, and
$\DALIAS{x}$ be the set of variables definitely aliasing with $x$.  In
the absence of a definite-aliasing analysis, then $\DALIAS{x}$ can be
safely taken as $\{x\}$.

\begin{example}
  \label{ex:sharingFlavours}
  Consider the following code fragment:
  {\em \begin{lstlisting}[firstnumber=10]
   if (...) then { y.f := x; } else { y.f := z; }
   w := x;
   x.g := $e$;
  \end{lstlisting}}

  \noindent Suppose that, initially, no variable is sharing with any
  other variable (i.e., there is no overlapping between data
  structures referred by different variables), and that the truth
  value of the boolean guard cannot be determined statically, so that
  both branches of the conditional statement have to be considered as
  possible executions.  In this case, the sharing and aliasing
  information before line 12 is as follows: \[ \SHARE{\mbox{\xx}} = \{
  \mbox{\xx}, \mbox{\yy}, \mbox{\ww} \} \qquad\qquad
  \DALIAS{\mbox{\xx}} = \{ \mbox{\xx},\mbox{\ww} \} \] Note that every
  variable is aliasing with itself (none of them is \CODE{null}), \ww
  is certainly aliasing with \xx because of line 11, and \yy is
  possibly sharing with \xx (the actual sharing depends on the value
  of the guard).
\end{example}

The $\GRULENAME{fassign}$ rule comes with three pre-conditions.
Pre-condition $({*})$ only applies to variables in category (1): those
definitely aliasing with $x$.  Pre-condition $({*}{*})$ applies to
category (2), while pre-condition $({*}{*}{*})$ applies to categories
(2) and (3).

\begin{itemize}
\item $({*})$ requires that updating the field $f$ of the location
  pointed to by $x$ (and all variables definitely aliasing with it)
  leads to an agreement on $\AGREEM'$, provided that the initial
  states agree on $\AGREEM$.
\item $({*}{*})$ is similar, but states that the agreements must hold
  for \emph{every} sequence of field selectors $\bar{g}$\footnote{In
    the following, the notation $\FSEQ{.f_1.f_2......f_n}$ (starting
    with a dot is intentional) will be used to represents sequences of
    field selectors.}, possibly the empty sequence.  This is needed
  since, given that some $y$ may share with $x$, it cannot be known
  which fields of $y$ will be updated, and how.  In practice, only the
  sequences of field selectors which are compatible with the class
  hierarchy of the program under study have to be considered, as shown
  in Example \ref{ex:unfeasibleFieldSequences}.
\item $({*}{*}{*})$ applies to variables that could be unaffected by
  the update, i.e., variables in categories (2) and (3) (note that the
  conditions $y \in \SHARE{x}$ in $({*}{*})$ and $y \notin \DALIAS{x}$
  in $({*}{*}{*})$ are not mutually exclusive, so that variables in
  category (2) satisfy both).  The relation between $\AGREEM$ and
  $\AGREEM'$ is clear in this case, as agreement on $\AGREEM$ must
  entail agreement on $\AGREEM'$.
\end{itemize}

\begin{example}[Infeasible sequences of field selectors]
  \label{ex:unfeasibleFieldSequences}
  In a Java-like language, a sequence like
  $\FSEQ{.\mbox{\CODE{g}}.\mbox{\CODE{f}}}$ is not compatible with the
  following class hierarchy since the class of \CODE{g} is \CODE{D}
  whereas \CODE{f} is declared in \CODE{C}:
  \begin{center}
    \begin{lstlisting}[numbers=none]
   class C { D f; D g; }        class D { D h; }
    \end{lstlisting}
  \end{center}
\end{example}

\begin{example}
  \label{ex:fassignNullity}
  Consider the statement $s \equiv \CODE{x.f := y}$.  Let $\PRED =
  \true$ and the agreement $\AGREEM'$ after $s$ be $\{
  \AGRS{\NULLDOM}{\mbox{\xx}}, \AGRS{\NULLDOM}{\mbox{\zz}} \}$.  Let
  also $\SHARE{\mbox{\xx}}$ before $s$ be $\{ \mbox{\xx} \}$.  In this
  case, an agreement $\AGREEM$ which satisfies the judgment
  \[ \TRIPLEB{\AGREEM}{\true}{\mbox{\CODE{x.f:=y}}}{\AGREEM'} \]
  can be the same $\{ \AGRS{\NULLDOM}{\mbox{\xx}}, \AGRS{\NULLDOM}{\mbox{\zz}} \}$
  because
  \begin{itemize}
  \item \zz is unaffected by the update (it belongs to category (3)),
    so that $\AGREEM(\mbox{\zz})$ must be at least as precise as
    $\AGREEM'(\mbox{\zz})$;
  \item category (2) contains no variables; and
  \item category (1) only includes \xx itself, and the nullity of \xx
    is clearly unaffected by the update.
  \end{itemize}
\end{example}

One may think that the universal quantification on field sequences in
pre-condition $({*}{*})$ results in an unacceptable loss of precision.
Indeed, to require that all possible updates to possibly-sharing
variables preserve the desired agreements seems to be too strict.
However, there are a number of things to be considered:
\begin{itemize}
\item A closer look to the rule shows that there is no easier way to
  account for sharing if traditional sharing analysis is used.
\item Example \ref{ex:fassignNullity} shows that it is still possible
  to get meaningful results on domains working on pointer variables.
\item The state of the art in static analysis of object-oriented
  languages indicates that abstract domains on pointers are likely to
  be quite simple ($\NULLDOM$ being one of them).
\item There is recent work \cite{ZanardiniG15sh} introducing a more
  precise, \emph{field-sensitive} sharing analysis which computes
  \emph{how} variables share: this analysis is able to detect which
  \emph{fields} are or are not involved in paths in the heap
  converging from two variables to a shared location.  In order to
  keep the discussion as simple as possible, the definition of
  $\GRULENAME{fassign}$ given in Figure \ref{fig:gRules} uses
  traditional sharing analysis.  However, the impact of
  field-sensitive sharing analysis is discussed in Section
  \ref{sec:practicalIssuesAndOptimizations}, where a refined version
  of $\GRULENAME{fassign}$ is given.
\end{itemize}

The domain of cyclicity introduced in Section
\ref{ex:referenceAbstractDomains} represents information about
\emph{data structures} in the heap, not only program variables.  In
this sense, field updates have to be regarded as potentially affecting
the propagation of agreements.

\begin{example}
  \label{ex:cyclicityDomain} Let $s$ be, again, the
  statement \CODE{x.f:=y}, and $\AGREEM'$ be $\{ \AGRS{\CYCLEDOM}{x}
  \}$; i.e., the interest is on the cyclicity of the data structure
  pointed to by $x$ after executing $s$.  Suppose also that $x$ and
  $y$ are certainly not sharing before $s$, and that an object whose
  type is compatible with $x$ has two reference fields $f$ and $g$.
  Then, the agreement $\AGREEM = \{
  \AGRS{\CYCLEDOM}{x},\AGRS{\CYCLEDOM}{y} \}$ is \emph{not} a correct
  precondition for the Hoare tuple to hold, since there can be two
  states $\state_1$ and $\state_2$ (Figure \ref{fig:cyclicity}) such
  that \begin{itemize} \item the data structure corresponding to $y$
    is acyclic, and equal in both states (therefore, there is an
    agreement on the cyclicity of $y$); \item the data structure
    corresponding to $x$ is cyclic in both $\state_1$ and $\state_2$,
    but (a) in $\state_1$ there is only a cycle originating from the
    location bound to $x.f$; and (b) in $\state_2$ there is a cycle
    originating from $x.f$ and another one originating from
    $x.g$.  \end{itemize} In this case, $\AGREEM(\state_1,\state_2)$
  holds but the resulting final states $\state'_1$ and $\state'_2$ do
  \emph{not} agree on $\AGREEM'$ since $x$ is acyclic in $\state'_1$
  (the only cycle has been broken) while it is still cyclic in
  $\state'_2$.  This behavior is captured by $\GRULENAME{fassign}$
  because condition $({*})$ applied to $x$ itself (which, by
  hypothesis, is the only variable in $\DALIAS{x}$) does not hold, so
  that the augmented triple cannot be proven.  On the other hand, $\{
  \AGRS{\IDDOM}{x},\AGRS{\CYCLEDOM}{y} \}$ would be a correct
  precondition.

  \begin{figure}
    \begin{center}
      \begin{tikzpicture}
        \node (yv1) at (-2,-0.2) {$y$};
        \node[draw,minimum height=0.4cm,minimum width=0.8cm] (y1) at (-2,-1) {};
        \draw[dotted,->] (yv1) -- (y1);
        \node (xv1) at (0,-0.2) {$x$};
        \node[draw,minimum height=0.4cm,minimum width=0.8cm] (x1) at (0,-1) {};
        \node[draw,minimum height=0.4cm,minimum width=0.8cm] (xf1) at (-0.5,-2) {};
        \node[draw,minimum height=0.4cm,minimum width=0.8cm] (xg1) at
        (0.5,-2) {};
        \draw[dotted,->] (xv1) -- (x1);
        \draw[->] (x1) -- node[left] {$f$} (xf1);
        \draw[->] (x1) -- node[right] {$g$} (xg1);
        \draw[->,dashed] (xf1) .. controls (-1.5,-3.5) and (0.5,-3.5)
        .. (xf1);
        \node at (-1.35,-2.7) {\emph{\scriptsize{(cycle)}}};
        \draw (-2.5,0) rectangle (1,-3.3);
        \node at (1.3,-3.1) {$\state_1$};

        \draw[->,thick] (1.3,-1.65) -- node[above] {\CODE{x.f:=y}} (4.3,-1.65);

        \node (pyv1) at (5,-0.2) {$y$};
        \node[draw,minimum height=0.4cm,minimum width=0.8cm] (py1) at (5,-1) {};
        \draw[dotted,->] (pyv1) -- (py1);
        \node (pxv1) at (7,-0.2) {$x$};
        \node[draw,minimum height=0.4cm,minimum width=0.8cm] (px1) at (7,-1) {};
        \node[draw,minimum height=0.4cm,minimum width=0.8cm] (pxg1) at
        (7.5,-2) {};
        \draw[dotted,->] (pxv1) -- (px1);
        \draw[->] (px1) -- node[above] {$f$} (py1);
        \draw[->] (px1) -- node[right] {$g$} (pxg1);
        \draw (4.5,0) rectangle (8,-3.3);
        \node at (8.3,-3.1) {$\state'_1$};
      \end{tikzpicture}

      \begin{tikzpicture}
        \node (yv1) at (-2,-0.2) {$y$};
        \node[draw,minimum height=0.4cm,minimum width=0.8cm] (y1) at (-2,-1) {};
        \draw[dotted,->] (yv1) -- (y1);
        \node (xv1) at (0,-0.2) {$x$};
        \node[draw,minimum height=0.4cm,minimum width=0.8cm] (x1) at (0,-1) {};
        \node[draw,minimum height=0.4cm,minimum width=0.8cm] (xf1) at (-0.5,-2) {};
        \node[draw,minimum height=0.4cm,minimum width=0.8cm] (xg1) at
        (0.5,-2) {};
        \draw[dotted,->] (xv1) -- (x1);
        \draw[->] (x1) -- node[left] {$f$} (xf1);
        \draw[->] (x1) -- node[right] {$g$} (xg1);
        \draw[->,dashed] (xf1) .. controls (-1.5,-3.5) and (0.5,-3.5) .. (xf1);
        \draw[->,dashed] (xg1) .. controls (-0.5,-3.5) and (1.5,-3.5) .. (xg1);
        \node at (-1.4,-2.7) {\emph{\scriptsize{(cycles)}}};
        \draw (-2.5,0) rectangle (1,-3.3);
        \node at (1.3,-3.1) {$\state_2$};

        \draw[->,thick] (1.3,-1.65) -- node[above] {\CODE{x.f:=y}} (4.3,-1.65);

        \node (pyv1) at (5,-0.2) {$y$};
        \node[draw,minimum height=0.4cm,minimum width=0.8cm] (py1) at (5,-1) {};
        \draw[dotted,->] (pyv1) -- (py1);
        \node (pxv1) at (7,-0.2) {$x$};
        \node[draw,minimum height=0.4cm,minimum width=0.8cm] (px1) at (7,-1) {};
        \node[draw,minimum height=0.4cm,minimum width=0.8cm] (pxg1) at
        (7.5,-2) {};
        \draw[dotted,->] (pxv1) -- (px1);
        \draw[->] (px1) -- node[above] {$f$} (py1);
        \draw[->] (px1) -- node[right] {$g$} (pxg1);
        \draw[->,dashed] (pxg1) .. controls (6.5,-3.5) and (8.5,-3.5) .. (pxg1);
        \draw (4.5,0) rectangle (8,-3.3);
        \node at (8.3,-3.1) {$\state'_2$};
      \end{tikzpicture}
    \end{center}
    \caption{How two executions agreeing on $\{
      \AGRS{\CYCLEDOM}{x},\AGRS{\CYCLEDOM}{y} \}$ do not agree on $\{
      \AGRS{\CYCLEDOM}{x} \}$ after \CODE{x.f:=y}}
    \label{fig:cyclicity}
  \end{figure}
\end{example}

\subsubsection{Rules $\GRULENAME{if1}$ and $\GRULENAME{if2}$}

In a conditional
\CODE{if ($b$)} $s_t$ \CODE{else} $s_f$
there are two possibilities.  Rule $\GRULENAME{if1}$ states that an
input agreement which induces the output one whichever path is taken
is a sound precondition.  Here, the assertion
$\TRIPLEB{\AGREEM}{\PRED}{s_t \diamond s_f}{\AGREEM'}$ means that
\[
  \begin{array}{rl}
    \forall \state_1,\state_2. & \AGREEM(\state_1,\state_2) \wedge
    \state_1 \models \PRED \wedge \state_2 \models \PRED \Rightarrow \\
    & \AGREEM'(\SEMANTICS{s_t}(\state_1), \SEMANTICS{s_t}(\state_2),
    \SEMANTICS{s_f}(\state_1), \SEMANTICS{s_f}(\state_2))
  \end{array}
\]
\noindent
where the judgment $\AGREEM'(\cdot,\cdot,\cdot,\cdot)$ means that all
four states agree on $\AGREEM'$.  This rule requires $\AGREEM'$ to
hold on the output state independently from the value of $b$.
Soundness is easy (note that the above assertion implies
$\TRIPLEB{\AGREEM}{\PRED}{s_t}{\AGREEM'}$ and
$\TRIPLEB{\AGREEM}{\PRED}{s_f}{\AGREEM'}$).

Note that such a $\AGREEM$ can always be found (in the worst case, it
assigns the identity upper closure operator $\IDDOM$ to each variable,
so that two states agree only if they are exactly equal).  However,
sometimes it can be more convenient to exploit information about $b$.
In such cases, $\GRULENAME{if2}$ can be applied, which means that the
initial agreement $\AGREEM_t \sqcap \AGREEM_f$ is strong enough to
verify the final one, provided the same branch is taken in both
executions, as $\AGREEM_b$ requires.  In fact, $\AGREEM_b$ is built
from $b$, and separates states according to its value:
\[
  \AGREEM_b(\state_1,\state_2) \ \ \Leftrightarrow
  \ \ \SEMANTICS{b}(\state_1) = \SEMANTICS{b}(\state_2)
\]
\noindent
The rule means that, whenever two states agree on the branch to be
executed, and the triples on both branches hold, the whole triple
holds as well.

\begin{lemma}[soundness of $\GRULENAME{if2}$]
  \label{lemma:aIfPPSoundness} If $\state_1$ and $\state_2$ both
  satisfy $\PRED$ and agree on
  $\AGREEM_b\sqcap\AGREEM_t\sqcap\AGREEM_f$, then the corresponding
  output states $\state'_1$ and $\state'_2$ agree on $\AGREEM'$ under
  the hypotheses of the rule.
\end{lemma}

\begin{proof}
  By hypothesis, the same branch is taken in both cases.  Conditions
  $\PRED\wedge b$ and $\PRED\wedge \lnot b$ are consistent since $s_t$
  (respectively, $s_f$) can only be executed when $b$ is true
  (respectively, false).  The agreement on $\AGREEM$ holds in both
  paths, so that the entire assertion is correct.
\end{proof}

Note that the rule to be chosen for the conditional depends on the
precision of the outcome: $\GRULENAME{if2}$ can be a good choice if
(1) it can be applied; and (2) the result is ``better'' than the one
obtained by $\GRULENAME{if1}$.  The second condition amounts to say
that, given the same final agreement $\AGREEM'$, the initial agreement
obtained by using $\GRULENAME{if2}$ is weaker (i.e., it is more likely
that two states agree on it) than the one obtained by using
$\GRULENAME{if1}$.

\begin{example}
  \label{ex:ifRules}
  Consider the code fragment

  \vspace{-5mm}

  \begin{lstlisting}[numbers=none]
   if (x>0) { x:=x+1; } else { x:=x-1; }
  \end{lstlisting}
  and let $\AGREEM = \{ \AGRS{\SIGNDOM}{x} \}$ be the agreement after
  the statement, i.e., the relevant property is the sign of \xx.  The
  rule $\GRULENAME{if2}$ is able to compute the same $\AGREEM$ as the
  input agreement because
  \begin{itemize}
  \item the triple $\TRIPLEB{\AGREEM} {\PRED\wedge b} {s_t} {\AGREEM}$
    holds since the condition $\PRED\wedge b$ guarantees that \xx is
    positive, and two states which agree on the sign before the
    increment will still agree after it (if \xx is positive in both
    states, then it will remain positive in both);
  \item similarly, the triple $\TRIPLEB{\AGREEM} {\PRED\wedge \lnot b}
    {s_f} {\AGREEM}$ also holds (if \xx is 0 in both states, then it
    will be negative in both; and if it is negative in both, it will
    remain negative in both);
  \item $\AGREEM_b$ is less precise than $\AGREEM$ (i.e., $\AGREEM
    \sqsubseteq \AGREEM_b$ since the latter only separates numbers
    into positive and non-positive), so that the input agreement
    $\AGREEM_b \sqcap \AGREEM \sqcap \AGREEM$ is equal to $\AGREEM$.
  \end{itemize}
  On the other hand, $\GRULENAME{if1}$ is not able to compute the same
  input agreement because the precondition
  $\TRIPLEB{\AGREEM}{\PRED}{s_t \diamond s_f}{\AGREEM}$ of the rule
  does not hold.  In fact, consider two states $\state_1 = \{
  \mbox{\xx} \la 1 \}$ and $\state_2 = \{ \mbox{\xx} \la 2 \}$: they
  agree on the sign of \xx, but
  $(\SEMANTICS{s_t}(\state_2))(\mbox{\xx})$ and
  $(\SEMANTICS{s_f}(\state_1))(\mbox{\xx})$ have different sign (the
  first is zero while the second is positive).
\end{example}

\subsubsection{Rule $\GRULENAME{while}$}

The meaning of the rule for loops can be understood by discussing its
soundness: if $\PRED$ is preserved after any iteration of the body,
and the agreement which is preserved by the body guarantees the same
number of iterations in both executions (i.e., it is more precise than
$\AGREEM_b$), then such an agreement is preserved through the entire
loop.

\begin{lemma}[soundness of $\GRULENAME{while}$]
  \label{lemma:aWhileSoundness}
  Let $\state_1^0$ and $\state_2^0$ satisfy $\PRED$, and agree on
  $\AGREEM\sqcap\AGREEM_b$.  Then, given $\state'_i =
  \SEMANTICS{\mbox{\CODE{while ($b$)}~$s_w$}}(\state^0_i)$, the result
  $(\AGREEM\sqcap\AGREEM_b)(\state'_1,\state'_2)$ holds.
\end{lemma}

\begin{proof}
  Let $\state_i^{n+1} = \SEMANTICS{s_w}(\state_i^n)$.  There are two
  cases:
  \begin{itemize}
  \item $\SEMANTICS{b}(\state_1^0) = \SEMANTICS{b}(\state_2^0) = {\it
    false}$: in this case, the body is not executed, and the result
    holds trivially;
  \item $\SEMANTICS{b}(\state_1^0) = \SEMANTICS{b}(\state_2^0) = {\it
    true}$: in this case, $\state_i^0 \models \PRED\wedge b$.  By the
    hypothesis of the rule, $\state_1^1$ and $\state_2^1$ agree on
    $\AGREEM\sqcap\AGREEM_b$, and $\PRED$ still holds since $\PRED
    \Rightarrow s_w(\PRED)$.
  \end{itemize}
  At every iteration, the hypotheses hold.  Moreover $\AGREEM_b$
  guarantees the same number of iterations in both executions.
  Consequently, for a terminating loop (non-termination is not
  considered), $\state_1^k$ and $\state_2^k$ will fall in the first
  case (false guard) after the same number $k$ of iterations.  These
  states are exactly $\state'_1$ and $\state'_2$, and agree on
  $\AGREEM\sqcap\AGREEM_b$ after the loop.
\end{proof}

\begin{theorem}[\GSOUNDNESS]
  \label{theorem:aSoundness}
  Let $s$ be a statement, $\AGREEM'$ be required after $s$, $\PRED$ be
  a predicate and $p$ be the program point before $s$.  Let also
  $\AGREEM$ be an agreement computed before $s$ by means of the
  \GSYSTEM.  Let $\trace_1$ and $\trace_2$ be two trajectories, and the
  states $\state_1 \in \trace_1[p]$ and $\state_2 \in \trace_2[p]$ satisfy
  $\AGREEM(\state_1,\state_2)$ and $\PRED$.  Then, the condition
  $\AGREEM'(\state'_1,\state'_2)$ holds, where $\state'_i =
  \SEMANTICS{s}(\state_i)$.
\end{theorem}

\begin{proof}
  Easy from Lemmas \ref{lemma:aIfPPSoundness} and
  \ref{lemma:aWhileSoundness}, and the discussion explaining each rule
  (especially, $\GRULENAME{fassign}$).
\end{proof}

\subsubsection{The \PSYSTEM}
\label{sec:thePSystem}

 % psystem
\begin{figure}
  \begin{center}
    \[ \PRULE{\PRESERVESB{\PRED'}{s}{\AGREEM'} \quad \PRED\Rightarrow
      \PRED' \quad \AGREEM' \sqsubseteq
      \AGREEM}{\PRESERVESB{\PRED}{s}{\AGREEM}}{weak} \qquad
    \PRULE{}{\PRESERVESB{\PRED}{\mbox{\CODE{skip}}}{\AGREEM}}{skip} \]
    \[ \PRULE{\PRESERVESB{\PRED}{e}{\AGREEM} \qquad \forall
      \state.~\state \models \PRED \Rightarrow \AGREEM(x)(\state(x))
      = \AGREEM(x)(\EVAL{e}{\state})}
       {\PRESERVESB{\PRED}{\mbox{\IMPASSIGN{$x$}{$e$}}}{\AGREEM}}{assign} \]
       \[ \PRULE{
         \begin{array}{c@{~}r@{~}l}
           ({*}) & \forall y \in \DALIAS{x}. & \forall \state
           \models \PRED.~\AGREEM(\state,\state[y.f \leftarrow
             \SEMANTICS{e}(\state)]) \\
           ({*}{*}) & \forall y \in \SHARE{x}. &
           \forall \bar{g}, \forall
           \state \models \PRED.~\AGREEM(\state,\state[y.\bar{g}
             \leftarrow \SEMANTICS{e}(\state)])
       \end{array}}{\PRESERVESB{\PRED}{x.f\mbox{\CODE{:=}}e}{\AGREEM}}{fassign}
       \]
       \[ \PRULE{\PRESERVESB{\PRED}{s_1}{\AGREEM} \qquad
         \PRESERVESB{s_1(\PRED)}{s_2}{\AGREEM}}
          {\PRESERVESB{\PRED}{s_1\mbox{\CODE{;}}s_2}{\AGREEM}}{concat} \]
          \[ \PRULE{\PRESERVESB{\PRED}{s_1}{\AGREEM_1} \qquad \PRESERVESB{\PRED}{s_2}{\AGREEM_2}}
             {\PRES{\mbox{\IMPIF{$b$}{$s_1$}{$s_2$}}}{\AGREEM_1 \sqcup \AGREEM_2}}{if} \qquad
             \qquad
             \PRULE{\PRESERVESB{\PRED}{s}{\AGREEM}}{\PRESERVESB{\PRED}{\mbox{\IMPWHILE{$b$}{$s$}}}{\AGREEM}}{while} \]
  \end{center}
  \caption{The \PSYSTEM}
  \label{fig:thePSystem}
\end{figure}

Property preservation can be proved by means of a rule system, the
\PSYSTEM (Figure \ref{fig:thePSystem}).  Most rules are
straightforward or very similar to \GSYSTEM rules, and characterize
when executing a certain statement preserves the properties
represented by an agreement $\AGREEM$ (i.e., for every variable $x$,
the property/uco $\AGREEM(x)$ is preserved).

For example, rule $\PRULENAME{assign}$ allows proving that a certain
agreement is preserved when the initial value of $x$ cannot be
distinguished from the value of the expression (i.e., the new value of
$x$), when it comes to the property $\AGREEM(x)$.  In rule
$\PRULENAME{fassign}$, a mechanism similar to $\GRULENAME{fassign}$ is
used: definite aliasing and possible sharing can be used to identify
which variables are affected by the field update.  As for
$\GRULENAME{fassign}$, an optimization based on field-sensitive
sharing analysis (Section \ref{sec:useOfFieldSensitiveSharing}) can be
introduced, which makes it easier to prove property preservation on
field updates.

Indeed, a number of optimizations can be applied to the \PSYSTEM (for
example, one could think that property preservation on
$s_1\mbox{\CODE{;}}s_2$ does not require the preservation of the same
properties on both statements separately).  However, this rule system
is not the central part of this paper, and Figure \ref{fig:thePSystem}
is just a sensible way to infer property preservation.

\subsection{Agreements and slicing criteria}
\label{section:agreementsAndSlicingCriteria}

This approach to compute abstract slices follows the standard
conditioned, non-iteration-count form of backward slicing.  Therefore,
a slicing criterion $\crit$ takes the form
$(\cI,\cX,\{n\}{\times}\NATURALS,\false,\cA)$ where $n$ is the last
program point, and $\cA$ is a sequence of ucos assigning a property to
each variable in $\cX$.  The initial agreement can be easily computed
from the criterion and is such that $\AGREEM = \{
\AGRS{\cA_{x}}{x}~|~x\in \cX \}$, where $\cA_{x}$ is the element of
$\cA$ corresponding to $x$.  This shows that there is a close relation
between this specific kind of slicing criteria and agreements, and in
the following, these concepts will be used somehow interchangeably in
informal parts.  It will be shown that this makes sense, i.e.,
criteria and agreements define tightly related notions.  Next
definition defines the correctness of an abstract slice of this kind,
where the slicing criterion is intentionally confused with an
agreement.

\begin{mydefinition}[Abstract slicing condition]
  \label{def:abstractSlicingCondition}
  Let $\prog^s$ be the slice of $\prog$ with respect to an agreement
  (criterion) $\AGREEM$.  In order for $\prog^s$ to be correct,
  $\SEMANTICS{\prog}(\state)$ and $\SEMANTICS{\prog^s}(\state)$ must
  agree on $\AGREEM$ for every initial $\state$:
  $\AGREEM(\SEMANTICS{\prog}(\state), \SEMANTICS{\prog^s}(\state))$.
\end{mydefinition}

\subsection{Erasing statements}
\label{sec:erasingStatements}

The main purpose of the \GSYSTEM is to propagate a final agreement
backwards through the program code, in order to have a specific
agreement attached to each statement\footnote{Here, a statement is not
  only a piece of code, but also a position (program point) in the
  program, so that no two statements are equal, even if they are
  syntactically identical.  For the sake of readability, the program
  point is left implicit.}.  This is done as follows: a program can be
seen as a sequence $\overline{s} = s_0;...;s_k$ of $k+1$ statements,
where each $s_i$ can be either a simple statement (skip, assignment,
field update) or a compound one (conditional or loop), containing one
(the loop body) or two (the branches of the conditional) sequences of
statements (either simple or compound, recursively).  The way to
derive an agreement for every statement in the program is depicted in
the pseudocode of Figure \ref{fig:agreementPropagation}.  The
procedure \CODE{labelSequence} takes as input
\begin{enumerate}
\item a sequence of statements (in the first call, it is the whole
  program code\footnote{Strictly speaking, the sequence of statements
    is the program without the initial sequence of \CODE{read}
    statements (remember that \CODE{read} statements are basically
    meant to provide the input).};
\item a pair of agreements: (2.a) the first one, $\AGREEM_{in}$, refers
  to the beginning of the sequence, and, in the first call, is such
  that the abstraction on each variable is $\IDDOM$; and (2.b) the
  second one, $\AGREEM_{out}$, is the desired final agreement, which
  corresponds to the slicing criterion as discussed in Section
  \ref{section:agreementsAndSlicingCriteria}; and
\item a predicate on states which is supposed to hold at the beginning
  of the sequence.
\end{enumerate}
\CODE{labelSequence} goes backward through the program code inferring,
for each $s_i$, an agreement $\AGREEM_i$ which corresponds to the
program point \emph{after} $s_i$.  $\AGREEM_k$ will be the same
$\AGREEM_{out}$, whereas, for each $i$, $\AGREEM_{i-1}$ will be
inferred by using the \GSYSTEM: more specifically, it is a (ideally,
the best) precondition such that the tuple
$\TRIPLEB{\AGREEM_{i-1}}{\PRED_{i-1}}{s_i}{\AGREEM_i}$ holds.  Note
that, since the initial $\AGREEM_0$ is the identity on all variables,
$\TRIPLEB{\AGREEM_{0}}{\PRED_{0}}{s_1;..;s_i}{\AGREEM_i}$ trivially
holds for every $i$ (execution is deterministic); however, the
$\AGREEM_{in}$ argument plays an important role when dealing with loop
statements.

Importantly, statements inside compound statements (e.g., assignments
contained in the branch of a conditional) are also labeled with
agreements.  This is done by calling \CODE{labelSequence} recursively.
Note that, in this case, if $\overline{s}_t$ and $\overline{s}_f$
are. respectively, the sequences corresponding to the ``then'' and
``else'' branch of a conditional statement $s_j$, then
\CODE{labelSequence} is called with second argument
$(\AGREEM_{in},\AGREEM_j)$; this is so because the state does not
change when control goes from the end of a branch to the statement
immediately after $s_j$.

The treatment of loops follows closely the definition of
$\GRULENAME{while}$.  In the augmented triple, $\AGREEM_{i-1}$ appears
before and after the statement; this is consistent with the rule.  The
condition $\AGREEM_{i-1} \sqsubseteq \AGREEM_i \sqcap \AGREEM_b$
guarantees that the
$\TRIPLEB{\AGREEM_{i-1}}{\PRED_{i-1}}{s_i}{\AGREEM_i}$ can be proven
by applying $\RULENAME{sub}$.  Moreover, the recursive call on the
body $\overline{s}_l$ has $(\AGREEM_{i-1},\AGREEM_{i-1})$ as its
second argument.

\begin{figure}
  \begin{pseudocode}
    procedure labelSequence(`$s_1;...;s_k$', $(\AGREEM_{in},\AGREEM_{out})$, $\PRED$) {
      $\AGREEM_0$ = $\AGREEM_{in}$;
      $\AGREEM_k$ = $\AGREEM_{out}$;
      every $\PRED_i$ is $\overline{s}(\PRED)$ where $\overline{s} = `s_1;...;s_i$';
      // (remember the transformed predicate $s(\PRED)$)
      for $i$ = $k$ downto $1$ {
        if ($s_i$ is a conditional) {
          let $b$ be the guard;
          let $\overline{s}_t$ and $\overline{s}_f$ be its branches;  
          call labelSequence($\overline{s}_t$, $(\AGREEM_0,\AGREEM_i)$, $\PRED_{i-1}\wedge b$);
          call labelSequence($\overline{s}_f$, $(\AGREEM_0,\AGREEM_i)$, $\PRED_{i-1}\wedge \lnot b$);
          $\AGREEM_{i-1}$ is such that $\TRIPLEB{\AGREEM_{i-1}}{\PRED_{i-1}}{s_i}{\AGREEM_i}$;
        } else if ($s_i$ is a loop) {
          $\AGREEM_{i-1}$ is an agreement such that
          - $\AGREEM_{i-1} \sqsubseteq \AGREEM_i \sqcap \AGREEM_b$ // (see rule $\GRULENAME{while}$)
          - $\TRIPLEB{\AGREEM_{i-1}}{\PRED_{i-1}}{s_i}{\AGREEM_{i-1}}$
          let $b$ be the guard;
          let $\overline{s}_l$ be the loop body;
          call labelSequence($\overline{s}_l$, $(\AGREEM_{i-1},\AGREEM_{i-1})$, $\PRED_{i-1}\wedge b$);
        } else { // non-compound statement
          $\AGREEM_{i-1}$ is such that $\TRIPLEB{\AGREEM_{i-1}}{\PRED_{i-1}}{s_i}{\AGREEM_i}$;
        }
      }
    }
  \end{pseudocode}
  \caption{Labeling program code with agreements by using the
    \GSYSTEM}
  \label{fig:agreementPropagation}
\end{figure}

Now, suppose that the judgment $\PRESERVESB{\PRED}{s}{\AGREEM_s}$ can
be proved, where $\AGREEM_s$ and $\PRED_s$ are computed by
\CODE{labelSequence}, and $\PRED$ refers to the program point before
$s$.  In this case, let $\prog'$ be the program $\prog$ where $s$ has
been replaced by $\CODE{skip}$, and $\state$ be an initial state.
Then, the following holds: $\SEMANTICS{\prog}(\state)$ agrees with
$\SEMANTICS{\prog'}(\state)$ with respect to $\AGREEM_{out}$, provided
that executions terminate (non-termination is not considered).

\begin{proposition}
  \label{prop:soundnessPreservation}
  Given a statement $s$ in $\prog$ such that
  $\PRESERVESB{\PRED}{s}{\AGREEM_s}$, the output states obtained by
  executing both $\prog$ and $\prog'$ on some $\state$ agree on the
  agreement $\AGREEM_{out}$ corresponding to the desired slicing
  criterion if the execution terminates.
\end{proposition}

\begin{proof}
  Let $\trace$ and $\trace'$ be two trajectories coming from
  executing, respectively, $\prog$ and $\prog'$ from the initial state
  $\state$.  Let $\state_{in}[j]$ and $\state'_{in}[j]$ be the states
  of $\trace$ and $\trace'$, respectively, when control reaches the
  program point before $s$ (or \CODE{skip}, in the case of $\prog'$)
  for the $j$-th time, and $\state_{out}[j]$ and $\state'_{out}[j]$ be
  their corresponding states after $s$ (or \CODE{skip}).  If $s$ is
  not contained in any loop, then $j$ can only be $1$, and the proof
  is trivial.  Otherwise, it can be any number up to some $k_s$ (a
  non-negative number).

  If $k_s = 0$, then the loop is never executed on the input $\state$,
  and the proof follows trivially.

  Otherwise, $\state_{in}[1]$ and $\state'_{in}[1]$ are identical
  because both executions went exactly through the same statements; as
  a consequence, they certainly agree on $\AGREEM_s$.  On the other
  hand, $\state_{out}[1]$ and $\state'_{out}[1]$ are in general not
  identical, but they still agree on $\AGREEM_s$: in fact, the
  following holds:
  \[ \begin{array}{cl@{\qquad}l}
    (1) & \AGREEM_s(\state_{in}[1],\state'_{in}[1]) & \mbox{(they are
      identical)} \\
    (2) & \AGREEM_s(\state'_{in}[1],\state'_{out}[1]) &
    \mbox{(semantics of \CODE{skip})}
    \\ (3) & \AGREEM_s(\state_{in}[1],\state_{out}[1]) & \mbox{(property
      preservation)} \\ (4) &
    \AGREEM_s(\state_{out}[1],\state'_{out}[1]) &
    \mbox{(transitivity applied to (1), (2) and (3))}
  \end{array} \]
  Due to how the program is labeled with agreements,
  $\AGREEM_s(\state_{out}[1],\state'_{out}[1])$ implies an agreement
  of both executions at the end of the loop body with respect to the
  agreement $\AGREEM_{end}$ labeling that program point.  This also
  means that both executions will still agree on $\AGREEM_{end}$ at
  the beginning of the next iteration, and (again, by construction)
  they will also agree on $\AGREEM_s$ when $s$ is reached for the
  second time.  This mechanism can be repeated until $k_s$ is reached,
  and it is easy to realize that the agreement on $\AGREEM_s$ at the
  last iteration implies the final agreement on $\AGREEM_{out}$ (see
  Figure \ref{fig:soundnessPreservation}).

  The crux of this reasoning is that, by construction, the agreements
  labeling each program point imply that, when a statement can be
  removed from the loop body, this means that the original program and
  the slice will execute the loop body the same number of times.  In
  general, this does not mean that every loop will be executed the
  same number of times (some property-preserving loops could be even
  sliced out completely).
\end{proof}

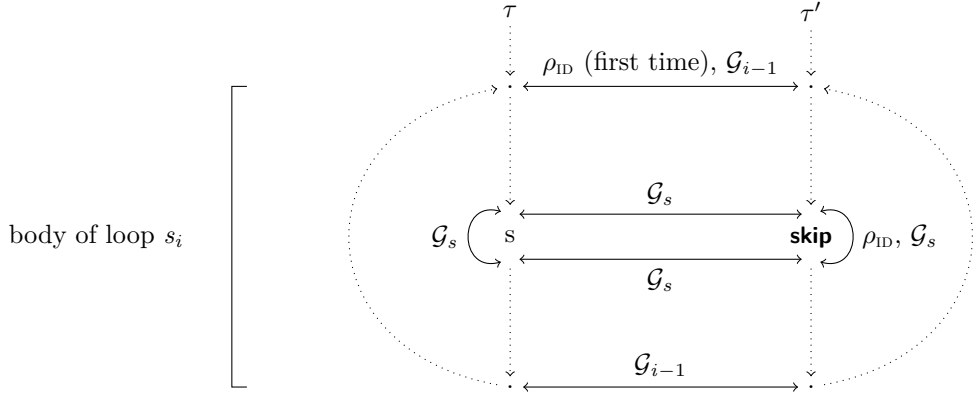
\begin{figure}
  \begin{center}
    \begin{tikzpicture}
      \node (pi) at (0,0) {$\trace$};
      \node (pip) at (4,0) {$\trace'$};
      \node (lb) at (-5.5,-3) {body of loop $s_i$};
      \draw (-3.5,-1) -- (-3.7,-1) -- (-3.7,-5) -- (-3.5,-5);
      \node (begin) at (0,-1) {$.$};
      \node (beginp) at (4,-1) {$.$};

      \draw[dotted,->] (pi) -- (begin);
      \draw[dotted,->] (pip) -- (beginp);

      \draw[<->] (begin) -- node[above] {$\IDDOM$ (first time),
        $\AGREEM_{i-1}$} (beginp);

      \node (end) at (0,-5) {$.$};
      \node (endp) at (4,-5) {$.$};

      \node (s0) at (0,-2.7) {};
      \node (s) at (0,-3) {s};
      \node (s1) at (0,-3.3) {};
      \draw[dotted,->] (begin) -- (s0);
      \draw[dotted,->] (s1) -- (end);
       
      \node (skip0) at (4,-2.7) {};
      \node (skip) at (4,-3) {\CODE{skip}};
      \node (skip1) at (4,-3.3) {};
      \draw[dotted,->] (beginp) -- (skip0);
      \draw[dotted,->] (skip1) -- (endp);
      
      \draw[<->] (s0) -- node[above] {$\AGREEM_s$} (skip0);
      \draw[<->] (s1) -- node[below] {$\AGREEM_s$} (skip1);
      \draw[<->] (end) -- node[above] {$\AGREEM_{i-1}$} (endp);

      \draw[<->] (s0) .. controls (-0.7,-2.5) and (-0.7,-3.5) .. node[left]
           {$\AGREEM_s$} (s1);
      \draw[<->] (skip0) .. controls (4.7,-2.5) and (4.7,-3.5) .. node[right]
           {$\IDDOM$, $\AGREEM_s$} (skip1);
      
      \draw[dotted,->] (end) .. controls (-2.8,-4.5) and (-2.8,-1.5)
      .. (begin);
      \draw[dotted,->] (endp) .. controls (6.8,-4.5) and (6.8,-1.5) .. (beginp);
    \end{tikzpicture}
  \end{center}
  \caption{Graphical representation of some aspects of Proposition
    \ref{prop:soundnessPreservation}}
  \label{fig:soundnessPreservation}
\end{figure}

This proposition can be used in order to remove all statements for
which such a property-preservation judgment can be proved.  In
general, the slice is computed by replacing all statements $s$ such
that $\PRESERVESB{\PRED'}{s}{\AGREEM_s}$ holds by $\CODE{skip}$, or,
equivalently removing all of them from the original code.  It is easy
to observe that this corresponds exactly to the notion of backward
abstract slicing given in Section
\ref{section:AbstractProgramSlicing}.

\begin{example}
  \label{ex:conditionsToSlice}
  Consider the following code, and let the final nullity of \xx be the
  property of interest (corresponding to the agreement $\{
  \AGRS{\NULLDOM}{x} \}$ after line 15): {\em
        \begin{lstlisting}[firstnumber=8]
   ...
   n := n*2;       // $\{ \AGRS{\ZERODOM}{n} \}$
   C x := new C(); // $\{ \AGRS{\ZERODOM}{n} \}$
   if (n=0) {      // $\{\}$ (final result on this branch always null)
     x := null;    // $\{ \AGRS{\NULLDOM}{x} \}$
   } else {        // $\{\}$ (final result on this branch never null)
     x := new C(); // $\{ \AGRS{\NULLDOM}{x} \}$
   }               // $\{ \AGRS{\NULLDOM}{x} \}$
      \end{lstlisting}
      }
  \noindent
  Each agreement on the right-hand side is the label after the
  statement at the same line.  Both lines $9$ and $10$ can be removed
  from the slice because:
  \begin{itemize}
  \item the final nullity of \xx only depends on the equality of \nn
    to $0$ before line 11, which is captured by the $\ZERODOM$ domain;
  \item line 10 does not affect \nn; and
  \item multiplying a number by 2 preserves the property of being
    equal to $0$.
  \end{itemize}
  Note that, although agreements after lines 11 and 13 are both empty
  (no matters how the state is $x$ will be null after line 12 and
  non-null after line 14), the one after line 10 is not because of
  rule $\GRULENAME{if2}$.
\end{example}

\begin{example}
  Consider again the code of Example \ref{ex:abstractSlicing}.
  Starting from the final agreement $\{ \AGRS{\CYCLEDOM}{\mbox{\CODE{list}}}
  \}$, the code is annotated as follows:

  {\em
    \begin{lstlisting}[firstnumber=33]
                        // $\{\AGRS{\CYCLEDOM}{list}\}$
   y := null;           // $\{\AGRS{\CYCLEDOM}{list}\}$
   x := list;           // $\{\AGRS{\CYCLEDOM}{list},\AGRS{\CYCLEDOM}{x}\}$
   while (pos>0) {      // $\{\AGRS{\CYCLEDOM}{list},\AGRS{\CYCLEDOM}{x}\}$
     y := x;            // $\{\AGRS{\CYCLEDOM}{list},\AGRS{\CYCLEDOM}{x}\}$
     x := x.next;       // $\{\AGRS{\CYCLEDOM}{list},\AGRS{\CYCLEDOM}{x}\}$
     pos := pos-1;
   }                    // $\{\AGRS{\CYCLEDOM}{list},\AGRS{\CYCLEDOM}{x}\}$
   z := new Node(elem); // $\{\AGRS{\CYCLEDOM}{list},\AGRS{\CYCLEDOM}{x}\}$
   z.next := x;         // $\{\AGRS{\CYCLEDOM}{list},\AGRS{\CYCLEDOM}{z},\AGRS{\CYCLEDOM}{x}\}$
   if (y = null) {      // $\{\AGRS{\CYCLEDOM}{z},\AGRS{\CYCLEDOM}{x}\}$
     list := z;         // $\{\AGRS{\CYCLEDOM}{list}\}$
   } else {             // $\{\AGRS{\CYCLEDOM}{list},\AGRS{\CYCLEDOM}{z},\AGRS{\CYCLEDOM}{x}\}$
     y.next = z;        // $\{\AGRS{\CYCLEDOM}{list}\}$
   }                    // $\{\AGRS{\CYCLEDOM}{list}\}$
    \end{lstlisting}
  }

  \noindent To prove the necessary tuples, it is important to note
  that \CODE{next} is the only reference field selector in the class
  \CODE{Node}, so that any cycle has to traverse it.  Moreover, to
  prove that the cyclicity of \CODE{list} after line 46 is equivalent
  to the cyclicity of \CODE{list} and \CODE{z} needs non trivial
  reasoning about data structures; concretely, it is necessary to have
  some \emph{reachability analysis} \cite{tcs13}\footnote{Note that
    pair-sharing-based cyclicity analysis
    \cite{DBLP:conf/vmcai/RossignoliS06} is not enough since \CODE{y}
    and \CODE{z} are sharing before line 46.} capable to detect that
  there is no path from \CODE{y} to \CODE{y} (otherwise, a cycle could
  be created by \CODE{y.next = z}).  This information should be
  available as $\PRED$, and allows to say that \CODE{y} and
  \CODE{list} (note that \CODE{list} is affected because it is sharing
  with \CODE{y}) are cyclic after line 46 if and only if \CODE{z} or
  \CODE{x} or \CODE{list} were before that line.  The same happens at
  line 42.
  
  Moreover, the agreement does not change in the loop body at lines
  36--40 because (1) data structures are not modified; (2) \CODE{list}
  is not affected in any way; and (3) the value of \CODE{x} changes,
  but its cyclicity does not (by executing \CODE{x:=x.next}, there is
  no way to make unreachable a cycle which was reachable before, or
  the other way around).  Importantly, this also means that the
  cyclicity of \CODE{list} and \CODE{x} is \emph{preserved} by the
  loop, so that it can be safely removed from the slice (the preserved
  property is the same as the agreement after line 40).
  
  On the other hand, the conditional statement at lines 43--47 cannot
  be removed directly because it is not possible to prove that the
  cyclicity of \CODE{list} is preserved through it (actually, it is
  not).  In order to prove that the whole code between lines 34 and 47
  preserves the cyclicity of \CODE{list}, a kind of \emph{case-based}
  reasoning could be used: (1) the initially acyclicity of \CODE{list}
  implies its final acyclicity; and (2) the initially cyclicity of
  \CODE{list} implies its final cyclicity.  Both these results can be
  proved by standard static-analysis techniques \cite{tcs13}.
\end{example}

Needless to say, slices could be sub-optimal (for example, the
requirement about executing loops the same number of times needs not
be satisfied by any correct slice).  It is not difficult to see that,
if a ``concrete'' slicing criterion would be considered instead of an
abstract one, then agreements would only be allowed to contain
conditions $\AGRS{\IDDOM}{x}$ for a certain set of variables.  The
\CODE{labelSequence} would work exactly the same way, with an
important difference: when trying to compute the precondition of a
tuple $\TRIPLEB{\_}{\PRED}{s}{\AGREEM}$, the possible outcome
could, again, contain only conditions $\AGRS{\IDDOM}{x}$.

\begin{example}
  \label{ex:conditionsToSlice1}
  In Example \ref{ex:conditionsToSlice}, suppose the final agreement
  be $\{ \AGRS{\IDDOM}{x} \}$, corresponding to a concrete slicing
  criterion interested in (the exact value of) $x$.  In this case, the
  agreement after line 10 could only be $\{ \AGRS{\IDDOM}{n} \}$,
  since (1) $\{\}$ would not be correct, and (2) no other abstract
  domain can appear in agreements.  This way, line 9 could not be
  sliced out since it does not preserve the $\IDDOM$ property of $n$.
\end{example}

Semantically, abstract slices are in general smaller than concrete
slices.  This is also the case of
Example \label{ex:conditionsToSlice1}.  Clearly, this does \emph{not}
imply that \emph{every} abstract slicing algorithm would remove more
statements than \emph{every} concrete slicing algorithm.

\subsection{Practical issues and optimizations}
\label{sec:practicalIssuesAndOptimizations}

This section discusses how the analysis can realistically deal with
the computation of abstract dependencies and the propagation of
agreements, and an optimization based on recent work on sharing.

\subsubsection{Agreements and ucos}
\label{sec:agreementsAndUcos}

One of the major challenges of the whole approach is how agreements
are propagated backwards through the code as precisely as possible,
i.e., being able to detect that agreement on some ucos before $s$
implies agreement on (possibly) other ucos after $s$.

Ideally, given $s$, $\PRED$ and $\AGREEM'$, the \GSYSTEM should find
the \emph{best} $\AGREEM$ such that
$\TRIPLEB{\AGREEM}{\PRED}{s}{\AGREEM'}$.  However, it is clearly
unrealistic to imagine that the static analyzer will always be able to
find the best ucos without going into severe scalability (even
decidability) issues: in general, there exist infinite possible
choices for an input agreement satisfying the augmented triple.  In
practice, an implementation of this slicing algorithm will be equipped
with a library of ucos among which the satisfaction of augmented
triples can be checked.  The wider the library, the more precise the
results.  Rule of the \PSYSTEM and the \GSYSTEM can be specialized
with respect to the ucos at hand.

\begin{example}
  \label{ex:conditionsToSlice2}
  Consider this code already presented in Example
  \ref{ex:conditionsToSlice}, where the slicing criterion is the final
  nullity of \xx:
  {\em \begin{lstlisting}[firstnumber=8]
   ...
   n := n*2;
   C x := new C();
   if (n=0) {
     x := null;
   } else {
     x := new C();
   }
  \end{lstlisting}}
  \noindent
  In order to be able to remove lines 9 and 10 from the slice, an
  analyzer has to ``know'' that, after merging the result of both
  branches, the uco $\ZERODOM$ precisely describes the agreement
  before line 11.  In other words, the analyzer must know both
  $\ZERODOM$ and $\NULLDOM$ in order to be able to manipulate
  information about them.  Moreover, given a library of ucos, rules
  can be optimized for some recurrent programming patterns like guards
  (\CODE{n=0}) or (\CODE{x=null}), or statements \CODE{m:=0}.
\end{example}

It is clear that to design of an analyzer which is able to deal with
all possible ucos is infeasible.  However, the combination of some
simple numeric or reference domains like the one described in this
paper would already lead to meaningful results.

\subsubsection{Use of Field-Sensitive Sharing}
\label{sec:useOfFieldSensitiveSharing}

As already mentioned, \emph{field-sensitive sharing analysis}
\cite{ZanardiniG15sh} is able to keep track of fields which are
involved in \emph{converging paths} from two variables to a common
location in the heap (the \emph{shared} location).  A
\emph{propositional formula} is attached to each pair of variables and
each program point, and specifies the fields involved in \emph{every}
pair of converging paths reaching a common location.  For example, if
the formula $\lnot \fpropositionl{\mbox{\CODE{f}}} \wedge
\fpropositionr{\mbox{\CODE{g}}}$ is attached to a pair of variables $(x,y)$
at a certain program point $n$ (written $S_n(x,y) = \lnot
\fpropositionl{\mbox{\CODE{f}}} \wedge \fpropositionr{\mbox{\CODE{g}}}$), this means
that the analysis was able to detect that, for \emph{every} two paths
$\pi_1$ and $\pi_2$ in the heap starting from $x$ and $y$,
respectively, and both ending in the same (shared) location,
\begin{itemize}
\item $\pi_1$ certainly does not traverse field \CODE{f}, as dictated
  by $\lnot \fpropositionl{\mbox{\CODE{f}}}$ (arrows from top-left to
  bottom-right refer to paths from $x$, i.e., the first variable in
  the pair under study); and
\item $\pi_2$ certainly traverses \CODE{g}, as prescribed by
  $\fpropositionr{\mbox{\CODE{g}}}$ (arrows from top-right to bottom-left
  refer to paths from $y$).
\end{itemize}
In presence of such an analysis, two kinds of improvements can be
potentially obtained when analyzing a field update:
\begin{itemize}
\item the number of variables which can be actually affected by an
  update is, in general, reduced since it is possible to guarantee
  that some (traditionally) sharing variables will not be affected;
\item even for variables which are (still) possibly sharing with $x$,
  the set of field sequences $\bar{g}$ to be considered can be
  substantially smaller.
\end{itemize}

\begin{example}
  \label{ex:fieldSensitiveSharing}
  Suppose that field-sensitive sharing analysis is able to guarantee
  the following at a program point before the field update
  \CODE{x.f:=e}:
  \begin{itemize}
  \item The formula $\lnot \fpropositionl{\mbox{\CODE{f}}} \wedge
    \fpropositionr{\mbox{\CODE{g}}}$ correctly describes the sharing between
    \xx and \yy; and
  \item The formula $\fpropositionr{\mbox{\CODE{h}}}$ correctly describes the
    sharing between \xx and \zz.
  \end{itemize}
  According to the traditional notion of sharing, both \yy and \zz may
  share with \xx.  However,
  \begin{itemize}
  \item the assignment \CODE{x.f:=e} provably does \emph{not} affect
    \yy because no path from \xx traversing \CODE{f} will reach a
    location that is also reachable from \yy; and
  \item when considering all the possible field sequences starting
    from \zz, only those containing \CODE{h} have to be considered.
  \end{itemize}
\end{example}

The rule \GRULENAME{fassign} can be refined by using field-sensitive
sharing, as follows.  Let $n$ be the program point before the field
update.
\begin{itemize}
\item in pre-condition $({*}{*})$, the only sharing variables that
  have to be dealt with are those for which it cannot be proved that
  they are unaffected by the update; this can be done by defining a
  new set $\FSHARE{f}{x}$ of variables which are possibly sharing with
  $x$ in such a way that some path from $x$ to a shared location could
  traverse $f$:
  \[ \FSHARE{f}{x} = \left\{ y~|~S_n(x,y) \not\models \lnot
  \fpropositionl{f} \right\} \] This means that $y$ is considered as
  potentially affected by the update when the propositional formula
  describing how it shares with $x$ does not entail that paths from
  $x$ to shared locations do not traverse $f$.
\item in the same pre-condition $({*}{*})$, the universal
  quantification on field sequences can be restricted to those
  compatible with field-sensitive information.  More formally, a field
  sequence $\FSEQ{.f_1.f_2.....f_n}$ has to be considered only if it
  is possible that a path from $y$ traversing exactly those fields
  ends in a shared location, or, equivalently, if the set $\{
  \fpropositionl{f}, \fpropositionr{f_1}, \fpropositionr{f_2}, .. ,
  \fpropositionr{f_n} \}$ is a \emph{model} of $S_n(x,y)$.  That such
  a set is a model of $S_n(x,y)$ is equivalent to say that the
  field-sensitive information is compatible with the existence of a
  pair of paths $\pi_1$ and $\pi_2$ such that (1) $\pi_1$ starts from
  $x$; (2) $\pi_2$ starts from $y$; (3) $\pi_1$ only traverses $f$;
  (4) $\pi_2$ traverses all and only the fields $f_1, f_2, ... f_n$;
  and (5) both paths end in the same shared location.  Let
  $\FSET{\bar{g}}$ be the set of fields contained in the field
  sequence $\bar{g}$.  Then, the above condition can be written as 
  \[ \SEQSHARE{f}{\bar{g}}{x}{y} \qquad \equiv \qquad
  \bigwedge_{p \in X} p \wedge \bigwedge_{q \notin X} \lnot q \quad
  \models \quad S_n(x,y) \] where $X = \{ \fpropositionl{f} \} \cup \{
  \fpropositionr{g} ~|~ g \in \FSET{\bar{g}} \}$, and $q \notin X$ means
  that $q$ is any proposition $\fpropositionl{h}$ or $\fpropositionr{h}$
  (for some field $h$) not included in $X$.
\end{itemize}

The refined \GRULENAME{fassign} rule, called \GRULENAME{fassign2},
comes to be

\[
\GRULE{
  \begin{array}{c@{~}r@{~}l}
    ({*}) & \forall y \in \DALIAS{x}. & \forall
    \state_1 \models \PRED,\state_2 \models
    \PRED. \\ 
    & & ~~\AGREEM(\state_1,\state_2)
    \Rightarrow \\
    & & ~~\AGREEM'(\state_1[y.f \leftarrow
      \SEMANTICS{e}(\state_1)],\state_2[y.f \leftarrow
      \SEMANTICS{e}(\state_2)]) \\
    ({*}{*}) & \forall y \in \FSHARE{f}{x}. &
    \forall \bar{g}.~\SEQSHARE{f}{\bar{g}}{x}{y} \Rightarrow
    (\forall
    \state_1 \models \PRED,\state_2 \models
    \PRED. \\
    & & ~~\AGREEM(\state_1,\state_2)
    \Rightarrow \\
    & & ~~\AGREEM'(\state_1[y.\bar{g} \leftarrow
      \SEMANTICS{e}(\state_1)],\state_2[y.\bar{g} \leftarrow
      \SEMANTICS{e}(\state_2)])) \\
    ({*}{*}{*}) & \forall y \notin \DALIAS{x}. & \AGREEM(y)
    \sqsubseteq \AGREEM'(y)
  \end{array}
}{
  \TRIPLEB{\AGREEM}{\PRED}{x.f\mbox{\CODE{:=}}e}{\AGREEM'}
}{fassign2}
\]

\subsection{Comparison with related algorithms}
\label{sec:comparisonWithRelatedAlgorithms}

The Tukra Abstract Program Slicing Tool
\cite{DBLP:conf/icsoft/HalderC12} implements the computation of a
Dependence Condition Graph \cite{DBLP:journals/scp/HalderC13} for
performing abstract slicing.  To use a Program Dependence Graph is
somehow alternative to the computation of agreements.  As far as the
author make it possible to understand, Tukra only deals with numerical
values, and it is not clear which properties are supported (i.e.,
which is the ``library'' of ucos mentioned in Section
\ref{sec:agreementsAndUcos}).

Moreover, the authors of that tool point out that their approach is
able to exclude some dependencies that were not ruled out in previous
work introducing abstract dependencies \cite{MastroeniZanardini}.
However, they do not consider that a rule system for computing
agreements (essentially, the \GSYSTEM described in the present paper)
was introduced \cite{Zanardini} before Tukra was developed, and does
not suffer from the limitations they describe.

\section{Related Work}
\label{sec:relatedWork}

The formal framework referred to in this paper \cite{AForm} is not the
only attempt to provide a unified mathematical framework from program
slicing.  In \cite{WardZedan}, the authors have precisely this aim.
In this work, the authors unify different approaches to program
slicing by defining a particular semantic relation, based on the
weakest precondition semantics, called \emph{semirefinement} such
that, given a program $\prog$, the possible slices are all the
programs that are semirefinements of $\prog$.  In this framework,
different forms of slicing are modeled as program transformations.
Hence, a program $\progq$ is a slice of $\prog$ if the transformation
of $\progq$ (corresponding to the particular form of slicing to
compute) is a semirefinement of the same transformation of $\prog$.
This approach is extremely interesting, but does not really allow to
compare the different forms of slicing, feature that we consider
fundamental for introducing the new abstract forms of slicing as
generalizations of the existing ones.  It may surely deserve further
research to study whether also abstract slicing could be modeled in
this framework.

As far as the relation between slicing and dependencies is concerned,
there are at least two works that are related with our ideas in
different ways.  One of the first works aiming at formalizing a
semantic approach to dependency, leading to a semantic computation of
slicing, is the \emph{information-flow logic} by Amtoft and Banerjee
\cite{AB07}.  This logic allows us to formally derive, by structural
induction, the set of all the \emph{independencies} among variables.
In Figure \ref{AB04-fig}, the original notation proposed by the
authors is used, where $[x\ltimes y]$ is to be read as ``the current
value of $x$ is independent of the initial value of $y$'', and holds
if, for each pair of \emph{initial} states which agree on all the
variables but $y$, the corresponding \emph{current} states agree on
$x$.  Hence, $T^{\#}$ stands for sets of independencies, and $G$ is a
set of variables representing the \emph{context}, i.e., (a superset
of) the variables on which at least one test surrounding the
statements depends on.

\begin{figure}[h]
  \begin{center}
    \framebox{ $\begin{array}{l}
        G\vdash\{T_{0}^{\#}\}\ x:=e\ \{T^{\#}\} \\
        \qquad\mbox{if}\ \forall [y\ltimes w]\in T^{\#}.\ (x\neq y\ \Ra\ [y\ltimes w]\in T_{0}^{\#})\\
        \qquad(x=y\ \Ra\ (w\notin G \wedge\ \forall {z \in \VARS{e}}.\:[z\ltimes w]\in T_{0}^{\#}))\\
        \\
        \irule{G_{0}\vdash\{T_{0}^{\#}\}s_{1}\{T^{\#}\}\ \
          G_{0}\vdash\{T_{0}^{\#}\}s_{2}\{T^{\#}\}}
        {G\vdash\{T_{0}^{\#}\}\ifc\ e\ \thenc\ s_{1}\ \elsec\ s_{2}\{T^{\#}\}}\\
        \qquad\mbox{if}\ G\subseteq G_{0}\ \wedge\ (w\notin G_{0}\ \Ra\
        \forall {x \in \VARS{e}}.\:[x\ltimes w]\in
        T_{0}^{\#})\\
        \\
        \irule{G_{0}\vdash\{T^{\#}\}s\{T^{\#}\}}
        {G\vdash\{T^{\#}\}\while\ e\ \dow\ s\{T^{\#}\}}\\
        \qquad\mbox{if}\ G\subseteq G_{0}\ \wedge\ (w\notin G_{0}\ \Ra\
        \forall {x \in \VARS{e}}.\:[x\ltimes w]\in T^{\#})
      \end{array}$
    }
  \end{center}
  \caption{A fragment of the independency logic}
  \label{AB04-fig}
\end{figure}

\noindent
In our aim of defining slicing in terms of dependencies, the first
thing we have to observe in this logic is that it always computes
(in)dependencies from the \emph{initial} values of variables. This
makes its use for slicing not so straightforward, since it loses the
\emph{local} dependency between statements.  Consider for example the
program fragment $\prog~=~\mbox{\CODE{w:=x+1; y:=w+2; z:=y+3}}$.  At
the end of this program, we know that \zz only depends on the initial
value of \xx, but, by using the logic in Fig.~\ref{AB04-fig}, we lose
the trace of (in)dependencies which, in this case, would involve all
the three assignments.  As a matter of fact, this logic is more
suitable for forward slicing, which is the one considered by the
authors \cite{AB07}, since it fixes the criterion \emph{on the input}.
In the trivial example given above, if we consider as criterion the
input of \xx, then we obtain that all the statements depend on \xx.
Therefore, any slice of the original program contains all statements
\cite{AB07}.
In the logic, more explicitly, this notion of dependency is used for
characterizing the set of independencies holding during the execution
of a program.

Another, more recent, approach to slicing by means of dependencies is
\cite{Danicic11}.  In this work, the authors propose new definitions
of control dependencies: non-termination sensitive and insensitive.
These new semantic notions of dependencies are then used for computing
more precise standard slices.  It could be surely interesting to study
the semantic relation between their notion of dependencies and the
ones we propose in this paper.

Finally, a related algorithm for computing abstract slices has been
already discussed in Section
\ref{sec:comparisonWithRelatedAlgorithms}.  It is necessary to point
out that the agreement-based approach to abstract slicing
\cite{Zanardini} was introduced before the Tukra tool.

\section{Conclusion and Future Work}
 % conclusionAndFutureWork
The present paper formally defines the notion of abstract program
slicing, a general form of slicing where properties of data are
observed instead of their exact value.  A formal framework is
introduced where the different forms of abstract slicing can be
compared; moreover, traditional, non-abstract forms of slicing are
also included in the framework, allowing to prove that non-abstract
slicing is a special case of abstract slicing where no abstraction on
data is performed.

Algorithms for computing abstract dependencies and program slices are
given.  Future work includes an implementation of this analysis for an
Object-Oriented programming language where properties may refer either
to numerical or reference values (to the best of our knowledge,
existing tools only deal with integer variables).  On the other hand,
we observed that the provided notion of abstract dependency is not
suitable for slicing computation by using PDGs.  We believe that it is
possible to further generalize the notion of abstract dependencies
allowing to characterize a recursive algorithm able to track backwards
both the variables that affect the criterion, and the abstract
properties of these variables affecting the abstract criterion.

Another interesting line of research is to understand how other
approaches to slicing can be extended in order to include abstract
slicing.  As noted before, it would be interesting to study whether it
is possible to model abstract slicing as a program transformation,
allowing us to define also abstract slicing in term of
semirefinement \cite{WardZedan}.  Another, more algorithmic,
interesting approach is the one proposed in \cite{Barros10}, where
weakest precondition and strongest postcondition semantics are
combined in a new more precise algorithm for standard slicing.  It
could be very interesting to understand whether this approach could be
extended in order to cope also with the computation of abstract forms
of slicing.

\bibliographystyle{acmsmall} 

\begin{thebibliography}{}

\bibitem[\protect\citeauthoryear{Amtoft and Banerjee}{Amtoft and
  Banerjee}{2007}]{AB07}
{\sc Amtoft, T.} {\sc and} {\sc Banerjee, A.} 2007.
\newblock A logic for information flow analysis with an application to forward
  slicing of simple imperative programs.
\newblock {\em Science of Computer Programming\/}~{\em 64,\/}~1, 3--28.

\bibitem[\protect\citeauthoryear{Barros, da~Cruz, Henriques, and Pinto}{Barros
  et~al\mbox{.}}{2010}]{Barros10}
{\sc Barros, J.~B.}, {\sc da~Cruz, D.}, {\sc Henriques, P.~R.}, {\sc and} {\sc
  Pinto, J.~S.} 2010.
\newblock Assertion-based slicing and slice graphs.
\newblock In {\em Proceedings of the 2010 8th IEEE International Conference on
  Software Engineering and Formal Methods}. SEFM '10. IEEE Computer Society,
  Washington, DC, USA, 93--102.

\bibitem[\protect\citeauthoryear{Bijlsma and Nederpelt}{Bijlsma and
  Nederpelt}{1998}]{BN98}
{\sc Bijlsma, A.} {\sc and} {\sc Nederpelt, R.} 1998.
\newblock Dijkstra-{S}cholten predicate calculus : concepts and misconceptions.
\newblock {\em Acta Informatica\/}~{\em 35,\/}~12, 1007--1036.

\bibitem[\protect\citeauthoryear{Binkley, Danicic, Gyim\'othy, Harman, Kiss,
  and Korel}{Binkley et~al\mbox{.}}{2006a}]{AForm}
{\sc Binkley, D.}, {\sc Danicic, S.}, {\sc Gyim\'othy, T.}, {\sc Harman, M.},
  {\sc Kiss, A.}, {\sc and} {\sc Korel, B.} 2006a.
\newblock A formalisation of the relationship between forms of program slicing.
\newblock {\em Science of Computer Programming\/}~{\em 62,\/}~3, 228--252.

\bibitem[\protect\citeauthoryear{Binkley, Danicic, Gyim\'othy, Harman, Kiss,
  and Korel}{Binkley et~al\mbox{.}}{2006b}]{TheoFoun}
{\sc Binkley, D.}, {\sc Danicic, S.}, {\sc Gyim\'othy, T.}, {\sc Harman, M.},
  {\sc Kiss, A.}, {\sc and} {\sc Korel, B.} 2006b.
\newblock Theoretical foundations of dynamic program slicing.
\newblock {\em Theoretical Computer Science\/}~{\em 360,\/}~1, 23--41.

\bibitem[\protect\citeauthoryear{Binkley and Gallagher}{Binkley and
  Gallagher}{1996}]{BinGalla96}
{\sc Binkley, D.~W.} {\sc and} {\sc Gallagher, K.~B.} 1996.
\newblock Program slicing.
\newblock {\em Advances in Computers\/}~{\em 43}.

\bibitem[\protect\citeauthoryear{Canfora, Cinitile, and {De Lucia}}{Canfora
  et~al\mbox{.}}{1998}]{Conditioned}
{\sc Canfora, G.}, {\sc Cinitile, A.}, {\sc and} {\sc {De Lucia}, A.} 1998.
\newblock Conditioned program slicing.
\newblock {\em Information and Software Technology\/}~{\em 40}, 11--12.

\bibitem[\protect\citeauthoryear{Cimitile, De~Lucia, and Munro}{Cimitile
  et~al\mbox{.}}{1996}]{CDM96}
{\sc Cimitile, A.}, {\sc De~Lucia, A.}, {\sc and} {\sc Munro, M.} 1996.
\newblock A specification driven slicing process for identifying reusable
  functions.
\newblock {\em Journal of Software Maintenance\/}~{\em 8,\/}~3, 145--178.

\bibitem[\protect\citeauthoryear{Cousot}{Cousot}{2001}]{C01-Dag}
{\sc Cousot, P.} 2001.
\newblock Abstract interpretation based formal methods and future challenges.
\newblock In {\em Informatics - 10 Years Back. 10 Years Ahead}. 138--156.

\bibitem[\protect\citeauthoryear{Cousot and Cousot}{Cousot and
  Cousot}{1977}]{CC77}
{\sc Cousot, P.} {\sc and} {\sc Cousot, R.} 1977.
\newblock Abstract interpretation: A unified lattice model for static analysis
  of programs by construction or approximation of fixpoints.
\newblock In {\em Proceedings of ACM Symposium on Principles of Programming
  Languages (POPL)}. ACM Press, New York, 238--252.

\bibitem[\protect\citeauthoryear{Cousot and Cousot}{Cousot and
  Cousot}{1979}]{CC79}
{\sc Cousot, P.} {\sc and} {\sc Cousot, R.} 1979.
\newblock Systematic design of program analysis frameworks.
\newblock In {\em Proceedings of ACM Symposium on Principles of Programming
  Languages (POPL)}. ACM Press, New York, 269--282.

\bibitem[\protect\citeauthoryear{Danicic, Barraclough, Harman, Howroyd, Kiss,
  and Laurence}{Danicic et~al\mbox{.}}{2011}]{Danicic11}
{\sc Danicic, S.}, {\sc Barraclough, R.~W.}, {\sc Harman, M.}, {\sc Howroyd,
  J.~D.}, {\sc Kiss, A.}, {\sc and} {\sc Laurence, M.~R.} 2011.
\newblock A unifying theory of control dependence and its application to
  arbitrary program structures.
\newblock {\em Theor. Comput. Sci.\/}~{\em 412,\/}~49, 6809--6842.

\bibitem[\protect\citeauthoryear{{De Lucia}}{{De Lucia}}{2001}]{DeLucia}
{\sc {De Lucia}, A.} 2001.
\newblock Program slicing: Methods and applications.
\newblock In {\em Proceedings of International Workshop on Source Code Analysis
  and Manipulation (SCAM)}.

\bibitem[\protect\citeauthoryear{Dijkstra}{Dijkstra}{1975}]{Dij75}
{\sc Dijkstra, E.} 1975.
\newblock Guarded commands, nondeterminacy and formal derivation of programs.
\newblock {\em Communications of the ACM\/}~{\em 18,\/}~8, 453--457.

\bibitem[\protect\citeauthoryear{Dijkstra and Scholten}{Dijkstra and
  Scholten}{1990}]{DS90}
{\sc Dijkstra, E.} {\sc and} {\sc Scholten, C.~S.} 1990.
\newblock {\em Predicate {C}alculus and {P}rogram {S}emantics}.
\newblock Springer-Verlag.

\bibitem[\protect\citeauthoryear{Field, Ramalingam, and Tip}{Field
  et~al\mbox{.}}{1995}]{FRT96}
{\sc Field, J.}, {\sc Ramalingam, G.}, {\sc and} {\sc Tip, F.} 1995.
\newblock Parametric program slicing.
\newblock In {\em Proceedings of ACM Symposium on Principles of Programming
  Languages (POPL)}. ACM Press, 379--392.

\bibitem[\protect\citeauthoryear{Gallagher and Lyle}{Gallagher and
  Lyle}{1991}]{GL91ieee}
{\sc Gallagher, K.~B.} {\sc and} {\sc Lyle, J.~R.} 1991.
\newblock Using program slicing in software maintenance.
\newblock {\em IEEE Transactions on Software Engineering\/}~{\em 17,\/}~8,
  751--761.

\bibitem[\protect\citeauthoryear{Genaim and Zanardini}{Genaim and
  Zanardini}{2013}]{tcs13}
{\sc Genaim, S.} {\sc and} {\sc Zanardini, D.} 2013.
\newblock Reachability-based {A}cyclicity {A}nalysis by {A}bstract
  {I}nterpretation.
\newblock {\em Theoretical Computer Science\/}~{\em 474,\/}~0, 60--79.

\bibitem[\protect\citeauthoryear{Giacobazzi, Jones, and Mastroeni}{Giacobazzi
  et~al\mbox{.}}{2012}]{GiacobazziJM12}
{\sc Giacobazzi, R.}, {\sc Jones, N.~D.}, {\sc and} {\sc Mastroeni, I.} 2012.
\newblock Obfuscation by partial evaluation of distorted interpreters.
\newblock In {\em Proceedings of the {ACM} {SIGPLAN} 2012 Workshop on Partial
  Evaluation and Program Manipulation, {PEPM} 2012, Philadelphia, Pennsylvania,
  USA, January 23-24, 2012}. 63--72.

\bibitem[\protect\citeauthoryear{Giacobazzi and Mastroeni}{Giacobazzi and
  Mastroeni}{2004a}]{GM04popl}
{\sc Giacobazzi, R.} {\sc and} {\sc Mastroeni, I.} 2004a.
\newblock Abstract non-interference: Parameterizing non-interference by
  abstract interpretation.
\newblock In {\em Proceedings of ACM Symposium on Principles of Programming
  Languages (POPL)}. ACM Press, 186--197.

\bibitem[\protect\citeauthoryear{Giacobazzi and Mastroeni}{Giacobazzi and
  Mastroeni}{2004b}]{GM04CSL}
{\sc Giacobazzi, R.} {\sc and} {\sc Mastroeni, I.} 2004b.
\newblock Proving abstract non-interference.
\newblock In {\em Annual Conf.\ of the European Association for Computer
  Science Logic (CSL~'04)}, {A.~T. J.~Marcinkowski}, Ed. Vol. 3210.
  Springer-Verlag, Berlin, 280--294.

\bibitem[\protect\citeauthoryear{Giacobazzi, Ranzato, and Scozzari}{Giacobazzi
  et~al\mbox{.}}{2000}]{GRSjacm}
{\sc Giacobazzi, R.}, {\sc Ranzato, F.}, {\sc and} {\sc Scozzari, F.} 2000.
\newblock Making abstract interpretations complete.
\newblock {\em J.\ of the ACM.\/}~{\em 47,\/}~2, 361--416.

\bibitem[\protect\citeauthoryear{Halder and Cortesi}{Halder and
  Cortesi}{2012}]{DBLP:conf/icsoft/HalderC12}
{\sc Halder, R.} {\sc and} {\sc Cortesi, A.} 2012.
\newblock Tukra: {A}n {A}bstract {P}rogram {S}licing {T}ool.
\newblock In {\em Proceedings of International Conference on Software Paradigm
  Trends (ICSOFT)}. 178--183.

\bibitem[\protect\citeauthoryear{Halder and Cortesi}{Halder and
  Cortesi}{2013}]{DBLP:journals/scp/HalderC13}
{\sc Halder, R.} {\sc and} {\sc Cortesi, A.} 2013.
\newblock Abstract program slicing on dependence condition graphs.
\newblock {\em Science of Computer Programming\/}~{\em 78,\/}~9, 1240--1263.

\bibitem[\protect\citeauthoryear{Hind}{Hind}{2001}]{Hind2001}
{\sc Hind, M.} 2001.
\newblock Pointer analysis: Haven't we solved this problem yet?
\newblock In {\em Proceedings of the Workshop on Program Analysis for Software
  Tools and Engineering (PASTE)}. ACM Press, New York, 54--61.

\bibitem[\protect\citeauthoryear{Hoare}{Hoare}{1969}]{Hoa69}
{\sc Hoare, C.} 1969.
\newblock An axiomatic basis for computer programming.
\newblock {\em Communications of the ACM\/}~{\em 12,\/}~10, 576--580.

\bibitem[\protect\citeauthoryear{Horwitz, Prins, and Reps}{Horwitz
  et~al\mbox{.}}{1989}]{horPR89}
{\sc Horwitz, S.}, {\sc Prins, J.}, {\sc and} {\sc Reps, T.} 1989.
\newblock Integrating non-interfering versions of programs.
\newblock {\em ACM Transaction on Programming Languages and Systems\/}~{\em
  11,\/}~3.

\bibitem[\protect\citeauthoryear{Hunt and Mastroeni}{Hunt and
  Mastroeni}{2005}]{HM05}
{\sc Hunt, S.} {\sc and} {\sc Mastroeni, I.} 2005.
\newblock The {P}{E}{R} model of abstract non-interference.
\newblock In {\em Proceedings of Static Analysis Symposium (SAS)}. Lecture
  Notes in Computer Science Series, vol. 3672. Springer-Verlag, 171--185.

\bibitem[\protect\citeauthoryear{Korel and Laski}{Korel and
  Laski}{1988}]{KorelLaski}
{\sc Korel, B.} {\sc and} {\sc Laski, J.} 1988.
\newblock Dynamic program slicing.
\newblock {\em Information Processing Letters\/}~{\em 29,\/}~3, 155--183.

\bibitem[\protect\citeauthoryear{Majumdar, Drape, and Thomborson}{Majumdar
  et~al\mbox{.}}{2007}]{MDT07}
{\sc Majumdar, A.}, {\sc Drape, S.~J.}, {\sc and} {\sc Thomborson, C.~D.} 2007.
\newblock Slicing obfuscations: design, correctness, and evaluation.
\newblock In {\em Proceedings of ACM Workshop on Digital Rights Management
  (DRM)}. ACM, New York, NY, USA, 70--81.

\bibitem[\protect\citeauthoryear{Mastroeni}{Mastroeni}{2013}]{Mastroeni13}
{\sc Mastroeni, I.} 2013.
\newblock Abstract interpretation-based approaches to security - {A} survey on
  abstract non-interference and its challenging applications.
\newblock In {\em Semantics, Abstract Interpretation, and Reasoning about
  Programs: Essays Dedicated to David A. Schmidt on the Occasion of his
  Sixtieth Birthday, Manhattan, Kansas, USA, 19-20th September 2013.} 41--65.

\bibitem[\protect\citeauthoryear{Mastroeni and Nikoli\'c}{Mastroeni and
  Nikoli\'c}{2010}]{MastroeniNicolic}
{\sc Mastroeni, I.} {\sc and} {\sc Nikoli\'c, D.} 2010.
\newblock Abstract {P}rogram {S}licing: {F}rom {T}heory towards an
  {I}mplementation.
\newblock In {\em Proceedings of International Conference on Formal Engineering
  Methods (ICFEM)}. Lecture Notes in Computer Science Series, vol. 6447.
  Springer-Verlag, 452--467.

\bibitem[\protect\citeauthoryear{Mastroeni and Zanardini}{Mastroeni and
  Zanardini}{2008}]{MastroeniZanardini}
{\sc Mastroeni, I.} {\sc and} {\sc Zanardini, D.} 2008.
\newblock Data dependencies and program slicing: From syntax to abstract
  semantics.
\newblock In {\em Proceedings of Symposium on Partial Evaluation and
  Semantics-Based Program Manipulation (PEPM)}. 125--134.

\bibitem[\protect\citeauthoryear{Ranzato and Tapparo}{Ranzato and
  Tapparo}{2002}]{RT02}
{\sc Ranzato, F.} {\sc and} {\sc Tapparo, F.} 2002.
\newblock Making abstract model checking strongly preserving.
\newblock In {\em Proceedings of Static Analysis Symposium (SAS)}. Lecture
  Notes in Computer Science Series, vol. 2477. Springer-Verlag, 411--427.

\bibitem[\protect\citeauthoryear{Reps}{Reps}{1991}]{reps91}
{\sc Reps, T.} 1991.
\newblock Algebraic properties of program integration.
\newblock {\em Science of Computer Programming\/}~{\em 17}, 139--215.

\bibitem[\protect\citeauthoryear{Reps and Yang}{Reps and Yang}{1989}]{RY88}
{\sc Reps, T.} {\sc and} {\sc Yang, W.} 1989.
\newblock The semantics of program slicing and program integration.
\newblock In {\em Proc.\ of the Colloq.\ on Current Issues in Programming
  Languages}, {J.~Diaz} {and} {F.~Orejas}, Eds. Lecture Notes in Computer
  Science Series, vol. 352. Springer-Verlag, Berlin, 360--374.

\bibitem[\protect\citeauthoryear{Rossignoli and Spoto}{Rossignoli and
  Spoto}{2006}]{DBLP:conf/vmcai/RossignoliS06}
{\sc Rossignoli, S.} {\sc and} {\sc Spoto, F.} 2006.
\newblock Detecting non-cyclicity by abstract compilation into boolean
  functions.
\newblock In {\em Proceedings of the International Conference on Verification,
  Model Checking, and Abstract Interpretation (VMCAI)}. Lecture Notes in
  Computer Science Series, vol. 3855. Springer-Verlag, 95--110.

\bibitem[\protect\citeauthoryear{Secci and Spoto}{Secci and
  Spoto}{2005}]{DBLP:conf/sas/SecciS05}
{\sc Secci, S.} {\sc and} {\sc Spoto, F.} 2005.
\newblock Pair-sharing analysis of object-oriented programs.
\newblock In {\em Proceedings of the Interanational Symposium on Static
  Analysis (SAS)}. Springer-Verlag, 320--335.

\bibitem[\protect\citeauthoryear{Tip}{Tip}{1995}]{Tip95}
{\sc Tip, F.} 1995.
\newblock A survey of program slicing techniques.
\newblock {\em Journal of Programming Languages\/}~{\em 3}, 121--181.

\bibitem[\protect\citeauthoryear{Ward and Zedan}{Ward and
  Zedan}{2007}]{WardZedan}
{\sc Ward, M.} {\sc and} {\sc Zedan, H.} 2007.
\newblock Slicing as a program transformation.
\newblock {\em ACM Transactions on Programming Languages and Systems\/}~{\em
  29,\/}~2.

\bibitem[\protect\citeauthoryear{Weiser}{Weiser}{1984}]{weiser}
{\sc Weiser, M.} 1984.
\newblock Program slicing.
\newblock {\em IEEE Trans.\ on Software Engineering\/}~{\em 10,\/}~4, 352--357.

\bibitem[\protect\citeauthoryear{Zanardini}{Zanardini}{2008}]{Zanardini}
{\sc Zanardini, D.} 2008.
\newblock The {S}emantics of {A}bstract {P}rogram {S}licing.
\newblock In {\em Proceeding of Working Conference on Source Code Analysis and
  Manipualtion (SCAM)}.

\bibitem[\protect\citeauthoryear{Zanardini}{Zanardini}{2015}]{ZanardiniG15sh}
{\sc Zanardini, D.} 2015.
\newblock {F}ield-{S}ensitive {S}haring.
\newblock {\em CoRR\/}~{\em abs/1306.6526}.

\end{thebibliography}

\section{APPENDIX: THE FORMAL SLICING FRAMEWORK}
 % programSlicing-frameworkN2
In this section, we first provide a better intuition of the differences between the forms of slicing introduced in Section~\ref{section:Background} by means of examples and then we recall the main notions introduced in \cite{AForm,TheoFoun} that has been generalized in this paper in the abstract form.

\subsection*{Different forms of slicing: some examples}
\begin{example}\label{Ex:crit}
  Consider the program on the left in Figure~\ref{Exfig:crit}.
  Suppose that the execution start with an initial value $2$ for
  \CODE{n} (written $\mbox{\nn}\la 2$).  States are denoted as
  $(m^k,\memory)$, where $m$ is the program point of the executed statement, $k$ is its
  current iteration (i.e., the statement at $m$ is being executed for the $k$-th time in the loop unrolling),
  and $\memory$ is the actual memory, represented by a list of
  pairings $x\la v$). In the picture, $m^k$ is depicted in the first (fully colored) box, while the memory is depicted in the remaining boxes, one for each variable. A program state that is not executed in a trace is depicted by overwriting a grey cross on each box (program point and variables).
The execution trajectory of the program is the following:
\begin{center}
\includegraphics[scale=.4]{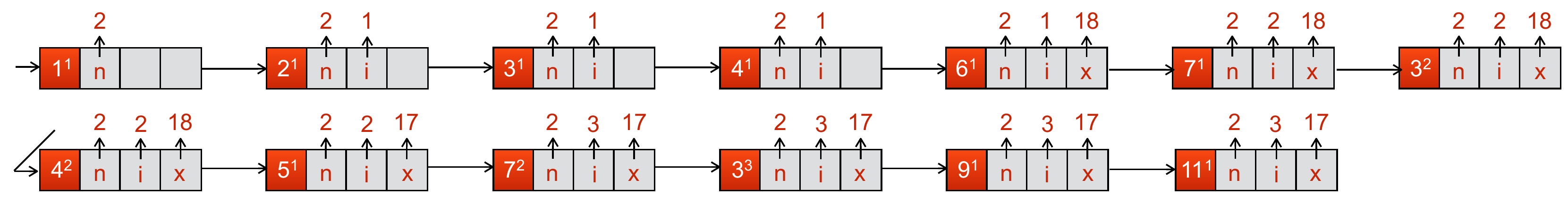}
\end{center}
  
\noindent  
  Consider now the code in the center, whose execution trajectory is the following:
\begin{center}
\includegraphics[scale=.4]{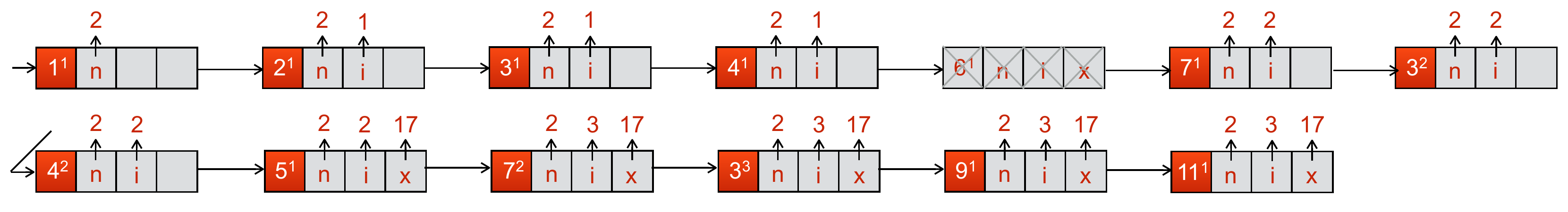}
\end{center}
  We can observe that the semantics of the second program follows
  precisely the same path on the statements which are in both programs
  (the only difference is the execution, in the original program of the statement at point $6$, erased in the "candidate" slice), hence it is a slice according both to the standard and the $\KL$
  form, namely w.r.t.\ $\crit=(\{\mbox{\nn}\la 2\},\{\xx\},\{\tuple{11,\NATURALS}\},\psi)$ for both the possible values of $\psi$.  \\
Suppose now we are interested in an $\IC$ form of slice, variable \xx at the second iteration of the  program point $3$. In this case, the program on the right is not a dynamic slice, since the value of
  \xx in the original program is $18$, while in the candidate slice it is
  undefined. In other words, this program is not slice of the program on the left w.r.t.\ the criterion $\crit=(\{\mbox{\nn}\la 2\},\{\xx\},\{\tuple{3,\{2\}}\},\psi)$ ($\psi$ may be both true or false). \\ 
  Finally, let us consider the execution of the program on the right in Fugure~\ref{Exfig:crit}:
  \begin{center}
\includegraphics[scale=.4]{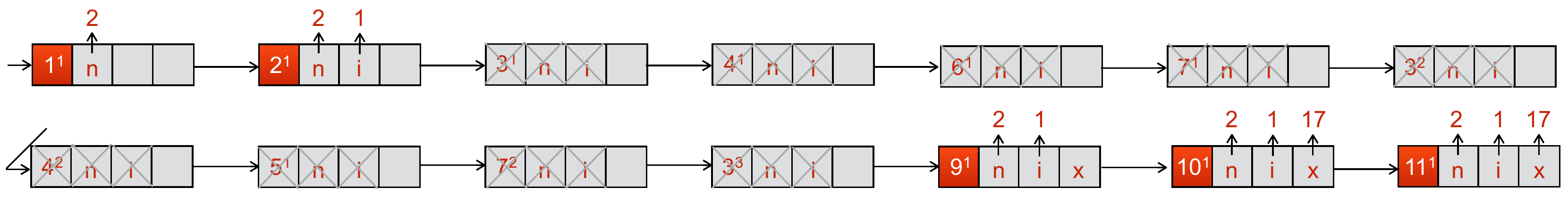}
\end{center}
 This last program is a standard dynamic slice since the
  final value of the variable of interest \xx is the same, but it is
  not a slice in the $\KL$ form, since in this last program the statement at program point $10$ is executed, while in the original program it is not executed. Namely, it is a slice w.r.t.\ the criterion $\crit=(\{\mbox{\nn}\la 2\},\{\xx\},\{\tuple{11,\NATURALS}\},\psi)$ only for $\psi=\false$.  
  \begin{figure}
    \begin{tabular}{c|c|c}
      \begin{lstlisting}
  read(n);
  i := 1;
  while (i <= n) do {
    if (i mod 2 = 0) {
      x := 17; }
    else { x := 18; }
    i := i + 1;       
  }
  if (i = 1) {
    x = 17; }
  write(i,n,x);
      \end{lstlisting}
      &
      \begin{lstlisting}
  read(n);
  i := 1;
  while (i <= n) {
    if (i mod 2 = 0) {
      x := 17; }
    
    i := i + 1;       
  }
  if (i = 1) {
    x = 17; }
  write(i,n,x);
      \end{lstlisting}
      &
      \begin{lstlisting}
  read(n);
  i := 1;
    
  if (i = 1) { 
    x = 17; }
  write(i,n,x);
      \end{lstlisting}
    \end{tabular}
    \caption{Programs of Example~\ref{Ex:crit}}\label{Exfig:crit}
  \end{figure}
\end{example}

The following example shows the difference between standard and $\KL$
forms of {\em static} slicing.

\begin{example}\label{ex:stkl}
  Consider the program $\prog$ on the left of Figure~\ref{fig:stkl}.
  Suppose static slicing is considered, i.e., all the possible initial
  memories are taken into account.  Given an input $v$, the state
  trajectory is:
   \begin{center}
\includegraphics[scale=.4]{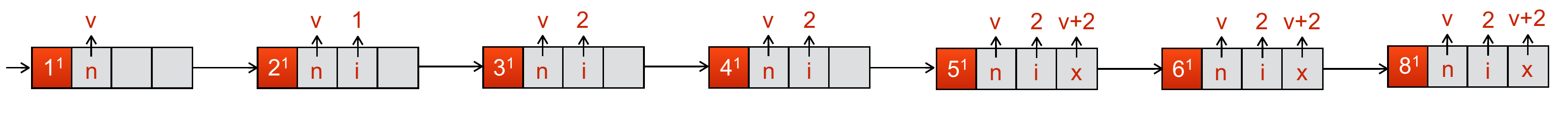}
\end{center}
  Consider now the code on the right: its state trajectory is
  \begin{center}
\includegraphics[scale=.4]{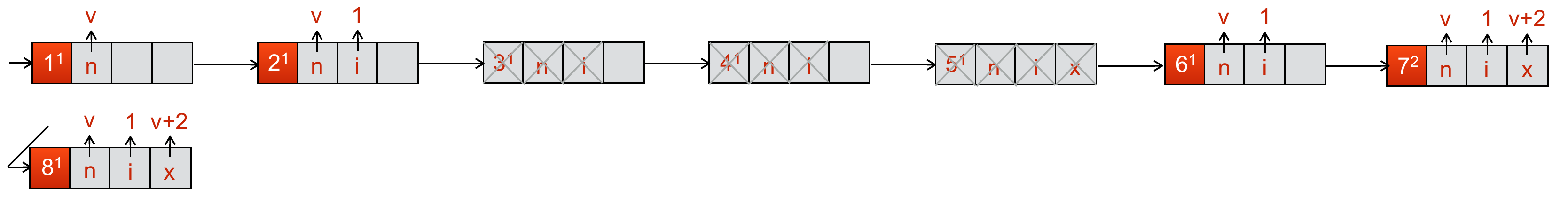}
\end{center}
  \begin{figure}
    \begin{center}
      \begin{tabular}{c|c}
        \begin{lstlisting}
  read(n);
    i := 1;
    i := 2;
    if (i mod 2 = 0) {
       x := i + n;  }
    if (i mod 2 = 1) {
       x := i + n + 1; }
  write(i,n,x);
        \end{lstlisting}
        &
        \begin{lstlisting}
  read(n);
    i := 1;
    
    if (i mod 2 = 1) {
       x := i + n + 1; }
  write(i,n,x);
        \end{lstlisting}
      \end{tabular}
    \end{center}
    \caption{Programs of Example~\ref{ex:stkl}}\label{fig:stkl}
  \end{figure}

  \noindent Consider the standard form of static slicing
  interested in \xx at program point $8^1$, i.e., $\crit=(\{\mbox{\nn}\la \NATURALS\},\{\xx\},\{\tuple{8,\{1\}}\},\false)$.  Then the program
  on the right is a slice of $\prog$ w.r.t.\ $\crit$, since in $8^1$ the value of \xx
  is $v+2$ in both cases.  On the other hand, if we consider the $\KL$
  form, $\crit=(\{\mbox{\nn}\la \NATURALS\},\{\xx\},\{\tuple{8,\{1\}}\},\true)$, then the program is no more a slice
  of $\prog$ since there is a program point, $7^1$, which is not
  reached in $\prog$.
\end{example}

\subsection*{The unified equivalence}

The first step for defining the formal framework is to define an
equivalence relation between programs, determining when a program is a
slice of another.  First, a \emph{restricted memory} is obtained from
a memory by restricting its domain to a set of variables.  More
formally, the restriction of $\memory$ with respect to a set of
variables $\cX$ is defined as $\MEM{\memory}{\cX}$ such that
$(\MEM{\memory}{\cX})(x)$ is equal to $\memory(x)$ if $x\in \cX$, and
undefined otherwise.  This restriction is used to project the trace
semantics only on those points of interest where we have to check the
correspondence between the original program and the candidate slice.

The \emph{trajectory projection} operator modifies a state trajectory
by removing all those states which do not contain occurrences of
program points which are relevant for the slicing criterion.

\begin{mydefinition}[Trajectory Projection \cite{AForm}]
  \label{def:Proj}
  Let $\crit=(\cI,\cX,\cO,\psi)$, and $\cL\subseteq\lnums$ such that $\cL\neq\emptyset$ if $\psi=\true$, $\cL=\emptyset$ otherwise. For any $n \in \lnums$, $k \in \NATURALS$, $\memory \in \memories$, we define the function $\Proj_0$ as:
  \[ \Proj^0_{(\cX,\cO,\cL)}(n^k,\memory) \defi \left\{
  \begin{array}{ll}
    \tuple{n^k, \MEM{\memory}{\cX}} & \mbox{if }(\exists \tuple{n,K} \in \cO.\:k \in K \\ \tuple{n^k,\bot}
      & \mbox{if }(\nexists \tuple{n,K} \in \cO.\:k \in K)\mbox{ and }n\in\cL \\ \varepsilon & \mbox{otherwise}.
  \end{array}
  \right.
  \]
  where $\varepsilon$ is the empty sequence.  The trace projection $\Proj$
  is the extension of $\Proj^0$ to sequences ($\circ$ is sequence
  concatenation):
  \[ \begin{array}{rl}
    \Proj_{(\cX,\cO,\cL)} (\tuple{(n^{k_1}_1, \memory_1), \ldots, 
      (n^{k_l}_l, \memory_l)})  =
     \Proj^0_{(\cX,\cO,\cL)}(n^{k_1}_1, \memory_1) \circ
      \ldots \circ \Proj^0_{(\cX,\cO,\cL)}(n^{k_l}_l, \memory_l)
  \end{array}
  \]
\end{mydefinition}

\noindent $\Proj^0$ takes a state from a state trajectory, and returns
either one pair or an empty sequence $\varepsilon$.  If
$n^k$ is an occurrence of interest, then it returns
$\tuple{(n^k, \MEM{\memory}{\cX})}$.  This means that, at $n$, we
consider exact values of variables in $\cX$.  If $n^k$ is not an
occurrence of interest, but, due
to a $\KL$ form, the projection has to keep trace of a set $\cL$ of executed statements (even if the variables in that point are not of interest), then $\Proj^0$ returns
$\tuple{(n^k, \MEM{\memory}{\emptyset})}$, meaning that we require
the execution of $n$, but we are not interested in the values of
variables in $\cX$.

Trajectory projection allows us to define all the semantic equivalence
relations characterizing on what a program and its slices have to agree due to the chosen criterion..  Given two programs $\prog$ and
$\progq$, we can say that $\progq$ is a slice of $\prog$ if it contains a subset of the original
statements and $\progq$ is \emph{equivalent} to $\prog$ with respect
to the semantic equivalence relation induced by chosen the slicing criterion.

\begin{mydefinition}[\cite{AForm}]
  \label{def:UnifiedEquivalence} Let $\prog$ and $\progq$ be
  executable programs, and $\crit=(\cI,\cX,\cO,\psi)$ be a slicing
  criterion.  Let $\lnums_\prog\subseteq\lnums$ be the set of program
  points of $\prog$, and $\cL=\lnums_\prog\cap\lnums_\progq$ if $\psi=\true$\footnote{Note that,
  when $\lnums_\progq \subseteq \lnums_\prog$, as we suppose in this paper, then $\cL=\lnums_\progq$.  We provide the general definition since
  the original definition of dynamic slicing \cite{KorelLaski}
  does not require that all the line of $\progq$ are included in
  $\cL$; however, our choice follows the paths taken in the original
  framework \cite{TheoFoun}.}
  ($\cL = \emptyset$ if $\psi=\false$).  Then $\prog$
  is \emph{equivalent} to $\progq$ w.r.t.\ $\crit$ if and only
  if \[ \forall\memory\in \cI.~~ \Proj_{(\cX, \cO, \cL)}(\trace_\prog^\memory)
  = \Proj_{(\cX,\cO,\cL)}(\trace_\progq^\memory)\] The function $\cE$ maps any criterion $\crit$ to the r
 to the corresponding semantic
  equivalence relation, hence, in this case, we write $\tuple{\prog,\progq}\in \cE(\crit)$.
\end{mydefinition}

\begin{example}
  \label{ex:formSl}
  Consider the Program $\prog$ in the left of
  Figure~\ref{fig:exsl}; let the input for \nn be $2$.
 \noindent
  Suppose we want to compute a non iteration-count $\KL$ form of
  dynamic slicing, i.e., $\crit=(\{\mbox{\nn}\la \{2\}\},\{\mbox{\ii},\mbox{\ss}\},\tuple{8,\NATURALS},\true)$. Namely, the variables of interest are
  \ii and \ss, which are observed at the program point $8$ each time it is reached, and
   the slice has not to execute statements not executed in the original program.  The program on the right of Figure~\ref{ex:formSl} is a slice w.r.t.\ $\crit$.
    In Figure~\ref{ExTrace} we have the execution trajectory of the original program (on the top), the execution trajectory of the candidate slice (in the middle) and the (same) projection of the two trajectories due to the chosen criterion (on the bottom).
  \begin{figure}
    \begin{center}
      \includegraphics[scale=.4]{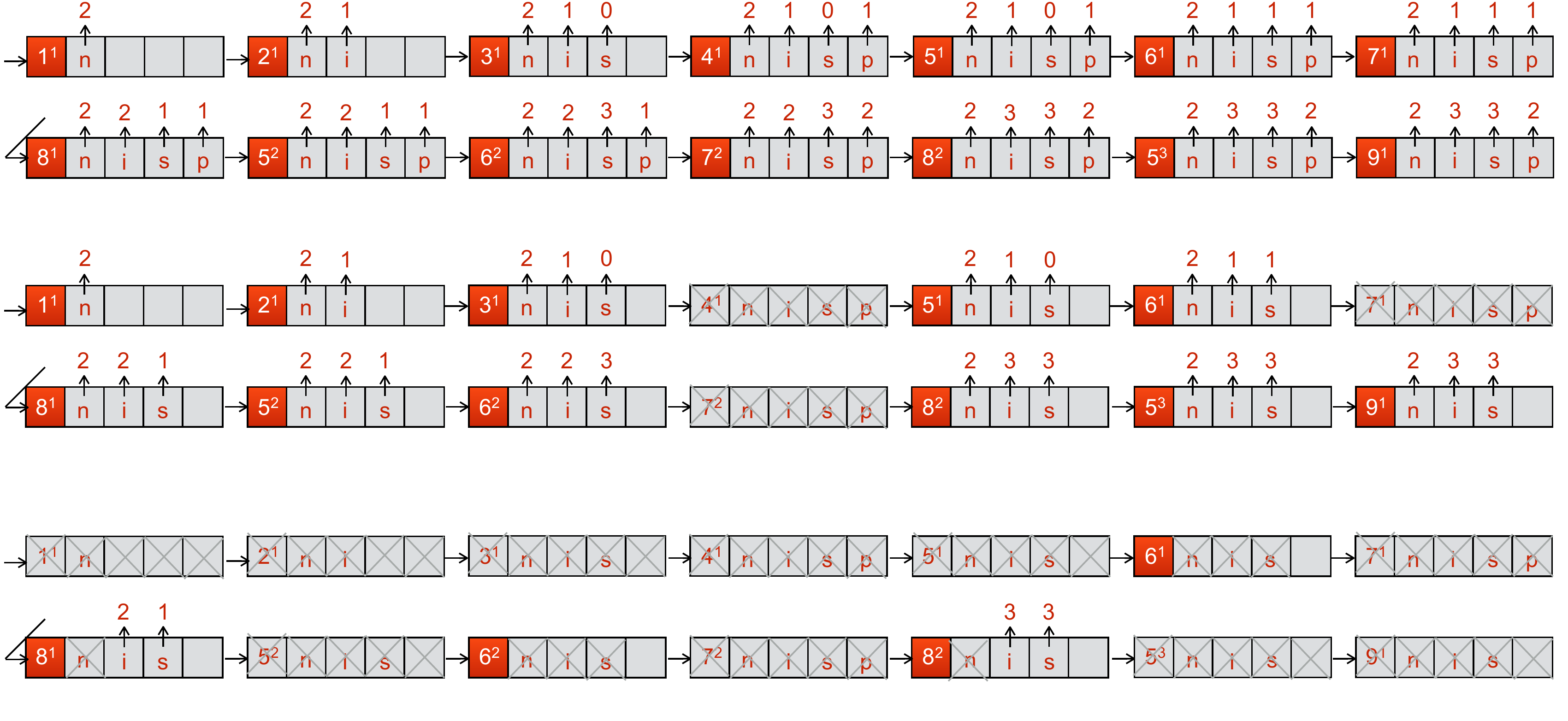}
    \end{center}
    \caption{Execution trace of $P$ in Example~\ref{ex:formSl} and of the static slice and of their projection.}\label{ExTrace}
  \end{figure}
\end{example}

The formal framework proposed in \cite{AForm,TheoFoun} represents
different forms of slicing by means of a $(\sqsubseteq,\cE)$ pair: a
\emph{syntactic preorder}, and a \emph{function from slicing criteria
  to semantic equivalences}.  The \emph{preorder} fixes a syntactic
relation between the program and its slices.  In traditional slicing,
 slices are obtained from the
original program by removing zero or more statements.  This preorder
is called \emph{traditional syntactic ordering}, simply denoted by
$\sqsubseteq$, and it is defined as follows: $\progq \sqsubseteq \prog
\Leftrightarrow \lnums_\progq \subseteq \lnums_\prog$
The second component $\cE$ fixes the semantic constraints that a
subprogram has to respect in order to be a slice of the original
program.  As we have seen before, the equivalence relation is uniquely
determined by the chosen slicing criterion determining also a specific form of slicing.
This way, Binkley \etal are able to characterize eight forms of
non-$\SIM$ slicing, and twelve forms of $\SIM$ slicing.

Finally, this framework is used to formally compare the different notions of
slicing.  First of all, it is defined a binary relation on slicing criteria
$\rightarrow$ \cite{AForm}: Let $\crit^1=\tuple{\cI^1,\cX^1,\cO^1,\psi^1}$ and $\crit^2=\tuple{\cI^2,\cX^2,\cO^2,\psi^2}$
\begin{equation}\label{eq:defrelcrit}
\crit^1 \rightarrow \crit^2\qquad \mbox{iff}\ 
\qquad \cI^1\subseteq\cI^2,\ \cX^1\subseteq\cX^2, \cO^1\subseteq\cO^2, \psi^1\Ra\psi^2
\end{equation}
At this point, we say that a form of
slicing $(\sqsubseteq, \cE^1)$ is \emph{weaker than} $(\sqsubseteq,
\cE^2)$ w.r.t. \small$\rightarrow$ \normalsize iff $\forall \crit^1,
\crit^2$, slicing criteria such that $\crit^1 \rightarrow \crit^2$,
and $\forall \prog, \progq$, if $\progq$ is a slice of $\prog$
w.r.t.\ $(\sqsubseteq, \mathcal{E}^2(\crit^2))$, then $\progq$ is a
slice of $\prog$ w.r.t.\ $(\sqsubseteq, \mathcal{E}^1(\crit^1))$ as
well. In this case we say that $(\sqsubseteq, \mathcal{E}^1)$ subsumes $(\sqsubseteq, \mathcal{E}^2)$.
\begin{figure}[tbp]
  \centering
  \includegraphics[scale=.8,viewport=4in 6in 5.6in 9.2in]{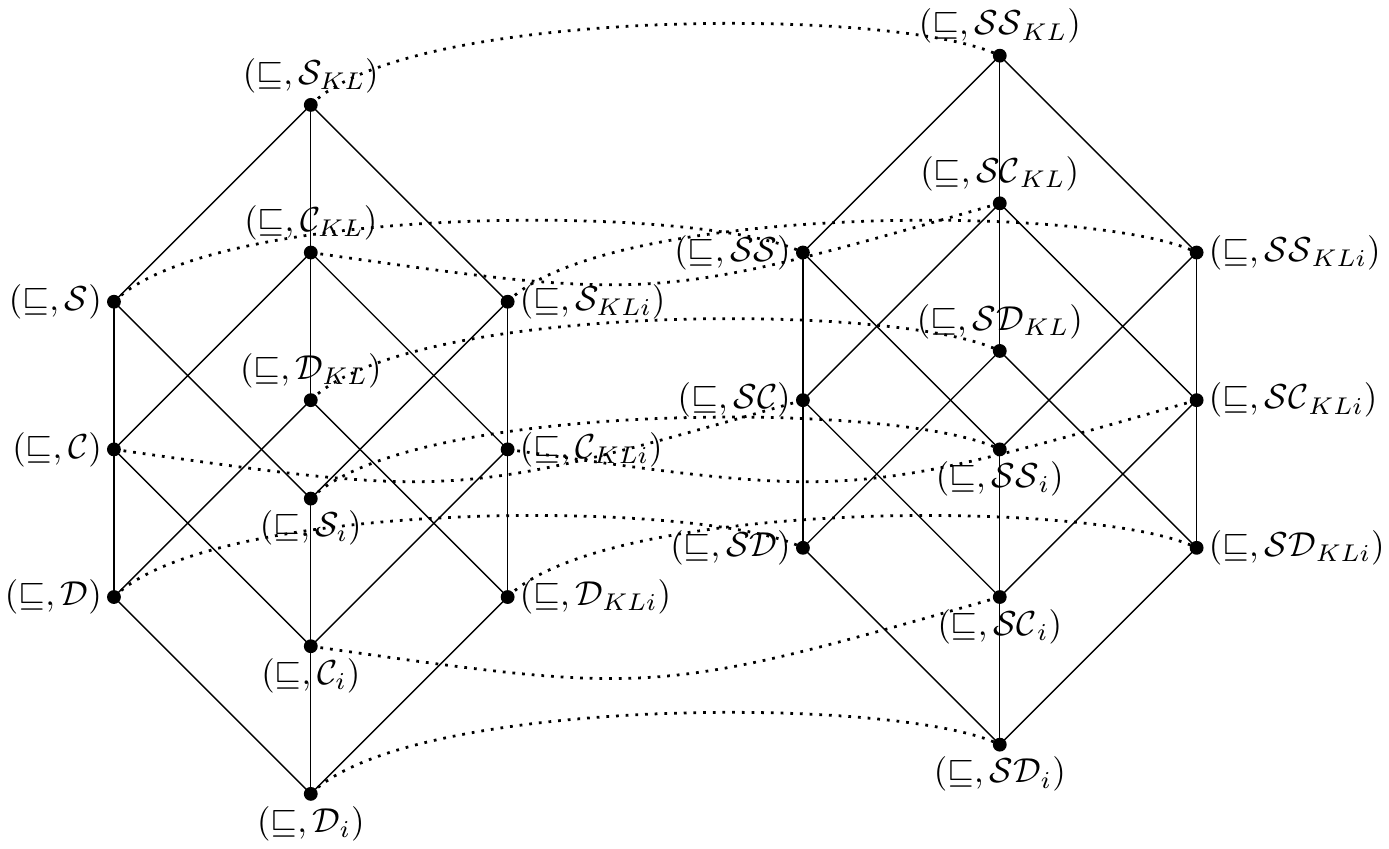}
  \caption{Given two forms $A$ and $B$, both (non-)$\SIM$., $A$ is
    weaker than $B$ if $A$ is connected to $B$ by a solid line and it
    is below $B$. If $A$ is non-$\SIM$. and $B$ is $\SIM$., then $A$
    is weaker than $B$ if $A$ is connected to $B$ by a dotted line and
    it is to the left of $B$.\label{fig:Ret1}}
\end{figure}
Following this definition, Binkley et al.~show that all forms of
slicing introduced in~\cite{AForm} are comparable in the way shown in
Figure~\ref{fig:Ret1}, where the symbols $\cS$, $\cC$, $\cD$,
$\cS\cS$, $\cS\cC$ and $\cS\cD$ represent static, conditioned,
dynamic, static $\SIM$, conditioned $\SIM$ and dynamic $\SIM$ types of
slicing, respectively.  Subscripts $i$, $\KL$ and $\KLi$ represent
$\IC$, $\KL$ and $\KLi$ forms of slicing, respectively; the absence of
subscripts denotes the standard forms of slicing.  In
Figure~\ref{fig:Ret1} we explicitly provide both the hierarchy
concerning $\SIM$ and non-$\SIM$ forms of conditioned slicing
constructed in~\cite{AForm}.

\label{appendix}

\end{document}